\documentclass[acmsmall]{acmart}
\AtBeginDocument{%
  }

\setcopyright{cc}
\setcctype{by}
\acmJournal{PACMMOD}
\acmYear{2026} \acmVolume{4} \acmNumber{5 (PODS)} \acmArticle{304}
\acmMonth{11} \acmDOI{10.1145/3850139}


\usepackage{enumitem}
\usepackage[ruled,lined,linesnumbered,noend]{algorithm2e}
\usepackage{url}
\usepackage{stmaryrd}
\usepackage{mathtools}

\usepackage[thmtools-compat]{keytheorems}
\usepackage{hyperref}

\newcommand{\doc}{\mathsf{D}}
\newcommand{\ta}{\mathtt{a}}
\newcommand{\tb}{\mathtt{b}}
\newcommand{\tc}{\mathtt{c}}

\newcommand{\eword}{\varepsilon}
\newcommand{\query}{u_{\mathcal C}}
\newcommand{\longright}{\mathsf{longestRight}}
\newcommand{\shortleft}{\mathsf{shortestLeft}}
\newcommand{\LRArray}{\mathsf{LR}}
\newcommand{\SLArray}{\mathsf{SL}}

\newcommand{\LCON}{\mathsf{LCON}}
\newcommand{\RCON}{\mathsf{RCON}}
\newcommand{\REG}{\mathsf{REG}}

\newcommand{\LL}{\mathcal{L}}
\newcommand{\embedLCONSubseq}{\mathsf{EmbedLCONSubseq}}

\newcommand{\nextMove}{\mathsf{nextMove}}
\newcommand{\incr}{\mathsf{incr}}

\newcommand{\embedAll}{\mathsf{EmbedLCONAll}}
\newcommand{\nextMoveLRCON}{\nextMove\mathsf{LRCON}}
\newcommand{\CSP}{\mathsf{CSP}}
\newcommand{\Pol}{\mathsf{Pol}}
\newcommand{\const}{K}
\newcommand{\domain}{T}

\let\oldnl\nl%
\newcommand{\nonl}{\renewcommand{\nl}{\let\nl\oldnl}}%

\theoremstyle{plain}
\newtheorem{theorem}{Theorem}[section]
\newtheorem{proposition}[theorem]{Proposition}
\newtheorem{lemma}[theorem]{Lemma}

\newtheorem{claim}[theorem]{Claim}
\newtheorem{definition}[theorem]{Definition}

\theoremstyle{remark}

\newtheorem{observation}[theorem]{Observation}

\RenewCommandCopy{\theHtheorem}{\thetheorem} %
\begin{document}

\received{30 May 2026}
\received[revised]{25 August 2026}
\received[accepted]{1 September 2026}

\title{Tractable Gap-Constraint Languages for Complex Event Recognition}

\author{Antoine Amarilli}
\email{a3nm@a3nm.net}
\orcid{0000-0002-7977-4441}
\affiliation{
  \institution{Univ. Lille, Inria, CNRS, Centrale Lille, UMR 9189 CRIStAL}
  \city{F-59000 Lille}
  \country{France}
}
\author{Florin Manea}
\email{florin.manea@cs.uni-goettingen.de}
\orcid{0000-0001-6094-3324}
\affiliation{
  \institution{University of G\"ottingen, Institute for Computer Science and CIDAS}
  \city{D-37077 G\"ottingen}
  \country{Germany}
}
\author{Tina Ringleb}
\email{tina.ringleb@cs.uni-goettingen.de}
\orcid{0009-0004-8928-2928}
\affiliation{
  \institution{University of G\"ottingen, Institute for Computer Science}
  \city{D-37077 G\"ottingen}
  \country{Germany}
}
\author{Markus L. Schmid}
\email{MLSchmid@MLSchmid.de}
\orcid{0000-0001-5137-1504}
\affiliation{
  \institution{Humboldt-Universit\"at zu Berlin}
  \city{Unter den Linden 6, D-10099, Berlin}
  \country{Germany}
}

\begin{abstract}
For strings $u, \doc \in \Sigma^*$, a \emph{subsequence embedding} of $u$ in $\doc$ is a function $e \colon \{1, 2, \ldots, |u|\} \to \{1, 2, \ldots, |\doc|\}$ with $e(i) < e(i+1)$ for every $i \in \{1, 2, \ldots, |u|-1\}$ and the $i^{\text{th}}$ symbol of $u$ equals the $e(i)^{\text{th}}$ symbol of $\doc$. A \emph{gap-constraint} for $u$ is a triple $(i, j, L)$ with $1 \leq i < j \leq |u|$ and $L$ is a regular language over $\Sigma$. An embedding $e$ \emph{satisfies} a gap-constraint $(i, j, L)$ if the factor of $\doc$ strictly between positions $e(i)$ and $e(j)$ is a word from $L$. We investigate the subsequence matching problem with gap-constraints, which is relevant in the context of complex event recognition (CER): given $u, \doc \in \Sigma^*$ and a set $\mathcal{C}$ of gap-constraints, find an embedding of $u$ in $\doc$ that satisfies all gap-constraints from $\mathcal{C}$. 
In general, subsequence matching 
with gap constraints
is NP-complete and the only known tractable variants restrict the interval structure of the gap-constraints. In this work, we show that we can solve subsequence matching with gap-constraints with an arbitrary interval structure rather efficiently (in fact, optimally under SETH) in time $O(|\doc| (|u| + \lVert \mathcal C\rVert))$ if the gap-constraint languages satisfy a property which we dub \emph{left-convexity}: whenever $u v w \in L$ and $v \in L$, then also $uv \in L$. Left-convex languages are sufficiently expressive to model interesting real-world scenarios considered in CER, e.g., length constraints $L = \{w \mid a \leq |w| \leq b\}$ for $a, b \in \mathbb{N}$. 
We also show how our algorithm can be used in order to efficiently enumerate all satisfying embeddings, which is particularly relevant for possible applications in CER. Finally, we show how \emph{non}-left-convex languages can lead to intractability, i.e., if in addition to length constraints we allow $\{\ta \ta, \eword\}$ as the only non-left-convex constraint language, then the problem is NP-complete again. 
\end{abstract}

\maketitle

\section{Introduction}\label{sec:intro}

In \emph{complex event recognition} (CER), we receive a long sequence of events generated by the run of a large system (e.g., a distributed computer system with many interacting components, or a surveillance system that continuously produces measurements); see the introductions of the surveys~\cite{ArtikisEtAl2017, GiatrakosEtAl2020}. The task is to process the received event sequence and search for certain patterns (i.e., strings of events), which are called \emph{complex events}. This explains why several query languages for CER are rooted in pattern matching and automata theory (see~\cite{AgrawalDGI08, ZhangEtAl2014, GrezEtAl2021}).
However, CER deviates from the classical string pattern matching paradigm because pattern occurrences are formalised as subsequences rather than consecutive factors, and there are additional constraints on how the subsequences are embedded (e.g., by asking for two events of type A that are \emph{not} separated by an event of type B),
which often amounts to regular constraints on some gaps in the embeddings.

A reasonable mathematical abstraction of (a core functionality of) CER is therefore \emph{subsequence matching with (regular) gap-constraints}.
This task is defined as searching for a suitable \emph{embedding} of a string $u \in \Sigma^*$ (over some finite alphabet $\Sigma$) in a larger document $\doc \in \Sigma^*$, namely, a total function $e \colon \{1, 2, \ldots, |u|\} \to \{1, 2, \ldots, |\doc|\}$ such that $e(i) < e(i + 1)$ for every $i \in \{1, 2, \ldots, |u| - 1\}$ and $u[i] = \doc[e(i)]$ for every $i \in \{1, 2, \ldots, |u|\}$ (i.e., the $i$-th symbol of $u$ equals the $e(i)$-th symbol of $\doc$). If $e$ exists, then $u$ is also called a (scattered) \emph{subsequence} of $\doc$. Finding an embedding of $u$ in $\doc$ is a trivial computational task that can be solved greedily in linear time. However, the situation becomes more interesting if we require $e$ to satisfy some additional constraints, and motivated by the CER setting we are concerned with so-called \emph{gap-constraints}. For every $i, j \in \{1, 2, \ldots, |u|\}$ with $i < j$, the \emph{$(i, j)$-gap induced by $e$} is the factor $\doc[e(i) + 1 : e(j) - 1]$, i.e., the factor of $\doc$ strictly between the images of $i$ and $j$ under $e$. A \emph{gap-constraint} is a triple $(i, j, L)$ with $1 \leq i < j \leq |u|$ and $L \subseteq \Sigma^*$, and the embedding $e$ \emph{satisfies} $(i, j, L)$ if and only if the $(i, j)$-gap induced by $e$ is in $L$. If $L$ is a regular language, then $(i, j, L)$ is a \emph{regular} gap-constraint.

Figure~\ref{fig:embeddingExample} illustrates an example for $u = \ta \tb \ta \tc$ and $\doc = \tc \ta \ta \tb \tb \tb \ta \tb \ta \tc \tb \tc$, and three gap-constraints displayed at the right. In general, there are several possible embeddings of $u$ in $\doc$, but the specific one illustrated by the arrows in Figure~\ref{fig:embeddingExample} is one that satisfies all three gap-constraints. If we change $e(4)$ from $12$ to $10$, we still have an embedding of $u$ in $\doc$, but the $(1, 4)$-gap $\doc[3 : 9] = \ta \tb \tb \tb \ta \tb \ta$ has only $4$ occurrences of $\tb$ and therefore does not satisfy the $(1, 4)$-gap-constraint. Changing $e(2)$ from $5$ to $4$ violates the $(2, 4)$-gap-constraint, while changing $e(2)$ from $5$ to $6$ results in an embedding that still satisfies all gap-constraints.

\begin{figure}[t]
  \centering
\includegraphics[height=2.5cm]{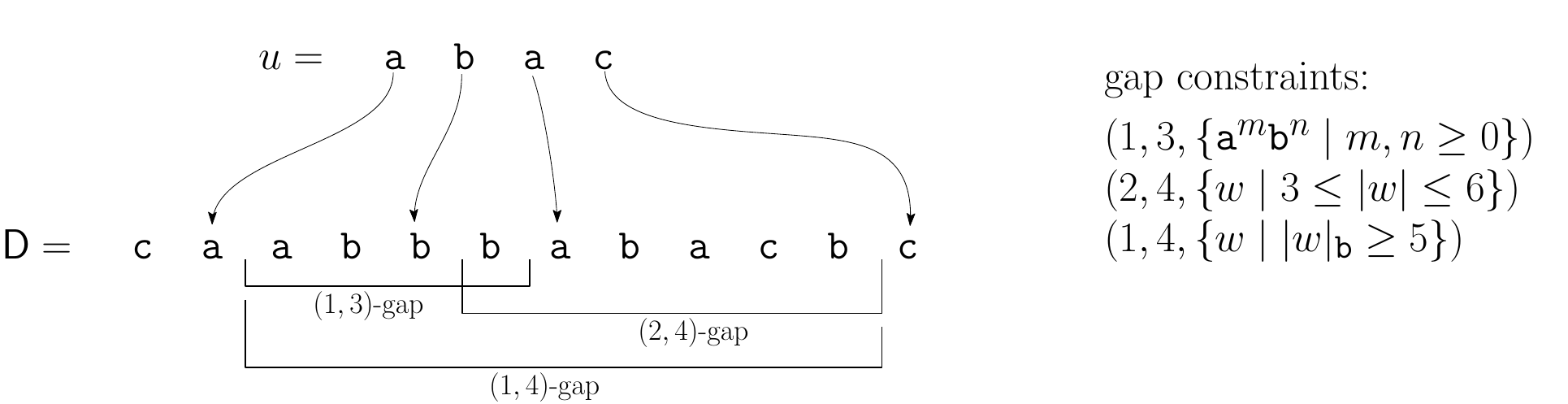}
\caption{Example of an embedding of $u$ in $\doc$ with gap-constraints (here, $|w|_x$ denotes the number of occurrences of letter $x$ in the string $w$).}
\label{fig:embeddingExample}
\end{figure}

Coming back to the literature on CER, we observe that, inspired by the SASE-language~\cite{WuDR06, AgrawalDGI08, ZhangEtAl2014}, the papers~\cite{KleestMeissnerEtAl2021,Kleest-MeissnerEtAl23, FrochauxKleestMeissner2023, SattlerKLSSW25, FrochauxKS25, SattlerKLSS025} introduce and investigate \emph{subsequence queries with gap-constraints}.\footnote{The queries investigated in these papers also allow variables in order to model equality constraints of the SASE-language. This is a functionality that goes beyond the setting of this paper and is also well-known to lead to intractability (see~\cite{KleestMeissnerEtAl2021}).} In this model, a string and a set of gap-constraints is interpreted as a query, and the corresponding result set contains all embeddings of the string into the event sequence such that all gap-constraints are satisfied. An important observation is that CER scenarios mostly consider very simple constraint languages -- both in practically motivated settings and in their theoretical abstractions -- with the most prominent example of a very simple yet practically relevant such class being so-called \emph{length constraints}, i.e., gap-constraints $(i, j, L)$ with $L = \{w \in \Sigma^* \mid a \leq |w| \leq b\}$ for some $a, b \in \mathbb{N} \cup \{0\}$. More precisely, the initial model of~\cite{KleestMeissnerEtAl2021} considers length constraints and only allows constraints between neighbouring positions. In~\cite{FrochauxKleestMeissner2023}, this setting has been extended to possibly overlapping gap-constraints between arbitrary positions (a feature also used in the example from Figure~\ref{fig:embeddingExample}), and the query language of the recent proposal~\cite{GarciaR25} can also formulate overlapping gap-constraints.

Motivated by CER, we are interested in the following tasks: Given $u, \doc \in \Sigma^*$ and a set $\mathcal{C}$ of regular gap-constraints, (1) compute a $\mathcal{C}$-embedding (i.e., one that satisfies the constraints of $\mathcal{C}$) of $u$ in $\doc$, or report that none exists; (2) enumerate all such
$\mathcal{C}$-embeddings with bounded delay.

As a first step, we observe that task (1) has already been investigated in the literature on string pattern matching. Particularly relevant 
to us 
is~\cite{DayKMS22}, where the authors investigate subsequence matching with arbitrary regular gap-constraints and with length constraints (note that although length constraints are a special case of regular constraints, they correspond to languages that can be concisely represented by two numbers written in binary,
and this can have an impact on the complexity). Moreover,~\cite{DayKMS22} only considers constraints between neighbouring positions (like~\cite{KleestMeissnerEtAl2021} mentioned above), while the general case of regular gap-constraints between arbitrary pairs of positions has been investigated in~\cite{manea2024generalisedgaps}. In this work~\cite{manea2024generalisedgaps}, it is also shown that task (1) is NP-complete in general, but it becomes tractable if the gap-constraints have a restricted nested structure, namely,
for any two gap-constraints $(i, j, L)$ and $(i', j', L)$, we have that $j \leq i'$ or $j' \leq i$ or $i \leq i' < j' \leq j$ or $i' \leq i < j \leq j'$. 

In terms of the CER setting, a drawback of the intractability result of~\cite{manea2024generalisedgaps} is that the case of \emph{simple} constraint languages is neglected. As possible examples of such languages we have already mentioned length constraints, but other simple language classes with conceivable practical relevance are easily found: alphabet restrictions, i.e., languages $\Gamma^*$ for $\Gamma \subsetneq \Sigma$; singleton languages; languages that upper and lower bound the number of occurrences of certain symbols; languages with unique start or end delimiters; languages defined by forbidden and required factors; etc. For such language classes, the hardness result of~\cite{manea2024generalisedgaps} does not apply. This motivates the study of the scenario where the \emph{structure} of the gap-constraints is unrestricted, but the constraint \emph{languages} are from some simple class of languages. Surprisingly, even the tractability of task (1) in the case of only length constraints (or in the case of only singleton languages) is still open.\footnote{These questions were 
posed at the Dagstuhl Seminar 24472 \emph{``Regular Expressions: Matching and Indexing''} \cite[Section~3.11]{GortzM0P24}.} Thus, our goal in this paper is to find classes of simple constraint languages for which subsequence matching with gap-constraints can be done efficiently. As usual in data management tasks, we are interested in algorithms that are linear in the size of the data (i.e., the document $\doc$), and instead of just solving decision problems (or computing a single witness), we are interested in efficiently computing 
the whole set of solutions by enumerating them with bounded delay
(i.e., task (2) mentioned above).

\subsection{Our Contribution}

As a first step, we observe that known results in the field of constraint satisfaction problems (CSP) can be used in a black-box fashion in order to obtain an efficient algorithm for restricted cases of our task (1). Namely, we call a set $\mathcal{C}$ of regular gap-constraints \emph{min-closed} for some $u, \doc \in \Sigma^*$ if for every two $\mathcal{C}$-embeddings $e, e' \colon \{1, \ldots, |u|\} \to \{1, \ldots, |\doc|\}$, their pointwise minimum $e_{\min}$ is also a \mbox{$\mathcal{C}$-embedding} (where we define $e_{\min}(x) \coloneq \min\{e(x), e'(x)\}$ for every $x \in \{1, \ldots, |u|\}$). For min-closed instances, we can show the following result, where $\lVert \mathcal C \rVert$ is the total size of all regular expressions that represent the regular gap-constraints.

\begin{theorem}\label{cspBlackBoxTheorem}
Restricted to min-closed instances, task (1) can be solved in time $O((|u| + \lVert\mathcal C \rVert) |\doc|^2)$.
\end{theorem}

It is not difficult to see that instances with only length constraints are necessarily min-closed, which already answers one of the open questions mentioned above. In general, for a given instance $u, \doc \in \Sigma^*$ and $\mathcal{C}$, we can check in polynomial time whether $\mathcal{C}$ is min-closed for $u$ and $\doc$. 

While this first result is a relevant observation, it seems to be of limited value in the context of CER, since the same set of gap-constraints might be min-closed for some pair $(u, \doc)$, but not for another pair $(u, \doc')$; thus, the min-closed property does not yield a reasonable query class. Instead, we are interested in finding language classes such that queries using constraints from these classes are necessarily tractable on every possible document $\doc$. Moreover, our main focus is not a complexity classification with respect to task (1), but finding efficient and natural algorithms for also solving task (2), for which the CSP meta-theorems do not seem to be helpful.

We identify a rather simple language property, called \emph{left-convexity}, which (as we will show) serves as a unifying explanation of the tractability of many cases: 

\begin{definition}\label{def:lcon}
A language $L\subseteq\Sigma^\ast$ is 
\emph{left-convex} if, for all $u,v,w\in\Sigma^\ast$, whenever 
$uvw\in L$ and $v\in L$, then 
$uv\in L$. A gap-constraint $(i, j ,L)$ is \emph{left-convex} if $L$ is a left-convex regular language. 
\end{definition}

Intuitively speaking, left-convexity of $L$ means that if a factor in a word from $L$ is also from~$L$, then the whole prefix including this factor must be in $L$. For example, length constraints are left-convex, and so are all the languages of the example given in Figure~\ref{fig:embeddingExample}. The name ``left-convex'' is inspired by other notions of convexity for formal languages as investigated in~\cite{AngB08, AngB09, BrzozowskiSX11, Thierrin72}, but to the best of our knowledge the specific left-convex property defined above has not been considered earlier in the literature. The symmetric notion of \emph{right-convex} is analogous and all our results for left-convex gap-constraints also hold for right-convex gap-constraints.

The concept of left-convexity covers many gap-constraint languages that are expressive enough to describe reasonable tasks in complex event recognition. The length constraints form the most prominent example, but there are many more interesting classes, e.g., languages defined by a forbidden set of infixes and a required set of infixes, or length constraints that are allowed to ignore certain letters. Especially the latter seems relevant for complex event recognition tasks, since upper and lower bounding the number of certain events between two specific events seems more reasonable than upper and lower bounding the total length of the factor in between. See Appendix~\ref{sec:leftConvexProp} for more information on the left-convex property (including language theoretical considerations and more examples of left-convex classes).

Algorithmically speaking, the point of left-convexity is the following: we can show that whenever $\mathcal{C}$ only contains left-convex constraint languages then $\mathcal{C}$ is necessarily min-closed for every $u, \doc \in \Sigma^*$. Hence, as a direct consequence of Theorem~\ref{cspBlackBoxTheorem}, if $\mathcal{C}$ only contains left-convex constraint languages, then we can solve task $(1)$ in time $O(( |u| + \lVert \mathcal C \rVert ) |\doc|^2)$.

For data management, the quadratic dependency on the data size $|\doc|$ is
unrealistic, 
so we 
give an algorithm that makes it linear, with the following structure.
We maintain a candidate mapping $e \colon \{1, \ldots, |u|\} \to \{1, \ldots, |\doc|\}$, which is initialised with $e(i) = i$ for every $i \in \{1, \ldots, |u|\}$. As long as $e$ is not a $\mathcal{C}$-embedding of $u$ in $\doc$ that satisfies all constraints from $\mathcal{C}$, we move some position $i$ to the right, i.e., we set $e(i) \coloneq e(i) + c$, until we have found such a satisfying embedding. 
This requires us to efficiently answer queries of the form $\doc[e(i) + 1 : e(j)
- 1] \in L$, which we show can be done thanks to left-convexity by computing,
for every $k\in [|\doc|]$ and language $L$ with $(i,j,L) \in \mathcal C$, the
size of a longest factor of $\doc$ starting at position $k+1$ that is contained in
$L$.

In the following informal statements of our main results, we let $u, \doc \in \Sigma^*$ and $\mathcal{C}$ be a set of regular left-convex gap-constraints for $u$. We then have:
\begin{theorem}\label{thm:mainalg}
We can compute a $\mathcal{C}$-embedding of $u$ in $\doc$, or report that no such embedding exists, in time $O(|\doc| (|u| + \lVert \mathcal C\rVert))$.
\end{theorem}

This algorithm is optimal, assuming the strong exponential time hypothesis SETH (Section~\ref{sec:lowerBound}).

A special property of our algorithm is that it can be called with an initial mapping $e_0$ and then produces (if possible) a $\mathcal{C}$-embedding $e$ of $u$ in $\doc$ with $e_0(i) \leq e(i)$ for every $i \in [|u|]$ that is minimal among all such embeddings. By carefully using this property, we are able to devise an algorithm for task $(2)$. In the next theorem, let $\query(\doc)$ be the set of all $\mathcal{C}$-embeddings of $u$ in $\doc$.

\begin{theorem}
  \label{thm:enum}
  We can enumerate the set $\query(\doc)$ with preprocessing time $O(T)$ and delay $O(|u| \cdot T)$, where $T$ is the running time of the matching algorithm of Theorem~\ref{thm:mainalg} above. (Note that this entails a total computation of $\query(\doc)$ in time $O(|\query(\doc)| \cdot |u| \cdot T)$.)
\end{theorem}

In the special case where all constraint languages are both left- and right-convex (as, e.g., length constraints), we can even get rid of the factor $|u|$ in the above theorem.

Finally, we investigate whether the left-convex property is necessary for tractability by considering a smallest language that is not left-convex: $\{\ta \ta, \eword\}$. Note that every language with strictly smaller cardinality or smaller maximum word size is immediately left-convex. Our algorithm solves the setting with length constraints rather efficiently, but we can show that if we additionally allow $\{\ta \ta, \eword\}$ as a constraint language, then the problem is NP-complete again, even if all other constraints are length constraints. By modifying our hardness reduction, we can derive two more relevant intractability results: there is a fixed regular language $L$ such that subsequence matching is NP-complete if all gap-constraints use $L$ as their constraint language, and the case where all constraints are left- or right-convex is also intractable. These hardness results are shown in Section~\ref{sec:hardness}.

\subsection{Further Related Work}

The classical concept of subsequences is employed in different areas of computer science: formal languages and logics (e.g., piecewise testable languages~\cite{Simon72, simonPhD, KarandikarKS15, CSLKarandikarS, journals/lmcs/KarandikarS19, PraveenEtAl2024}, subword order and downward closures~\cite{HalfonSZ17, KuskeZ19, Kuske20, Zetzsche16}), combinatorics on words~\cite{RigoS15, FreydenbergerGK15, LeroyRS17a, Rigo19, Seki12, MateescuSY04, Salomaa05, SchnoebelenVeron2023}, modelling concurrency~\cite{Riddle1979a, Shaw1978, BussSoltys2014}, database theory (especially event stream processing~\cite{ArtikisEtAl2017, GiatrakosEtAl2020, ZhangEtAl2014, KleestMeissnerEtAl2021, Kleest-MeissnerEtAl23, FrochauxKleestMeissner2023}). Moreover, many classical algorithmic problems are based on subsequences, e.g., longest common subsequence~\cite{baeza1991searching} or shortest common supersequence~\cite{Maier:1978} (see \cite{AdamsonSDKKMS25, FleischmannKKMNSW23} and the survey~\cite{abs-2208-14722}, for recent results on string problems concerned with subsequences). The longest common subsequence problem has also recently received substantial attention in fine-grained complexity (see~\cite{DBLP:conf/fsttcs/BringmannC18, BringmannK18, AbboudEtAl2015, AbboudEtAl2014}).

In a different area, the notion of length constraints has also been studied in
the setting of so-called extended conjunctive regular path queries (ECRPQs) on graph
databases~\cite{barcelo2012expressive}, where we must match a graph pattern in a
graph database while obeying length requirements on the paths used to match the
pattern. However, because the data considered in such tasks is an arbitrary
graph rather than a string, the evaluation of ECRPQs 
is intractable in combined complexity (as
is already the case with conjunctive query evaluation), and thus this setting
will not yield tractable cases of subsequence matching with length constraints.
One exception is the recent work~\cite{figueira2022when} which gives 
classes of ECRPQs with tractable combined complexity, but this tractability is
achieved by bounding the treewidth of the ECRPQs (among other restrictions),
whereas the interval structure of the gap constraints may have arbitrary
treewidth in our setting.

\section{Preliminaries}
\label{sec:prelim}

Let $\mathbb{N} = \{1, 2, 3, \ldots\}$ and, for every $i, j \in \mathbb{N}$ with
$i \leq j$, we define $[i, j] = \{i, i+1, \ldots, j\}$ and $[i] = \{1, 2,
\ldots, i\}$;
for convenience we define $[0] = \emptyset$. We let $\Sigma$ be a finite alphabet of symbols (sometimes called letters) and write $\Sigma^*$ for the set of strings (sometimes called words) over $\Sigma$. We write $\eword$ for the empty string. For a string $w \in \Sigma^*$, we denote by $|w|$ its length and, for every $i \in \{1, 2, \ldots, |w|\}$, we denote by $w[i]$ the $i^{\text{th}}$ symbol of $w$. 
For $i, j \in \{1, 2, \ldots, |w|\}$ with $i \leq j$, we denote by $w[i:j]$ the factor (also called infix) $w[i]w[i+1]\cdots w[j]$; in particular, $w[i:i] = w[i]$.
We extend the notation for $i > j$ as follows. If $i = j + 1$, then we set $w[i : j] = \eword$ and if $i > j + 1$, then we set $w[i : j] = \bot$, where $\bot$ means \emph{undefined}. This technical particularity is due to the fact that we often want to talk about the factor of a string $w$ strictly in between positions $i$ and $j$, which is $\eword$ if $j = i + 1$ and undefined if $j \leq i$. Thus, with our definition, we can use $w[i + 1 : j - 1]$ to refer to this factor that lies strictly in between positions $i$ and $j$. \emph{Regular expressions}, \emph{nondeterministic finite automata with $\eword$-transitions}, called $\eword$NFA for brevity, and the class $\REG$ of \emph{regular languages} are defined in the usual way (see, e.g.,~\cite{hopcroft}). We write $\LL(X)$ for the language of a regular expression $X$ or an $\eword$NFA $X$. We use the classical result that a given regular expression $r$ can be converted in time $O(|r|)$ into an $\eword$NFA $A$ such that $\LL(A) = \LL(r)$ and $|A| = O(|r|)$ (see~\cite[Section~3.2.3]{hopcroft}).

\medskip

\subsection{Embeddings, Subsequences and Gap-Constraints}
Let us define
embeddings and subsequences:

\begin{definition}
Given strings $w, u\in\Sigma^\ast$ with $|u| \leq |w|$, we say that a mapping $e\colon [|u|]\to[|w|]$ is an \emph{embedding of $u$ in $w$} if $e(1)<e(2)<\ldots<e(|u|)$ and $w[e(i)] = u[i]$ for every $i \in [|u|]$, i.e., $u = w[e(1)] w[e(2)] \cdots w[e(|u|)]$.
In this case, we write $u \preceq_e w$ and call $u$ a \emph{subsequence} of~$w$ induced by $e$. We write $u \preceq w$ to denote that $u \preceq_e w$ for some embedding $e$ of $u$ in $w$ and also say that $e$ \emph{witnesses} $u \preceq w$. For $i\in [|u|]$, we say that $u_e[i]$ is \emph{embedded} on position $e(i)$ of $w$. 
\end{definition}

\begin{definition}
Let $w, u\in\Sigma^\ast$ with $|u| \leq |w|$ and let $e\colon [|u|]\rightarrow
[|w|]$ be a mapping.
For every $i,j\in[|u|]$, with $i<j$, we define the
\emph{$(i, j)$-gap} induced by $w$ and $e$ as $w[e(i)+1:e(j)-1]$.
\end{definition}

We define gaps not only between two consecutive positions (i.e., the $(i,i+1)$-gap which is strictly between $e(i)$ and $e(i+1)$), but also between positions that are further apart
(i.e., the $(i,j)$-gap with $i<j$, which contains in particular the images $e(k)$ for each $i < k < j$). Intuitively, 
the $(i,j)$-gap $w[e(i)+1:e(j)-1]$ is the factor that occurs strictly between the positions corresponding to the images of $i$ and $j$ under the
mapping 
$e$. 
Note that gaps may be empty or undefined strings: $w[e(i)+1:e(j)-1]=\eword$ iff $j = i + 1$ and $e(j) = e(i) + 1$, and $w[e(i)+1:e(j)-1]=\bot$ iff $e(i)\geq e(j)$.
Next, we introduce gap-constraints.

\begin{definition}
For a language $L \subseteq \Sigma^\ast$, an \emph{$L$-gap-constraint} for a
string $u \in \Sigma^\ast$ is a triple $(i,j,L)$, where $i, j \in [|u|]$ such
that $i < j$. 
A mapping 
$e \colon [|u|]\rightarrow [|w|]$, where $w \in
\Sigma^*$ with $|u| \leq |w|$, satisfies the gap-constraint $(i,j,L)$ if and
only if  the gap
$w[e(i)+1:e(j)-1]$ belongs to~$L$.
\end{definition}

We call $L$ the \emph{gap-constraint language} of an $L$-gap-constraint. For any $L$-gap-constraint $c=(i,j,L)$, we also say that $c$ is an \emph{$(i,j)$-gap-constraint} and we use the term \emph{gap-constraint} to denote any $L$-gap-constraint for an arbitrary language $L$. We say that $c=(i, j ,L)$ is a \emph{regular constraint} if $L \in \REG$. For our algorithms, we will assume that regular gap-constraints $c=(i, j ,L)$ are represented by the numbers $i, j$ and a regular expression $r$ with $\LL(r) = L$; thus, we will also write them as $c=(i, j, r)$.
The size $|c|$ of the constraint is defined as $|c| = |r|$.
Note that, as we explained above, we can also assume that any regular expression $r$ is given by an $\eword$NFA.

We call $c$ a \emph{length constraint} if $L = \{v\in\Sigma^\ast\mid a\leq |v|\leq b\}$ for some $a, b \in \mathbb{N} \cup \{0\}$ with $a \leq b$. In this case, a regular expression $r$ for $L$ can be concisely described by the interval $[a, b]$; thus, we will write length constraints as $(i, j, [a, b])$ instead of $c=(i, j, L)$ or $c=(i, j, r)$, and we consider such a length constraint $c$ to be of constant size (since it is represented by $4$ numbers). 

For a set $\mathcal{C}$ of gap-constraints for $u \in \Sigma^*$, $w \in \Sigma^*$ with $|w| \geq |u|$ and 
mapping 
$e \colon [|u|]\rightarrow [|w|]$, 
we say that $e$ \emph{satisfies} $\mathcal{C}$ if and only if it satisfies every $c \in \mathcal C$. 

An embedding $e \colon [|u|]\rightarrow [|w|]$ of $u$ in $w$ is a \emph{\mbox{$\mathcal C$-embedding}} of $u$ in $w$ (denoted by $u \preceq_{e, \mathcal{C}} w$) if it satisfies the set $\mathcal{C}$ of gap-constraints for $u$.
We shall also just write $u \preceq_{\mathcal{C}} w$ to denote that $u \preceq_{e, \mathcal{C}} w$ holds for some $e$. 

When we are dealing with a set $\mathcal{C}$ of regular gap-constraints as input for an algorithm, then we measure the \emph{size} of $\mathcal{C}$ as $\lVert \mathcal{C} \rVert = \sum_{c \in \mathcal{C}} |c|$, which by the above 
  is $\sum_{(i, j, r) \in \mathcal{C}} |r|$. Further, the \emph{number} of gap-constraints is $|\mathcal{C}|$. Obviously, we always have $\lVert \mathcal{C} \rVert \geq |\mathcal{C}|$.

\medskip

\subsection{Subsequence Matching with Gap-Constraints}
In the rest of this paper, we are interested in the following computational problem. The input is a \emph{query string} $u \in \Sigma^*$ along with a set $\mathcal{C}$ of regular\footnote{Note that we will always assume that the gap-constraints are \emph{regular}.} gap-constraints for $u$, and a \emph{document} $\doc \in \Sigma^*$. Our interpretation is that the pair $(u,{\mathcal C})$ is a query and $u_{\mathcal{C}}(\doc) = \{e \mid u \preceq_{e, \mathcal{C}} \doc\}$ is the \emph{result set} of $(u,{\mathcal C})$ on $\doc$, i.e., the set of all $\mathcal{C}$-embeddings of $u$ in~$\doc$.
We study the \emph{matching problem}, which is to compute a witness from $u_{\mathcal{C}}(\doc)$ or to report that $u_{\mathcal{C}}(\doc) = \emptyset$ (for our hardness results, we will consider the obvious decision problem variant that simply checks whether $u_{\mathcal{C}}(\doc) \neq \emptyset$).
We will also consider the problem of computing the whole result set $u_{\mathcal{C}}(\doc)$, which we approach as the problem to 
enumerate its contents (i.e., each $\mathcal{C}$-embedding is produced exactly once, with no duplicates).
As is common in such query evaluation scenarios, we assume that the query $(u, {\mathcal C})$ is much smaller than the data $|\doc|$; thus, for running times, our focus is on the dependency on $|\doc|$. Further, when computing a set, we are interested in the total time to do that (in dependency of the size of the set), while for enumerating a set we are interested in two dimensions: the running time of the preprocessing phase (i.e., until the first answer is produced), and the bound on the worst-case delay after each answer (i.e., the time until we produce either the next answer or signal that the enumeration is finished).

\section{Subsequence Matching With Gap-Constraints as a CSP}\label{sec:CSP}

In this section, we show that interpreting the subsequence matching with gap-constraints problem as a Constraint Satisfaction Problem (for short, CSP) 
leads to some initial tractability results.
For instance, it 
gives a polynomial algorithm for the setting where
left-convex regular constraints (e.g., length constraints) are considered. 

Let us briefly introduce the CSP-related machinery that we will use in this section; we follow the definitions from \cite{FreuderM06,Bessiere06}. A more thorough discussion on this topic is given in Appendix~\ref{CSPAppendix}. 

A \emph{constraint satisfaction problem} is a triple ${\mathcal P}=\langle V, \domain, \const \rangle$, where $V=\{1,\ldots, m\}$ is a set of {\em variables} (denoted, for simplicity, by natural numbers), $\domain$ is a set called the {\em domain}, and $\const$ is a set of {\em constraints} $\{\const_1, \ldots,\const_q\}$, such that each $\const_i\in \const$ is a pair $\langle s_i,R_i\rangle$, where $s_i\subseteq V$ is a set of variables of size $n_i$, called the \emph{constraint scope}, and $R_i\subseteq \domain^{n_i}$, given as a list of tuples, 
is an $n_i$-ary relation over $\domain$, called the \emph{constraint relation}; w.l.o.g., we assume $s_i$ to be increasingly ordered, and denote by $s_i[j]$ the $j^{th}$ element of $s_i$. A \emph{solution} for $\mathcal P$ is a function $\phi\colon V\rightarrow \domain$ such that 
for each $\langle s_i,R_i\rangle \in \const$ 
the tuple $\langle \phi(s_i[1]),\ldots,\phi(s_i[n_i]) \rangle$ %
is in~$R_i$. 
The CSP ${\mathcal P}=\langle V, \domain, \const\rangle$ is \emph{normalised} if no two distinct constraints from $\const$ have the same scope. \looseness=-1

In general, for a CSP ${\mathcal P}$, we are interested in whether it admits a solution or not. To this end, it was shown \cite{Zhuk20} that the tractability of ${\mathcal P}$ is exactly determined by structural properties of the class of relations over $\domain$ from which the constraints of $\const$ stem. But already the less recent survey~\cite{CohenJ06} showed two particularly relevant such classes of constraints (which also allow for simpler CSP-solving algorithms, compared to the general algorithm given in \cite{Zhuk20}): the min-closed and, respectively, max-closed relations. A $k$-ary relation $R$ is \emph{min-closed} if for every pair of $k$-tuples $t,t'\in R$ we have that the $k$-tuple $\langle \min\{t[1],t'[1]\}, \ldots, \min\{t[k],t'[k]\}\rangle$ also belongs to~$R$ (max-closed relations are defined analogously). The following result holds. \looseness=-1

\begin{proposition}[Example 6.39 in \cite{CohenJ06}, originally in \cite{JeavonsC95}]\label{prop:minclosedCSP}
We can decide in polynomial time whether a CSP ${\mathcal P}=\langle V, \domain, \const\rangle$, where all constraints of $\const$ are min-closed (resp., max-closed), admits a solution.
\end{proposition}

Let us make some remarks on the above result. Firstly, note that there are quite a few interesting examples of min- or max-closed relations, as shown in \cite{JeavonsC95}. For instance, all unary constraints are min- and max-closed, and the same holds for all basic arithmetic constraints over the natural numbers in the constraint programming language CHIP \cite{HentenryckDT92}, which include, e.g., relations defined by linear (in)equalities between variables. 
Secondly, many efficient algorithms support the above statement, not only the one from \cite{JeavonsC95}; in particular, if there is a constant upper bound on the arity of all constraint relations, then algorithms implementing the technique of \emph{enforcing generalised arc consistency}, see \cite{Bessiere06}, achieve a polynomial running time for the above problem with a polynomial degree that depends on the arity bound. For normalised CSPs where each constraint is unary or binary, such an algorithm is AC4 (defined in \cite{YuanlinY01}, as an efficient implementation of AC3 \cite{Mackworth77}) which works in $O(|\const| |\domain|^2)$ time (which is optimal \cite{Bessiere06}). 
Note that, since each relation $R_i$ is unary or binary, the size of each constraint $\const_i = \langle s_i,R_i\rangle\in\const$ is asymptotically equal to the cardinality $r_i$ of the constraint relation $R_i$, given as a list of tuples, and thus $|\const|=O(\sum_i r_i) $. \looseness=-1

Coming now back to the subsequence matching with gap-constraints problem, we show how this can be formalised as a CSP. Recall that we are given a \emph{query string} $u \in \Sigma^*$, $|u|=m$, along with a set $\mathcal{C}$ of regular gap-constraints for $u$, and a \emph{document} $\doc \in \Sigma^*$, $|\doc|=n$, and we want to see if there exists a $\mathcal{C}$-embedding of $u$ in $\doc$. Let $V=[m]$ and $\domain=[n]$. 
The set $\const'$ contains the following constraints:\looseness=-1
\begin{itemize}
\item $\langle (i) , R_i\rangle$, where $i\in [m]$ and $R_i=\{ a\in [n]\mid  \doc[a]=u[i]\}$.
\item $\langle (i,i+1) , R_{(i,i+1)}\rangle$, where $i\in [m-1]$ and $R_{(i,i+1)} =\{(a,b)\in [n]\times [n] \mid a<b\}$. 
\item $\langle (i,j) , R_{(i,j,L)}\rangle$, where $(i,j,L)\in {\mathcal C}$ and $R_{(i,j,L)} =\{(a,b) \in [n]\times[n]\mid  \doc[a+1:b-1]\in L\}$. 
\end{itemize}

Let $\CSP'_{u,\doc,{\mathcal C}}=\langle V,\domain,\const'\rangle$ be the constraint satisfaction problem defined for the instance $u,\doc, {\mathcal C}$ of the subsequence matching with gap-constraints problem. It is immediate that this CSP has a solution if and only if there exists a ${\mathcal C}$-embedding of $u$ in $\doc$. Note that, while having only unary and binary constraints, this CSP is not necessarily normalised (as there can be multiple constraints for the same pair of variables).
Fortunately, we can show the following result (see Appendix~\ref{CSPAppendix} for a proof).\looseness=-1

\begin{restatable}{lemma}{RegularMembershipSubstrRest}
\label{lem:RegularMembershipSubstr}
Given $u, \doc,$ and ${\mathcal C}$, we can construct in $O(\lVert{\mathcal C} \rVert |\doc|^2)$ time the problem $\CSP'_{u,\doc,{\mathcal C}}=\langle V,\domain,\const'\rangle$. In the same time complexity, we can construct $\CSP_{u,\doc,{\mathcal C}}=\langle V,\domain,\const\rangle$, a normalised problem with $O(|u|+|{\mathcal C}|)$ constraints, which has a solution if and only if there exists a ${\mathcal C}$-embedding of $u$ in $\doc$.\looseness=-1
\end{restatable}

Now, in the case where all the constraints of $\CSP_{u,\doc,{\mathcal C}}$ happen to be min-closed, we can use Proposition~\ref{prop:minclosedCSP} to get the following result, using, e.g., the AC4 algorithm to enforce arc consistency. 

\begin{restatable}{theorem}{CSPsolutionRest}
\label{thm:CSPsolution}
Given $u, \doc,$ and ${\mathcal C}$, let $\CSP_{u,\doc,{\mathcal C}}=\langle V,\domain,\const\rangle $ be constructed as in Lemma~\ref{lem:RegularMembershipSubstr}. Assuming that the set of constraints $\const$ is min-closed, then we can decide whether there exists a ${\mathcal C}$-embedding of $u$ in $\doc$ in time $O(|\doc|^2 ( |u| + \lVert{\mathcal C}\rVert ))$ (which includes the construction of $\CSP_{u,\doc,{\mathcal C}}$).  
\end{restatable} 

We note that we can easily test in polynomial time whether it is the case that the set of constraints $\const$ is min-closed, once we have constructed $\CSP_{u,\doc,{\mathcal C}}$. However, the above result provides no algorithm to solve the problem when $\const$ is not min-closed, and we also have no guarantee on when this might happen. In the sequel we will see how enforcing the \emph{left-convexity} of the regular languages used as gap-constraints can ensure that the constraints are min-closed.

Recall the definition of left-convex languages from 
the Introduction (Definition~\ref{def:lcon}): a language $L$ is left-convex if $uvw \in L$ and $v\in L$ implies $uv \in L$. Let $\LCON$ be the class of left-convex languages.

\begin{restatable}{lemma}{LCONmin}
\label{lem:LCON_min_closed}
If $L \in \LCON$ for all $(i, j, L) \in \mathcal{C}$, then all constraints of $\CSP_{u,\doc,{\mathcal C}}$ are min-closed. 
\end{restatable}

Clearly, if $L$ is right-convex, then all constraints of $\CSP_{u,\doc,{\mathcal C}}$ are max-closed. So, the result of Theorem~\ref{thm:CSPsolution} applies for the case when all gap-constraints in ${\mathcal C}$ are left-convex (respectively, right-convex). In the following, and as the main technical contribution of our work, we show that this result can be improved in a non-trivial way regarding two different aspects. Firstly, we are able to improve the running time from $O(|\doc|^2 (|u| + \lVert{\mathcal C}\rVert ))$ to $O(|\doc| ( |u| + \lVert{\mathcal C}\rVert ))$ (i.e., from quadratic to linear in the data size, which is highly relevant for applications in complex event recognition), which is conditionally optimal, and secondly we are able to extend our algorithm to enumerating all solutions with bounded delay instead of producing only one witness.

\section{Improvement to Linear Dependency on the Data Size}\label{sec:mainAlgo}

In this section, we substantially improve the upper bound from the previous section obtained by the CSP-based approach. More precisely, we show that the quadratic dependency on the size of the data $|\doc|$ can be lowered to a linear dependency:

\begin{theorem}\label{mainAlgoTheorem}
The matching problem with left-convex constraints can be solved in $O(|\doc| (|u| + \lVert \mathcal C\rVert))$.
\end{theorem}

We can also show that under the strong exponential time hypothesis (SETH, for short), this bound cannot be improved, i.e., for $\epsilon>0$, there is most likely no algorithm with running time $O((|\doc| (|u| + \lVert \mathcal C\rVert))^{1-\epsilon})$ (see Section~\ref{sec:lowerBound} below).

Before we can state our algorithm for Theorem~\ref{mainAlgoTheorem}, we will need the following crucial definition. For every $k\in [|\doc|]$ and language $L$ with $(i,j,L) \in \mathcal C$, we define $\longright(k, L) = \max(\{t\in[k+1,|\doc|]\mid \doc[k+1:t-1]\in L\} \cup \{0\})$,
i.e., if $\longright(k, L)=t$, then $\doc[k+1:t-1]$ is the \emph{longest} factor of $\doc$ starting at position $k+1$ that belongs to $L$. Analogously, we define $\shortleft(k, L) = \max(\{t\in[k-1]\mid \doc[t+1:k-1]\in L\} \cup \{0\})$,
i.e., if $\shortleft(k, L)=t$, then $\doc[t+1:k-1]$ is the \emph{shortest} factor of $\doc$ ending at position $k-1$ that belongs to $L$. Moreover, if no such factor exists or if $k\notin[|\doc|]$, then $\longright(k, L)$
and $\shortleft(k, L)$
return $0$.

\begin{algorithm}
    \caption{$\embedLCONSubseq(u, \mathcal{C}, \doc, e_0)$}
        \label{mainAlgo}
   \KwIn{$u, \mathcal{C}$ s.t. $\mathcal{C}$ only contains left-convex gap-constraints, $\doc \in\Sigma^\ast$, $e_0 \colon [|u|] \to [|\doc|]$.}
    \KwOut{$e_0$-minimal $\mathcal C$-embedding of $u$ in $\doc$ if it exists, and $\bot$ otherwise.}
        $e \coloneq e_0$; $S \coloneq \{1, 2, \ldots, |u|\}$\;
        \While{$S \neq \emptyset$}
        {
        Let $s \in S$ be arbitrarily chosen and $S \gets S \setminus \{s\}$\;\label{startLine}
        \lIf{$e(s) > |\doc|$}{\textbf{return} $\bot$\label{returnBotLine}}
        \If{$u[s] \neq \doc[e(s)]$\label{symbolCond}}
        {
        	$e(s) \gets e(s) + 1$; $S \gets S \cup \{s\}$\;\label{addsLine}
        }
        \If{$s < |u|$ and $e(s) \geq e(s + 1)$\label{orderCond}}
        {
        	$e(s + 1) \gets e(s) + 1$; $S \gets S \cup \{s + 1\}$\;\label{endFirstPartLine}
        }
        \ForEach{$(i, j, L) \in\mathcal C$ with $s \in \{i, j\}$}
        {
        	\If{$\doc[e(i) + 1 : e(j)-1] \notin L$\label{condLineOne}}
		{
            		\lIf{$\longright(e(i), L) > e(j)$\label{condLineTwo}}
			{
				$s' \gets j$ \textbf{else} $s' \gets i$
     			}
     			$e(s') \gets e(s') + 1$; $S \gets S \cup \{s'\}$\;
     		}
        }
        }
        \textbf{return} $e$\;\label{returneLine}
\end{algorithm}

Let us explain our algorithm (given as Algorithm~\ref{mainAlgo}) on an intuitive level. The algorithm starts with an initial mapping $e \coloneq e_0 \colon [|u|] \to [|\doc|]$, which is then changed into a $\mathcal C$-embedding of $u$ in $\doc$ (if one exists) by only moving single positions to the right, i.e., by single steps that re-define some $e(i) \coloneq e(i) + c$
and leave all other $e(i')$ with $i' \in [|u|] \setminus \{i\}$ unchanged.
Now for our current mapping $e$ there are three types of \emph{violations} that may cause $e$ to \emph{not} be a $\mathcal C$-embedding of $u$ in $\doc$: (1) $u[s] \neq \doc[e(s)]$ for some $s \in [|u|]$, (2) $e(s) \geq e(s + 1)$ for some $s \in [|u| - 1]$, and (3)~$\doc[e(i) + 1 : e(j)-1] \notin L$ for some $(i, j, L) \in\mathcal C$
with $s\in\{i,j\}$. In the set $S$, we maintain positions $s \in [|u|]$ for which we still have to check whether they participate in one of these violations, i.e., we maintain the invariant that if $s \notin S$, then $s$ does not participate in any violations of any of these three types (note that for violation (3) this means that there is no $(i, j, L) \in\mathcal C$ with $s \in \{i, j\}$ such that $\doc[e(i) + 1 : e(j)-1] \notin L$).
Therefore, as soon as $S$ is empty, we have found our $\mathcal C$-embedding of $u$ in $\doc$. Let us next discuss how we modify the current mapping $e$. As long as $S$ is not empty, we consider some $s \in S$, and remove it from $S$. If $s$ participates in a violation of type (1), then we move $e(s)$ one step to the right (see Line~\ref{symbolCond}) and if $s$ participates in a violation of type (2), then we move $e(s + 1)$ to $e(s) + 1$ (see Line~\ref{orderCond}). If $s$ is responsible for a violation of type (3), then there is some $(i, j, L) \in\mathcal C$ with $s \in \{i, j\}$ and $\doc[e(i) + 1 : e(j) - 1] \notin L$. In this case, we move either $e(j)$ or $e(i)$ one step to the right, depending on whether or not $\longright(e(i), L) > e(j)$ (see Line~\ref{condLineTwo}). Moreover, whenever we move some position, we have to put it in $S$, since moving it may cause violations that have to be checked later. In every iteration of the while-loop, either we move a position of the current mapping $e$ to the right, or, if we do not change $e$, then we remove an element from $S$ without adding new elements. Consequently, at some point either $S$ will be empty, in which case we output the current mapping, or the current mapping $e$ satisfies $e(i) > |\doc|$ for some $i \in [|u|]$, in which case we conclude that no $\mathcal{C}$-embedding exists.

We will next prove the correctness of this approach and then we prove the running time claimed in Theorem~\ref{mainAlgoTheorem} (for some proofs we only provide sketches; full details can be found in Appendix~\ref{mainAlgoAppendix}). However, before we can proceed, we need some more definitions and statements.

In the following, we fix an instance $(u, \mathcal{C}, \doc)$ of the matching problem, where every gap-constraint from $\mathcal{C}$ is a left-convex gap-constraint. We denote by $\leq$ the pointwise order on mappings, i.e., for two mappings $e, e'\colon [|u|] \to [|\doc|]$, we write $e \leq e'$ if and only if $e(i) \leq e'(i)$ for every $i \in [|u|]$. Clearly, `$\leq$' is a partial order, which means that we can talk about mappings that are \emph{minimal} (with respect to a set of mappings). For some mapping $e_{0}\colon [|u|] \to [|\doc|]$, we say that $e\colon [|u|] \to [|\doc|]$ is an \emph{$e_{0}$-minimal $\mathcal{C}$-embedding of $u$ in $\doc$} if and only if $e$ is a $\mathcal{C}$-embedding of $u$ in $\doc$ with $e_0 \leq e$ and $e$ is minimal within the set of all $\mathcal{C}$-embeddings $e'$ of $u$ in $\doc$ with $e_0 \leq e'$. Note that this notion is defined for a mapping $e_0 \colon [|u|] \to [|\doc|]$,
i.e., $e_0$ is not necessarily an embedding of $u$ in $\doc$ (and, even when it is an embedding of $u$ in $\doc$, it may still violate gap-constraints from $\mathcal{C}$). A key point of the left-convex property is stated by the next lemma (similar to Lemma \ref{lem:LCON_min_closed}).\looseness=-1

\begin{lemma}\label{uniquenessLemma}
Let $e_{0}\colon [|u|] \to [|\doc|]$ be a mapping, and assume that there exists a $\mathcal{C}$-embedding $e$ of $u$ in $\doc$ with $e_0 \leq e$. Then there is a unique $e_{0}$-minimal $\mathcal{C}$-embedding of $u$ in $\doc$.
\end{lemma}

\begin{proof}
If there are two $e_{0}$-minimal $\mathcal{C}$-embeddings $e_1$ and $e_2$ of $u$ in $\doc$, then we can consider the pointwise minimum $e_{\min}$ of $e_1$ and $e_2$, which is an embedding of $u$ in $\doc$. Obviously, $e_0 \leq e_{\min}$. We show that $e_{\min}$ is a $\mathcal{C}$-embedding of $u$ in $\doc$. To this end, let $(i, j, L) \in \mathcal{C}$ be arbitrarily chosen. If $(e_{\min}(i), e_{\min}(j)) = (e_1(i), e_1(j))$ or $(e_{\min}(i), e_{\min}(j)) = (e_2(i), e_2(j))$, then $e_{\min}$ obviously satisfies $(i, j, L)$. If $(e_{\min}(i), e_{\min}(j)) = (e_1(i), e_2(j))$ (the case $(e_{\min}(i), e_{\min}(j)) = (e_2(i), e_1(j))$ is analogous), then we have $e_1(i) \leq e_2(i) < e_2(j) \leq e_1(j)$ and $\doc[e_1(i) + 1 : e_1(j) - 1] \in L$ and $\doc[e_2(i) + 1 : e_2(j) - 1] \in L$. Since $L$ is left-convex, this implies $\doc[e_1(i) + 1 : e_2(j) - 1] \in L$ and therefore $\doc[e_{\min}(i) + 1 : e_{\min}(j) - 1] \in L$. Hence, $(i, j, L)$ is satisfied by $e_{\min}$. Finally, we note that $e_1 \neq e_{\min}$ or $e_2 \neq e_{\min}$ contradicts the $e_{0}$-minimality of $e_1$ or $e_2$; thus, $e_1 = e_2 = e_{\min}$.
\end{proof}

We are now ready to state the correctness of Algorithm~\ref{mainAlgo} in the form of the following lemma.

\begin{restatable}{lemma}{mainAlgoCorrectnessRest}\label{mainAlgoCorrectnessLemma}
On input $u \in \Sigma^*$, $\mathcal{C}$ (where $\mathcal{C}$ only contains left-convex gap-constraints), $\doc \in\Sigma^\ast$ and $e_0 \colon [|u|] \to [|\doc|]$, Algorithm~\ref{mainAlgo} returns an $e_0$-minimal $\mathcal C$-embedding of $u$ in $\doc$ if it exists, and $\bot$ otherwise.
\end{restatable}

\renewcommand{\proofname}{Proof Sketch}
\begin{proof}
For every $e \colon [|u|] \to [|\doc|]$, we define the following invariant.

\smallskip

\noindent \emph{Invariant $(\dagger)_{e}$}: If there exists a $\mathcal C$-embedding $e^*$ of $u$ in $\doc$ with $e_0 \leq e^*$, then $e \leq e^*$ (where $e$ is the current mapping of the algorithm). 

\smallskip

It can be easily verified that $(\dagger)_{e}$ holds at the beginning of the algorithm, and it can be shown that the modifications of $e$ carried out by an iteration of the while-loop maintain $(\dagger)_{e}$. 
For the modifications possibly caused by the if-statements in Line~\ref{symbolCond}~and~\ref{orderCond} this is not too hard to see. If the condition of Line~\ref{symbolCond} is satisfied, then $u[s] \neq \doc[e(s)]$ and we will change $e$ into $e'$ with $e'(i) = e(i)$ for every $i \in [|u|] \setminus \{s\}$ and $e'(s) = e(s) + 1$. Since $(\dagger)_e$ is satisfied, we know that any $\mathcal C$-embedding $e^*$ of $u$ in $\doc$ with $e_0 \leq e^*$ satisfies $e \leq e^*$. But since $u[s] \neq \doc[e(s)]$, we also know that $e^*(s) \neq e(s)$, which means that $e'(s) = e(s) + 1 \leq e^*(s)$. Consequently, any $\mathcal C$-embedding $e^*$ of $u$ in $\doc$ with $e_0 \leq e^*$ satisfies $e' \leq e^*$; thus, $(\dagger)_{e'}$ is satisfied. A similar reasoning applies to the case of Line~\ref{orderCond}. 

The more difficult cases are the calls of Line~\ref{condLineOne} inside the foreach-loop, whose correctness hinges on the left-convex property. Let us assume that the condition of Line~\ref{condLineOne} is satisfied for some $(i, j, L) \in\mathcal C$ with $s \in \{i, j\}$, i.e., $\doc[e(i) + 1 : e(j)-1] \notin L$. Moreover, let us assume that $\longright(e(i), L) > e(j)$, which means that we change $e$ into $e'$ with $e'(k) = e(k)$ for every $k \in [|u|] \setminus \{j\}$ and $e'(j) = e(j) + 1$. In particular, we also know that $\doc[e(i) + 1 : \longright(e(i), L) - 1] \in L$. If $(\dagger)_{e'}$ is not satisfied, then there is some $\mathcal C$-embedding $e^*$ of $u$ in $\doc$ with $e_0 \leq e^*$, but $e' \not\leq e^*$. Since $(\dagger)_e$ is satisfied, we know that $e \leq e^*$, which means that $e(j) = e^*(j)$. Since $\doc[e(i) + 1 : e(j)-1] \notin L$, we also know that $e'(i) = e(i) < e^*(i)$. Let us now consider the string $uvw$ with $u = \doc[e(i) + 1 : e^*(i)]$, $v = \doc[e^*(i) + 1 : e^*(j)-1]$ and $w = \doc[e^*(j) : \longright(e(i), L) - 1]$. As observed above, $uvw \in L$, and since $e^*$ is a $\mathcal C$-embedding of $u$ in $\doc$, we also know that $v \in L$, which, by the left-convexity of $L$, implies that $uv = \doc[e(i) + 1 : e(j) - 1] \in L$; a contradiction. The case $\longright(e(i), L) \leq e(j)$, where we have $e'(k) = e(k)$ for every $k \in [|u|] \setminus \{i\}$ and $e'(i) = e(i) + 1$, is a bit simpler. If $(\dagger)_{e'}$ is not satisfied, then there is some $\mathcal C$-embedding $e^*$ of $u$ in $\doc$ with $e_0 \leq e^*$, but $e' \not\leq e^*$. Since $(\dagger)_e$ is satisfied, we know that $e \leq e^*$, which means that $e(i) = e^*(i)$. Since $\doc[e(i) + 1 : e(j)-1] \notin L$, we also know that $e'(j) = e(j) < e^*(j)$, which means that $\longright(e(i), L) > e(j)$; a contradiction.

Note that we might also call Line~\ref{condLineOne} in the case where $e(i) \geq e(j)$, which means that $\doc[e(i) + 1 : e(j)-1] = \bot \notin L$. This constitutes a special case, for which correctness can nevertheless be proven in a very similar way as sketched above (see Appendix~\ref{mainAlgoAppendix}).

This shows that $(\dagger)_{e}$ holds before every iteration of the while-loop. Next, we formulate another invariant, for which we need the following terminology. We say that $s \in [|u|]$ \emph{satisfies the symbol condition} with respect to $e$ if $u[s] = \doc[e(s)]$, $s \in [|u|-1]$ \emph{satisfies the order condition} with respect to $e$ if $e(s) < e(s + 1)$, and $s, s' \in [|u|]$ \emph{satisfy the gap condition} with respect to $e$ if $(i, j, L) \in\mathcal C$ with $\{i, j\} = \{s, s'\}$ implies $\doc[e(i) + 1 : e(j) - 1] \in L$. 
For every $e \colon [|u|] \to [|\doc|]$ and $S \subseteq [|u|]$, we define:

\smallskip

\noindent \emph{Invariant $(\ddagger)_{e, S}$}: Each $s \in [|u|] \setminus S$ satisfies the symbol condition and the order condition with respect to $e$, and all $s, s' \in [|u|] \setminus S$ satisfy the gap condition with respect to $e$. 

\smallskip

Obviously, $(\ddagger)_{e, S}$ is satisfied at the beginning of the algorithm, since $S = [|u|]$. It can be shown that $(\ddagger)_{e, S}$ is maintained by iterations of the while-loop as follows. We can first show that if $(\ddagger)_{e, S}$ is satisfied before an iteration, then $(\ddagger)_{e, S \cup \{s\}}$ is satisfied after the iteration (note that $e$ and $S$ always refer to these program variables at the point of the algorithm we refer to). 
This is because whenever we change any position of $e$
then we also put it into $S$. So it remains to prove that  $(\ddagger)_{e, S}$ holds after the iteration. 
The fact that $(\ddagger)_{e, S \cup \{s\}}$ holds directly means that all $s' \in [|u|] \setminus (S \cup \{s\})$ satisfy the symbol and order condition, and all $s', s'' \in [|u|] \setminus (S \cup \{s\})$ satisfy the gap condition. Now if $s \in S$, then $(\ddagger)_{e, S \cup \{s\}} = (\ddagger)_{e, S}$ and we are done. If $s \notin S$, then the fact that $s$ was never added to $S$ by any of the Lines~\ref{symbolCond},~\ref{orderCond}~and~\ref{condLineOne} directly means that $s$ must also satisfy the symbol and order condition, and all $i, j \in [|u|] \setminus S$ with $s \in \{i, j\}$ satisfy the gap condition. Thus, $(\ddagger)_{e, S}$ holds after the iteration.

These two invariants directly imply the correctness. The algorithm obviously terminates at some point. If $S$ gets empty in the last iteration, then we return $e$, which, due to $(\ddagger)_{e, S}$ has to be a $\mathcal{C}$-embedding of $u$ in $\doc$, and due to $(\dagger)_e$, has to be $e_0$-minimal.
If 
$e(i) > |\doc|$ in the last iteration, then, due to $(\dagger)_{e}$,
there is no $\mathcal{C}$-embedding $e^\ast$ of $u$ in $\doc$ with $e_0\leq e^\ast$ and therefore we return $\bot$.
\end{proof}
\renewcommand{\proofname}{Proof}

Let us next come to estimating the running time of Algorithm~\ref{mainAlgo}. If we can perform the checks of Lines~\ref{condLineOne}~and~\ref{condLineTwo}, i.e., checking $\doc[e(i) + 1 : e(j)-1] \notin L$ and $\longright(e(i), L) > e(j)$, in constant time, then Algorithm~\ref{mainAlgo} has indeed the running time claimed in Theorem~\ref{mainAlgoTheorem}.

\begin{restatable}{lemma}{mainAlgoRunningTimeRest}\label{mainAlgoRunningTimeLemma}
Under the assumption that we can check the conditions of Lines~\ref{condLineOne}~and~\ref{condLineTwo} in constant time, Algorithm~\ref{mainAlgo} terminates after time $O(|\doc| (|u| + |\mathcal C|))$.
\end{restatable}

\renewcommand{\proofname}{Proof Sketch}
\begin{proof}
We call an iteration of the while-loop an \emph{$s$-iteration} if $s$ is removed from $S$ in Line~\ref{startLine}. We note that for each $s \in [|u|]$ there are $O(|\doc|)$ $s$-iterations, so there are $O(|\doc|\cdot |u|)$ iterations of the while-loop. Hence, executing Lines~\ref{startLine}~to~\ref{endFirstPartLine} of all iterations requires time $O(|\doc| \cdot |u|)$. Each $s$-iteration performs $k_s$ iterations of the foreach-loop, where $k_s$ is the number of gap-constraints $(i, j, L)$ with $s \in \{i, j\}$. Thus, we need time $O(|\doc| k_s)$ for all the foreach-loops inside of $s$-iterations, which over all $s \in [|u|]$ sums up to $O(|\doc| \cdot |\mathcal C|)$ (note that $\sum^{|u|}_{i = 1} k_i = O(|\mathcal C|)$).
\end{proof}
\renewcommand{\proofname}{Proof}

In order to conclude the proof of Theorem~\ref{mainAlgoTheorem}, we have to show how the checks of Lines~\ref{condLineOne}~and~\ref{condLineTwo} can be carried out efficiently. First, we will see that, as a consequence of the left-convex property of $L$, the numbers $\longright(p, L)$ and $\shortleft(q, L)$ tell us whether $\doc[p+1:q-1]\in L$:

\begin{lemma}\label{longestRightShortestLeftLemma}
Given $p,q\in[|\doc|]$ and given a language $L\in\LCON$, we have $\doc[p+1:q-1]\in L$ if and only if $\longright(p, L) \geq q$ and $\shortleft(q, L) \geq p$.
\end{lemma}

\begin{proof}
If $\doc[p+1:q-1]\in L$, then $\longright(p, L) \geq q$ and $\shortleft(q, L)\geq p$ trivially hold. Thus, for $x \coloneq \shortleft(q, L)$ and $y \coloneq \longright(p, L)$, let us assume that we have $x\geq p$ and $y\geq q$. Then $\doc[p+1:y-1]=\doc[p+1:x]\doc[x+1:q-1]\doc[q:y-1]\in L$ and $\doc[x+1:q-1]\in L$ hold. Since $L\in\LCON$, we know that $\doc[p+1:x]\doc[x+1:q-1]=\doc[p+1:q-1]\in L$ holds. %
\end{proof}

Consequently, if we have all the numbers $\longright(k, L)$ and $\shortleft(k, L)$ for every $k \in [|\doc|]$ and $(i, j, L) \in\mathcal C$ at our disposal, then we can perform the checks of Lines~\ref{condLineOne}~and~\ref{condLineTwo} in constant time. Thus, let $\LRArray$ and $\SLArray$ be arrays such that, for every $k\in [|\doc|]$ and language $L$ with $(i,j,L) \in \mathcal C$ for some $i$ and~$j$, we have $\LRArray[k][L] = \longright(k, L)$ and $\SLArray[k][L] = \shortleft(k, L)$. We can compute these arrays using the NFA $M_L$ for $L$, i.e., we construct a product graph of $M_L$ and $\doc$ and then use dynamic programming on this graph; see Appendix~\ref{computeLRSLArraysAppendix} for details.

\begin{restatable}{lemma}{computeLRSLArraysRest}
\label{computeLRSLArrays}
We can compute the arrays $\LRArray$ and $\SLArray$ in time $O(|\doc| \cdot \lVert \mathcal{C} \rVert)$.
\end{restatable}

In order to conclude the proof of Theorem~\ref{mainAlgoTheorem}, we note that we can solve the matching problem with left-convex gap-constraints for some instance $(u, \mathcal{C}, \doc)$ by simply running Algorithm~\ref{mainAlgo} on this instance with the initial mapping $e_0$ defined by $e_0(i) \coloneq i$ for every $i \in [|u|]$.

\subsection{Conditional Lower Bound}\label{sec:lowerBound}

We can show that the upper bound $O(|\doc| (|u| + \lVert \mathcal C\rVert))$ of Theorem~\ref{mainAlgoTheorem} is conditionally optimal. Let $r$ be a regular expression over $\Sigma$ and let $w \in \Sigma^*$. We define the strings $\doc=\#\$w\$\#$ and $u=\#\#$, where $\$, \# \notin \Sigma$, and the set of gap-constraints $\mathcal C=\{(1,2,L)\}$, where $L = \{\$ w \$ \mid w \in \mathcal{L}(r)\}$. Obviously, $(1,2,L)$ is a left-convex gap-constraint, which means that $(u, \mathcal{C},\doc)$ is an instance of the matching problem, and this instance can be obtained in linear time from $w$ and $r$. Moreover, $w \in \mathcal{L}(r)$ if and only if $u \preceq_{\mathcal{C}} \doc$. Hence, if we can check $u \preceq_{\mathcal{C}} \doc$ in time $O((|\doc|\lVert \mathcal C \rVert)^{1-\epsilon})$, for some $\epsilon>0$, then we can also check $w \in \mathcal{L}(r)$ in time $O((|w| |r|)^{1-\epsilon})$. The latter contradicts 
SETH (as it 
contradicts the orthogonal vectors hypothesis, see~\cite{backurs2016regular,DBLP:journals/theoretics/BringmannGKL24}); clearly, if $L$ 
is specified by an arbitrary $\eword$NFA, the same lower bound holds.
Similarly, even if the languages defining the gap-constraints are chosen from some very simple classes (e.g., length constraints), the upper bound $O(|\doc| (|u| + \lVert \mathcal C\rVert))$ of Theorem~\ref{mainAlgoTheorem} is still conditionally optimal. In \cite[Theorem 3]{DayKMS22} it was shown that, even if we only allow gap-constraints of the form $(i,i+1,[a,b])$, with $a,b$ constants, then checking whether $u \preceq_{\mathcal{C}} \doc$ in time ${O}((|\doc|(|u|+\lVert \mathcal C \rVert))^{1-\epsilon})$, for $\epsilon > 0$, would contradict the orthogonal vectors hypothesis, and, as such, SETH \cite{doi:10.1142/9789813272880_0188}. 

\subsection{Jumping Over Larger Factors and Other Alternative Approaches}\label{gaoetalapproach}

  We observe that, when $u[s]$ is not mapped to the right letter by~$e$, then
  Line~\ref{addsLine} only increments $e(s)$ by one. One easy improvement in
  practice would be to jump
  directly to the next occurrence of the letter $u[s]$ in~$\doc$ (or fail if
none exists). Similarly, we observe that when in an iteration of the while-loop the condition of
Line~\ref{condLineOne} is satisfied, i.e., $\doc[e(i) + 1 : e(j)-1] \notin L$,
then Algorithm~\ref{mainAlgo} moves either position $i$ or $j$ by only one
position, i.e., either $e(i) \coloneq e(i) + 1$ or $e(j) \coloneq e(j) + 1$.
However, it can be shown that in such a case, we could as well set $e(i)
\coloneq p$ and $e(j) \coloneq q$, where $p,q\in[|\doc|]$ such that $e(i) \leq
p, e(j) \leq q$, $\doc[p+1:q-1]\in L$ and $p,q$ are minimal with this property.
(The uniqueness of the (pointwise) minimal pair $(p,q)$ is a direct consequence of the left-convexity of $L$; see Lemma~\ref{uniquenessLemma}.)
The correctness of this follows again from the left-convexity property. Especially for instances where $\doc$ is sparse with respect to factors that match the constraint languages, shifting positions over larger chunks of $\doc$ as explained above could help finding a valid embedding (or verifying that none exists) much faster. 

Of course, the question is how the pair $(p, q)$ can be computed efficiently.
For this, we can again exploit the numbers $\longright(k, L)$ and $\shortleft(k,
L)$. More precisely, our desired $p$ and $q$ are given as $p =
\min\{x\in[e(i),|\doc|]\mid \longright(x, L) \geq e(j)\}$ and $q =
\min\{x\in[e(j),|\doc|] \mid \shortleft(x, L)\geq e(i)\}$. However, it is too
costly to na\"ively compute $p$ and $q$ according to these definitions; we can
instead use a data structure result by Gao et
al.~\cite{gao_et_al:LIPIcs.ESA.2020.54} for \emph{orthogonal range successor}
queries to speed up this computation. Unfortunately, this approach leads to an
overall running time of $\tilde{O}(|\doc| (|u| + \lVert \mathcal C\rVert))$,
i.e., asymptotically worse by logarithmic factors.
On the positive side, it can be shown that for the special cases where all gap-constraints are just length constraints, or all gap-constraints are singleton languages, the above described variant of Algorithm~\ref{mainAlgo} that allows jumping over larger chunks of $\doc$ can actually be implemented to run, in the worst case, in time $O(\lVert \mathcal C\rVert + |\doc| (|u| + | \mathcal C|))$, so faster than the general approach from Theorem~\ref{mainAlgoTheorem}. Moreover, in the general case, even though in an asymptotic worst-case analysis we have to pay with logarithmic factors if we want to be able to jump over larger portions of $\doc$, such an improvement could still pay off in an experimental analysis on real-world instances (which would go beyond the scope of this work). All the technical details of the approach described in this subsection are provided in Appendix~\ref{sec:GaoAppendix}.

\section{Computing and Enumerating All Satisfying Embeddings}\label{sec:computeAll}

Our Algorithm~\ref{mainAlgo} from the previous section decides the matching problem (i.e., it checks whether $\query(\doc)\neq\emptyset$) by producing a witness embedding $e\in\query(\doc)$. Next, we extend the procedure to enumerate with bounded delay \emph{all} $\mathcal C$-embeddings of $u$ in~$\doc$. This is 
motivated by the fact that, in practical tasks like complex event recognition,
computing or enumerating all solutions is more important than just computing some witness. For presentational reasons, we first devise an algorithm that computes $\query(\doc)$, which shall then be extended to an enumeration algorithm.

This extension hinges on the property of Algorithm~\ref{mainAlgo} that, called for some starting mapping $e_0$, it returns not just any $\mathcal C$-embedding of $u$ in $\doc$, but the \emph{$e_0$-minimal} $\mathcal C$-embedding $e_{\min}$ of $u$ in $\doc$, if it exists. Hence, our algorithm can be employed recursively in the following general way: We first call Algorithm~\ref{mainAlgo} for the trivial initial mapping $e_0$, where $e_0(i)=i$ for all $i\in[|u|]$, to obtain the $e_0$-minimal ${\mathcal C}$-embedding $e_{\min}$ of $u$ in $\doc$. Next, we want to call Algorithm~\ref{mainAlgo} again to compute somehow a ``next'' $\mathcal C$-embedding $e'_{\min}$ of $u$ in $\doc$ that is larger than $e_{\min}$, and so on. Obviously, we cannot just choose $e_{\min}$ as initial mapping for the next call of Algorithm~\ref{mainAlgo}, since $e_{\min}$ is a $\mathcal C$-embedding of $u$ in $\doc$, so the algorithm would simply output $e_{\min}$ again. Instead, we have to slightly modify $e_{\min}$ before using it as the initial mapping of the next call, and simply moving one of its positions one step to the right seems to be a natural choice of such a modification. However, we have to be careful to explore the complete solution space while ensuring that each embedding is discovered exactly once.
Let us next sketch how this recursion has to be organised.

For a given initial mapping $e_0$
and for every $i \in \{0, 1, 2 \ldots, |u|-1\}$, we define $S(e_0, i) = \{e \mid e_0 \leq e, e(1)=e_0(1),\ldots,e(i)=e_0(i), u\preceq_{e,\mathcal C}\doc\}$, i.e., the set of all $\mathcal{C}$-embeddings of $u$ in $\doc$ lower bounded by $e_0$ that agree with $e_0$ on the first $i$ positions. This set can be partitioned in a natural way as follows. Let us first note that $S(e_0, i) = \emptyset$ if and only if we have property (1): $S(e_0, 0) = \emptyset$ (i.e., there is no ${\mathcal C}$-embedding $e$ of $u$ in $\doc$ with $e_0 \leq e$) or property $(2)$: the $e_0$-minimal $\mathcal C$-embedding $e_{\min}$ of $u$ in $\doc$ does exist, but $e_{\min}(k) > e_0(k)$ for some $k \in [i]$ (i.e., no $\mathcal C$-embedding of $u$ in $\doc$ that is lower bounded by $e_0$ also agrees with $e_0$ on the first $i$ positions). These properties can be checked by one call of Algorithm~\ref{mainAlgo} with initial mapping $e_0$: If the algorithm returns $\bot$, then property $(1)$ holds, and if it returns some $\mathcal C$-embedding $e_{\min}$, we can simply check whether $e_{\min}(k) > e_0(k)$ for some $k \in [i]$ in order to check property $(2)$. 

In the following, for every $j \in \{i+1, \ldots, |u|\}$ 
and mapping $e$, 
let $e^{(j)}$ be obtained from~$e$
by moving the image of $j$ one position to the right. If $S(e_0, i) \neq \emptyset$ and $e_{\min}$ is the $e_0$-minimal $\mathcal C$-embedding of $u$ in $\doc$ (which has been computed by one call to Algorithm~\ref{mainAlgo}), then it can be shown that $S(e_0, i)$ is the \emph{disjoint} union of $\{e_{\min}\}$ and all the sets $S(e_{\min}^{(i+1)},i), S(e_{\min}^{(i+2)},i+1), \ldots, S(e_{\min}^{(|u|)},|u|-1)$. This entails a recursive procedure to compute the set $S(e_0, 0)$, which equals $\query(\doc)$ if $e_0$ is chosen such that $e_0(i)=i$ for all $i\in[|u|]$.

We can now sketch our procedure. Initially, we are given the trivial mapping $e_0$ with $e_0(i)=i$ for all $i\in[|u|]$, and compute the $e_0$-minimal $\mathcal C$-embedding $e_{\min}$ of $u$ in $\doc$. Then, for every $i\in[|u|]$, we recursively determine the set of all $\mathcal C$-embeddings $e'$ with $e_{\min}<e'$, where $e'$ and $e_{\min}$ coincide on the first $i-1$ positions, and where $e'(i)>e_{\min}(i)$. Assume we are given some mapping $e$ and $j\in\{0,1,\ldots,|u|-1\}$, indicating that we are not allowed to move positions $1,\ldots,j$, and the $e$-minimal $\mathcal C$-embedding $e^\ast$ of $u$ in $\doc$ exists and coincides with $e$ on the first $j$ positions. Then, we output $e^\ast$, determine which parts (characterised by $(e^\ast)^{(i)}$ and $i\in[j+1,|u|]$) of the partition still contain embeddings, and recursively output the elements of all non-empty subsets. Recall that, as explained above, we can check whether a subset given by $(e^\ast)^{(i)}$ for some $i\in[j+1,|u|]$ is empty.

In order to obtain an enumeration algorithm with bounded delay, we have to be careful with respect to one aspect. It might be possible that some of the sets $S(e^{(j)}, j-1)$ with $j \in \{i+1, \ldots, |u|\}$ are empty, so if we start a recursive call for it this may lead to many chains of recursive calls that will terminate without a new result, which blows-up the delay. Thus, we need a procedure $\nextMove(e,i)$ that, given a $\mathcal C$-embedding $e$ of $u$ in $\doc$ and $i\in\{0,1,\ldots,|u|-1\}$, determines the smallest $i'\in[i+1,|u|]$ such that there is an $e^{(i')}$-minimal $\mathcal C$-embedding that does not increase any position $k\in[i'-1]$ of $e$. Thus, after having finished a recursive call for computing some $S(e_{\min}^{(i)},i-1)$, 
$\nextMove(e_{\min}, i)$
tells us for which $i' > i$ we should recursively compute the set $S(e_{\min}^{(i')},i'-1)$ next, i.e., it tells us that the computation of the sets $S(e_{\min}^{(i+1)},i), \ldots, S(e_{\min}^{(i'-1)},i'-2)$ can be skipped, since they are empty anyway. Such a $\nextMove$ procedure can be devised as follows: We simply use Algorithm~\ref{mainAlgo} (as explained above) to check emptiness for $S(e_{\min}^{(i+1)},i), S(e_{\min}^{(i+2)},i+1), \ldots$ until we find a non-empty such set (this obviously introduces a factor $O(|u|)$, which can be removed in the case that our constraint languages are both left- and right-convex; see explanations below). For the enumeration algorithm, it is also important that we can organise our recursive calls in such a way that the algorithm is tail-recursive. 
Let us now state our 
result
(see Appendix~\ref{sec:enumDetails} for details).

\begin{theorem}\label{theorem:compvariant}
Let $u\in\Sigma^\ast$ be a query string, let $\mathcal C$ be a set of left-convex gap-constraints for $u$, and let $\doc\in\Sigma^\ast$ with $|u|\leq|\doc|$. Then we can enumerate the set $\query(\doc)$ of all $\mathcal C$-embeddings of $u$ in $\doc$ with $O(|\doc| (|u| + \lVert \mathcal C\rVert))$ preprocessing time and $O(|u|\cdot |\doc| (|u| + \lVert \mathcal C\rVert))$ delay. (Note that this entails a total computation of $\query(\doc)$ in time  $O(|\query(\doc)|\cdot|u|\cdot|\doc| (|u| + \lVert \mathcal C\rVert))$.)
\end{theorem}

Setting $T$ to be the running time of Algorithm~\ref{mainAlgo}, we can therefore compute (or enumerate) $\query(\doc)$ in time $O(|\query(\doc)| \cdot |u| \cdot T)$ (or with $O(T)$ preprocessing time and $O(|u|\cdot T)$ delay). This leads to the natural question of whether the factor $|u|$ can be removed. We can show that this is indeed possible, if the constraint languages are all left- \emph{and} right-convex. This is not too unlikely, e.g., it is provided by the practically relevant length constraints.

The reason for this factor $|u|$ is that after having completed the recursive call for one set $S(e^{(i)},i-1)$, we need to know the smallest $i'$ with $i' > i$ such that $S(e^{(i')},i'-1)$ is non-empty so that we can start a recursive call for this set. We determine this by the $\nextMove$ procedure mentioned above. Now if all constraint languages are left- and right-convex, we can apply the following trick in order to implement the $\nextMove$ procedure more efficiently.

Recall that, given some $\mathcal C$-embedding $e$ of $u$ in $\doc$ and $i\in\{0,1,\ldots,|u|-1\}$, the procedure $\nextMove(e,i)$ determines the smallest $i'\in[i+1,|u|]$ such that there is an $e^{(i')}$-minimal $\mathcal C$-embedding that does not increase any position $k\in[i'-1]$ of $e$. 
Instead of checking this for each $i + 1, i + 2, \ldots$ individually, we construct a mapping $e'$ that agrees on its first $i$ positions with $e$ and has its remaining positions $i+1, i+2, \ldots, |u|$ pushed all the way to the right, i.e., $e'(j) = |\doc| - (|u| - j)$ for every $j \in [i + 1, |u|]$. Then, we compute a $\mathcal C$-embedding $e_{\max}$ of $u$ in $\doc$ that is \emph{$e'$-maximal}. For this, we simply apply a \emph{reversed} version of Algorithm~\ref{mainAlgo} (note that the concept of $e'$-maximality of $\mathcal C$-embeddings and the uniqueness Lemma~\ref{uniquenessLemma} apply analogously). Obviously, this reversed version of Algorithm~\ref{mainAlgo} is correct due to the right-convexity of the constraint languages. This $e'$-maximal $\mathcal C$-embedding of $u$ in $\doc$ determines the desired $i'$ as follows. 

Since $e_{\max}$ is $e'$-maximal, we have that $e_{\max} \leq e'$. Moreover, $e \leq e'$ and $e$ is a $\mathcal C$-embedding of $u$ in $\doc$, which means that $e$ itself is also a candidate for an $e'$-maximal $\mathcal C$-embedding of $u$ in $\doc$ and therefore $e \leq e_{\max}$. Consequently, $e \leq e_{\max} \leq e'$, which means that $e_{\max}(j) = e'(j)$ for every $j \in [i]$, since, by definition, we have that $e(j) = e'(j)$ for every $j \in [i]$.
Let now $k$ be the smallest position of $[i + 1, |u|]$ such that $e_{\max}(k) > e(k)$. This means that there must be an $e^{(k)}$-minimal $\mathcal C$-embedding of $u$ in $\doc$ that does not increase any position $j \in [k-1]$ of $e$. Moreover, since $k$ is chosen minimal, we can conclude that $k$ is our desired position $i'$. On the other hand, if there is no such $k$, then $e_{\max}=e$ and there can be no $\mathcal C$-embedding $e^\ast$ of $u$ in $\doc$ with $e<e^\ast$ that agrees with $e$ on the first $i$ positions. Thus, $S(e,i)=\{e\}$ and there are no further subsets left to consider.
In particular, we need only one call to Algorithm~\ref{mainAlgo} for this procedure.

Adapting our procedure for Theorem~\ref{theorem:compvariant} with this trick allows us to prove the following improvement in the case where all constraint languages are left- and right-convex (full details can be found in Appendix~\ref{sec:enumDetails}):

\begin{restatable}{theorem}{lrconEnumRest}\label{thm:lrconEnum}
Let $u\in\Sigma^\ast$ be a query string, let $\mathcal C$ be a set of gap-constraints for $u$ that are both left- and right-convex, and let $\doc\in\Sigma^\ast$ with $|u|\leq|\doc|$. Then we can enumerate the set $\query(\doc)$ with $O(|\doc| (|u| + \lVert \mathcal C\rVert))$ preprocessing time and $O(|\doc| (|u| + \lVert \mathcal C\rVert))$ delay. (Note that this entails a total computation of $\query(\doc)$ in time $O(|\query(\doc)|\cdot |\doc| (|u| + \lVert \mathcal C\rVert))$.)
\end{restatable}

\section{Hardness Results for Subsequence Matching with
Non-Left-Convex Languages}\label{sec:hardness}

Having presented our algorithm for subsequence matching with left-convex
constraint languages, we now present complexity lower bounds in the case of
non-left-convex languages.
It is known that the matching problem is NP-complete in general (see~\cite{manea2024generalisedgaps}). 
However, the reduction from~\cite{manea2024generalisedgaps} requires complicated non-left-convex constraint languages. 
In the following, we will demonstrate that if we even slightly deviate from the left-convex setting of the previous section, then we will obtain intractability again. More precisely, even if we only allow length constraints (which are rather simple left-convex gap-constraints) and in addition $\{\ta \ta, \eword\}$ as the only non-left-convex\footnote{This language is indeed not left-convex: for $u = w = \ta$ and $v = \eword$, we have $uvw = \ta \ta \in L$ and $v = \eword \in L$, but $uv = \ta \notin L$.} gap-language, 
then we will show that subsequence matching is again NP-complete. Note that the language $\{\ta \ta, \eword\}$ could be considered as one of the smallest non-left-convex languages, since every language with strictly smaller cardinality or smaller maximum word size is necessarily left-convex.%

Let us first introduce some terminology. For a fixed language $L$, the \emph{matching problem with $L$-constraints} is the restricted case of the matching problem where the input query string $u$ and regular gap-constraints $\mathcal{C}$ satisfy $\mathcal{C} \subseteq \{(i, j, L) \mid i, j \in [|u|]\}$, i.e., all gap-constraints are $L$-constraints. For the \emph{matching problem with length constraints and $L$-constraints}, we additionally allow length constraints, i.e., we consider the restricted case of the matching problem where the input is assumed to obey the restriction that $\mathcal{C} \subseteq \{(i, j, L), (i, j, [\ell, r]) \mid i, j \in [|u|], \ell, r \in \{0, 1, \ldots, |\doc|\}\}$, where $\doc$ is the input document. In the rest of this section, we assume that $\Sigma = \{\ta, \tb\}$.

We first show that the matching problem is NP-complete, even if every constraint from $\mathcal{C}$ is either a length constraint or an $L$-constraint with $L = \{\ta \ta, \eword\}$.

\begin{restatable}{theorem}{aaLangThm}
\label{aaepsLanguageHardnessTheorem}
The matching problem with length constraints and $\{\ta \ta, \varepsilon\}$-constraints is NP-complete.
\end{restatable}

From this result we can also derive other interesting lower bounds. A main
ingredient of our reduction (see Appendix~\ref{sec:hardnessDetails}) is that we
can use the constraint language $\{\ta \ta, \eword\}$ to enforce that $\ta \ta$
is embedded into $\ta \ta \ta \ta$ either by mapping to the middle two
$\ta$-occurrences or to the outer two $\ta$-occurrences. However, by a slightly
more involved construction, we can show that using the language
$\{\ta\tb,\varepsilon\}$ instead of $\{\ta\ta,\varepsilon\}$ also works. 
This implies that the matching problem becomes intractable even when we restrict instances to require that each constraint language $L$ is either left-convex or right-convex, or is a length constraint.
Indeed, since $\{\ta\tb, \ta, \varepsilon\} \cap \{\ta\tb, \tb, \varepsilon\} = \{\ta\tb, \varepsilon\}$, we can ``simulate'' a constraint $(i, j, \{\ta\tb, \varepsilon\})$ (which is neither left- nor right-convex, and leads to intractability) by a left-convex constraint $(i, j, \{\ta\tb, \ta, \varepsilon\})$ and a right-convex constraint $(i, j, \{\ta\tb, \tb, \varepsilon\})$. 

Our reductions heavily use the property that we can use different constraints on different pairs of positions of $u$. This leads to the question of whether we could regain tractability by disallowing length constraints and only allowing one single regular language for all constraints in $\mathcal{C}$. By a non-trivial modification of the reduction from Theorem~\ref{aaepsLanguageHardnessTheorem} we can answer this question in the negative: there is a fixed (non-left-convex) regular language $L$ over $\Sigma = \{\ta, \tb\}$ such that the matching problem with only $L$-constraints is NP-complete even in the absence of length constraints.

\section{Conclusions and Future Work}

We have shown that using left-convex (or right-convex) languages as constraint languages in the context of CER based on subsequence matching entails efficient matching and enumeration algorithms. Moreover, our hardness results suggest that even using simple non-left-convex languages can lead to intractability. What is missing is a full dichotomy in the sense of finding a language class such that subsequence matching is tractable if and only if we use constraint languages from that language class. We believe that this is a rather challenging research task. Following the CSP approach of  Section~\ref{sec:CSP}, one direction could be to study whether more general conditions on gap-constraint languages can ensure the tractability of the CSP-instances that we construct.

Other than that, our algorithms are easy to implement and have running times that make them competitive also in a practical scenario. Consequently, prototype implementations and an experimental analysis would be a possible future research task. Furthermore, it would be interesting to investigate whether our general algorithmic approach can be adapted to query languages for complex event processing tailored to more practical scenarios (e.g.~\cite{GarciaR25}).

\begin{acks}
Antoine Amarilli is partially supported by the ANR project ANR-25-CE23-1215 (“Expand”).
Florin Manea and Tina Ringleb are supported by the German Research Foundation (DFG) through the Heisenberg project 466789228 and the research project 562495345. Markus L. Schmid is supported by the German Research Foundation (Deutsche Forschungsgemeinschaft, DFG) – project number 522576760 (gefördert durch die Deutsche Forschungsgemeinschaft (DFG) – Projektnummer 522576760).
\end{acks}

\bibliographystyle{ACM-Reference-Format}
\bibliography{main}

\appendix

\section{More Information About the Left-Convex Property}\label{sec:leftConvexProp}

In this section, we explore in more detail the class of left-convex languages for which we have presented an efficient algorithm. We show that left-convex languages serve as a unifying explanation of the tractability of many classes of regular languages that are interesting on their own. Then, we discuss the closure properties of the class of left-convex languages, showing that they are closed under intersection but not under union. We last discuss the complexity of testing whether a language is left-convex: we show that it is PSPACE-complete given an NFA, but that it is in polynomial time when the input is a deterministic finite automaton (DFA).

Recall the definition of the left-convex property given in the Introduction (Definition~\ref{def:lcon}). Analogously, we define a language $L$ to be \emph{right-convex} if $uvw\in L$ and $v\in L$ implies $vw \in L$. We shall use $\RCON$ to denote the class of right-convex languages. Also note that $L$ is right-convex if and only if $L^R$ is left-convex. 

\subsection{Examples of Left-Convex Languages}
We overview here more example languages which have the left-convex property (the proofs
that the respective classes of languages have this property are immediate and
omitted):
\begin{itemize}[nosep]
    \item Infix-closed (respectively, prefix-closed) languages,  i.e., $v \in L \implies u \in L$ for all proper infixes (respectively, prefixes) $u$ of $v$.
    \item Infix-free 
     languages, i.e., $v \in L \implies u \notin L$ for all proper infixes $u$ of $v$. 
    \item Finite languages where all words have the same length, i.e.,
      $L=\{v_1,v_2,\dots,v_k \mid |v_1|=|v_2|=\ldots=|v_k|\}$. Indeed, these are a special case of
      infix-free languages.
    \item Languages with unique start or end delimiter, i.e., $L \subseteq
    \{xv \mid v \in \Sigma^\ast\}$ or $L \subseteq \{vx \mid v \in
    \Sigma^\ast\}$, where $x \notin \Sigma$.
    \item Languages closed under left-extension, i.e., $v \in L
      \Rightarrow uv \in L$ for all $u, v \in \Sigma^*$.

        \item Languages defined by a forbidden set of infixes and a required set of infixes, i.e., given sets of words $F, R \subseteq \Sigma^*$, $L = (\Sigma^* \setminus \Sigma^* F \Sigma^*) \cap (\Sigma^* R
      \Sigma^*)$, for some alphabet $\Sigma$. \\
      Note that we can additionally add a set of forbidden prefixes and a set of required suffixes and the languages thus defined are still left-convex.
    \item Downward and upward closures of any language. \\
    Recall that the downward and upward closure of a language $L$ is the set of all subsequences (respectively, supersequences) of the words of $L$.
    \item Length constraints, i.e., languages $L = \{w \mid x
      \leq |w| \leq y\}$ for some integers $x \leq y$.
\item Languages defined by upper and lower bounding the number of occurrences of certain symbols, e.g., $L = \{w \mid 3 \leq |w|_{\ta} \leq 7 \wedge 5 \leq |w|_{\tb} \leq 23\}$, where $|w|_{\ta}$ denotes the number of occurrences of $\ta$ in $w$.
    \item Length constraints with neutral letters from some subalphabet $\Gamma \subseteq \Sigma$ that do not contribute to the length of the gaps. More formally, languages $L = \{w \mid x \leq \sum_{\ta \in \Sigma\setminus\Gamma} |w|_{\ta} \leq y\}$ for some $x, y \in \mathbb{N}$ with $x \leq y$. 
\end{itemize}

Note that prefix- and suffix-free languages are not necessarily left-convex. For instance, the language $L=\{\ta\tb\ta,\tb\}$ is both prefix-free and suffix-free, but is not left-convex, as for $u=\ta, v=\tb, w=\ta$ we have that $uvw=\ta\tb\ta \in L$, $v = \tb \in L$, but $uv=\ta\tb\notin L$. 

Also, languages defined by a required set of prefixes ($L= S \Sigma^*$ for some set $S\subseteq \Sigma^*$, over alphabet $\Sigma$) are not necessarily left-convex. For instance, for $S=\{\ta\tb\ta,\tb\}$, let $u=\ta$, $v=\tb$, $w=\ta$, so $uvw=\ta \tb \ta \in L$, $v = \tb \in L$, but $uv=\ta\tb\notin L$. With respect to left-convex languages that are not right-convex (or right-convex languages that are not left-convex), we observe that any language of the form $S \Sigma^*$ for some $S \subseteq \Sigma^*$, even if it is not left-convex, is right-convex. Thus, the language defined by $S=\{\ta\tb\ta,\tb\}$ is an example of a right-convex language that is not left-convex (and reversing it yields an example of a left-convex language that is not right-convex). 

The language $\{\ta\ta, \eword\}$ is a very simple example of a language that is neither left- nor right-convex. This can be easily extended to classes of languages that are neither left- nor right-convex, e.g., languages of the form $\{w, \eword\}$ with $|w| \geq 2$, or languages of the form $L=\{w\in \{\ta,\tb\}^*\mid |w|_{\ta}$ is divisible by $k \}$ and $k\geq 2$.

\subsection{Closure Properties}
The class of left-convex languages (or right-convex languages) is obviously not closed under union, because the singleton languages $\{\eword\}$ and $\{\ta\ta\}$ are trivially both left-convex and right-convex but their union $\{\eword, \ta\ta\}$ is neither. However,
the class of left-convex languages (as well as the class of right-convex languages) is easily seen to be closed under intersection; this is particularly interesting in the context of this paper, as it shows that multiple left-convex gap-constraints $(i, j, L_1), (i, j, L_2), \ldots, (i, j, L_m)$ for the same pair $(i, j)$ of positions could always be interpreted as a single left-convex gap-constraint $(i,j, \bigcap^m_{k = 1} L_k)$. However, obtaining a representation of the language $\bigcap^m_{k = 1} L_k$ from the representations of the languages $L_k$ with $k \in [m]$ is usually computationally inefficient. Thus, it is a relevant feature of our algorithm that it can handle the situation of several left-convex constraints for the same pair of positions without the necessity of computing a representation of the intersection.

\subsection{Membership Testing}
Let us next consider the problem of checking whether a given regular language is left-convex.

\begin{theorem}\label{LeftConvexityCheckTheorem}
Checking whether a given NFA accepts a left-convex language is PSPACE-complete, but it can be done in polynomial time for a given DFA.
\end{theorem}

\begin{proof}
Let us first start with a proof for the DFA-case, since we shall use parts of it for the NFA-case later on. 

Let $M$ be a DFA. We observe that $\mathcal{L}(M)$ is not left-convex if $M$ accepts a word $x$ that can be factorised into $x = uvw$ such that $M$ accepts $v$, but not $uv$. This is the case if and only if there is a path from the start state $q_0$ to an accepting state $q_f$ such that:
\begin{itemize}[nosep]
\item there are states $p_1$ and $p_2$ that lie on this path in that order,
\item $p_2$ is non-accepting,
\item the word $v$ read between $p_1$ and $p_2$ is accepted by $M$ (i.e., there is a $v$-labelled path from $q_0$ to some accepting state).
\end{itemize}
For the if-direction, note that if we denote the word read between $q_0$ and $p_1$ as $u$, the word read between $p_1$ and $p_2$ as $v$, and the word read between $p_2$ and $q_f$ as $w$, then we have that $uvw \in \mathcal{L}(M)$ and $v \in \mathcal{L}(M)$, but $uv \notin \mathcal{L}(M)$ (since $p_2$ is not accepting). For the only-if-direction, assume that $M$ accepts a word $x = uvw$ such that $M$ accepts $v$, but not $uv$. Then we take the accepting path of $M$ on $x$ and let $p_1$ and $p_2$ be the states reached after reading $u$ and $uv$, respectively. Obviously, the conditions stated above are satisfied.

We can therefore solve the problem as follows: In polynomial time, we can make sure that every state of $M$ is reachable from the start state and that there is exactly one state that cannot reach an accepting state, which we call the trap state. For every pair $(p_1, p_2)$ of states such that neither $p_1$ nor $p_2$ are the trap state and $p_2$ is not accepting, we check whether there is a string $v$ that brings $M$ from $p_1$ to $p_2$ and from $q_0$ to an accepting state. To this end, we construct the DFA $M_{p_1, p_2}$ that accepts all words that $M$ can read between $p_1$ and $p_2$, and then the DFA $M_{\cap}$ that accepts $\mathcal{L}(M) \cap \mathcal{L}(M_{p_1, p_2})$. Finally, we only have to check emptiness for $M_{\cap}$.

This algorithm can clearly be carried out in polynomial time. 

Let us now come to the NFA-case and let us start with the upper bound, i.e., membership to PSPACE. Let $M$ be an NFA with state set $Q$ and let $M_{\mathsf{D}}$ be the DFA obtained from $M$ by the subset construction (obviously, we cannot afford to explicitly construct $M_{\mathsf{D}}$), and let us also assume that the state set $Q_{\mathsf{D}}$ of $M_{\mathsf{D}}$ is the full power set of $Q$. We can now check whether $\mathcal{L}(M)$ is left-convex, by checking whether $\mathcal{L}(M_{\mathsf{D}})$ is left-convex as described above in the DFA-case. Hence, we have to explain how we can do this nondeterministically in space polynomial in $|M|$. 

We first guess two states $P_1$ and $P_2$ of $M_{\mathsf{D}}$ (recall that states of $M_{\mathsf{D}}$ are just subsets of $M$'s states) such that $P_2$ is not an accepting state of~$M_{\mathsf{D}}$. Next, we check whether there is a word that brings $M_{\mathsf{D}}$ from the start state to the state $P_1$. To this end, let us first observe that there is such a word if and only if there is such a word of length at most $2^{|Q|}$ (due to standard pumping arguments). Hence, we can guess a word of length $2^{|Q|}$ symbol by symbol and in parallel simulate $M_{\mathsf{D}}$ on this word, and we stop as soon as we reached $P_1$ or the word length has reached $2^{|Q|}$. In each step of this simulation, we only have to store $O(|Q|)$ states of $M$ and a number that is at most $2^{|Q|}$ (to keep track of the word length). Since we can store a number that is at most $2^{|Q|}$ with $|Q|$ bits, we only require space polynomial in $|Q|$. In a similar way, we can nondeterministically check whether there is a word that brings $M_{\mathsf{D}}$ from $P_2$ to some accepting state.  

Now, it only remains to check whether there is a word $v$ that brings $M_{\mathsf{D}}$ from $P_1$ to $P_2$ and also from the start state to some accepting state. We first observe that there is such a word if and only if there is such a word of length at most $|Q_{\mathsf{D}}|^2$ (recall that $Q_{\mathsf{D}}$ is the state set of $M_{\mathsf{D}}$). Indeed, this follows again from a pumping argument, as we now explain. If the shortest such word $v$ satisfies $|v| > |Q_{\mathsf{D}}|^2$ and there is a $v$-path $(S_0, S_1, \ldots, S_{|v|})$ from the start to an accepting state and a $v$-path $(T_0, T_1, \ldots, T_{|v|})$ from $P_1$ to $P_2$, then $(S_i, T_i) = (S_j, T_j)$ for some $i, j$ with $0 \leq i < j \leq |v|$. This means that $(S_0, \ldots, S_i, S_{j+1}, \ldots, S_{|v|})$ is a path from the start to an accepting state and $(T_0, \ldots, T_i, T_{j+1}, \ldots, T_{|v|})$ is a path from $P_1$ to $P_2$, and the label of both paths is identical. This contradicts the minimality of $v$. Hence, we only have to check if there is a word $v$ of size at most $|Q_{\mathsf{D}}|^2 = (2^{|Q|})^2$ that brings $M_{\mathsf{D}}$ from $P_1$ to $P_2$ and also from the start to an accepting state. This can be done as follows: We guess a word of length $(2^{|Q|})^2$ symbol by symbol and in parallel simulate $M_{\mathsf{D}}$ on this word starting in the start state and also simulate $M_{\mathsf{D}}$ on this word starting in $P_1$. We stop as soon as we reach an accepting state with the first simulation and $P_2$ with the second simulation, or when the word length has reached $(2^{|Q|})^2$. Again, we only require space that is polynomial in $|Q|$, since we have to store $O(|Q|)$ states of $M$ and a number of size at most $(2^{|Q|})^2$. This concludes the proof that the problem is in PSPACE.

Let us next come to the lower bound, i.e., that the problem is PSPACE-hard. We prove this by reducing from the PSPACE-hard universality problem for NFAs, i.e., the problem to decide $\mathcal{L}(M) = \Sigma^*$ for a given NFA $M$ over $\Sigma$. 

Let $M$ be an arbitrary NFA that accepts a language $L$ and we want to decide whether $L = \Sigma^*$. We define the regular language $K = \{\# x \# \mid x \in L\} \cup \{\# x \# \# \mid x \in \Sigma^*\} \cup \{\#\}$, where $\#$ is a fresh letter not in $\Sigma$, and we note that an NFA for $K$ can be easily constructed in polynomial time from~$M$. We now claim that $K$ is not left-convex if and only if $L \neq \Sigma^*$.

Let us first prove the if-direction and assume that $L \neq \Sigma^*$, which means that there is some $x \in \Sigma^*$ with $x \notin L$. We will now consider the word $\# x \# \#$ and its factorisation $u = \# x$, $v = \#$ and $w = \#$. Obviously, $uvw \in K$ and $v \in K$, but $uv = \# x \# \notin K$ because $x \notin L$. Consequently, $K$ is not left-convex.

For the only-if-direction, let us assume that $K$ is not left-convex, which means that there is some $uvw \in K$ with $v \in K$, but $uv \notin K$. Since $uvw \in K$, we have that $uvw \in \{\# x \# \mid x \in L\}$ or $uvw \in \{\# x \# \# \mid x \in \Sigma^*\}$ or $uvw \in \{\#\}$. We observe that $uvw \in \{\# x \# \mid x \in L\}$ is not possible, since then $v \in K$ implies that $v$ is the first or last $\#$-occurrence or $v = uvw$, which contradicts $uv \notin L$. Likewise, $uvw = \#$ is not possible while ensuring 
$v \in K$ and $uv \notin K$. Hence, we must have $uvw = \# x \# \#$ for some $x \in \Sigma^*$. 
If $v$ is a prefix or a suffix, we would have $v = uv$ or $uv = uvw$, respectively, which would contradict $uv \notin K$. Since $v \in K$, we conclude that $v$ must be the second occurrence of $\#$. Thus, $u = \# x$, $v = \#$ and $w = \#$. Therefore, $uv \notin K$ implies $\# x \# \notin K$, which means that $x \notin L$. Consequently, $L \neq \Sigma^*$. 
\end{proof}

\section{Full Details for Section~\ref{sec:CSP}}\label{CSPAppendix}

In this section, we discuss in more details the formulation of the subsequence matching with gap-constraints problem as a Constraint Satisfaction Problem (for short, CSP), and how this leads to a series of interesting results. As mentioned in the main part of the paper, this allows us to show that the respective problem can be solved in polynomial time when left-convex regular constraints (e.g., length constraints) are considered. 

Let us first introduce the CSP-related machinery that we will use in this section; we follow the definitions from \cite{CohenJ06}, which extend the terminology from the main part of the paper to a more general setting. This setting allows us a more precise and detailed presentation compared to that from the main part.

Let $\domain$ be a set (called domain in the following) and let $R_\domain$ denote the set of all finitary relations over $\domain$ (subsets of $\domain^k$ for any $k\geq 1$). A {\em constraint language} $\Gamma$ over $\domain$ is a subset of $R_\domain$.

We also define the set $O_\domain$ of all mappings $f:\domain^k \rightarrow \domain$, for all $k\geq 1$. A $k$-ary operation $f\in O_\domain$ is called a \emph{weak near-unanimity operation} (WNU, for short) if it satisfies $f(y, x,\ldots , x) = f(x, y, x, \ldots , x) = \cdots  = f(x, x,\ldots , x, y)$ for all $x, y \in \domain$; $f$ is called \emph{idempotent} if $f(x, x,\ldots , x) = x$ for all $x \in\domain$. A binary mapping $f\in O_\domain$ which satisfies the identities $f(x, f(y, z)) = f(f(x, y), z)$ (associativity), $f(x, y) = f(y, x)$ (commutativity), and which is idempotent is called a \emph{semilattice operation} over $\domain$. 

For  $f:\domain^k \rightarrow \domain$ and tuples $t_1, \ldots , t_k\in \domain^n$, we define $f(t_1,\ldots , t_k) \in \domain^n$ to be the $n$-tuple
$ \langle f(t_1[1], \ldots, t_k[1]), \ldots , f(t_1[n], \ldots, t_k[n])\rangle$. We say that $f:\domain^k \rightarrow \domain$ \emph{preserves} an $n$-ary relation $R$ over $\domain$ (or $f$ is a {\em polymorphism} of $R$) if $f(t_1, \ldots, t_k) \in R$ for all $n$-tuples $t_1, \ldots, t_k \in R$. For some constraint language $\Gamma$ over $\domain$, we define $\Pol(\Gamma)= \{f \in O_\domain \mid f \mbox{ preserves each relation from }\Gamma\}$. 

Now, for any domain $\domain$ and any constraint language $\Gamma$ over $\domain$, the constraint satisfaction problem $\CSP(\Gamma)$ is the decision problem which has as instance (as input) a triple ${\mathcal P}=\langle V, \domain, \const\rangle$, where $V=\{1,\ldots, m\}$ is a set of variables (denoted, for simplicity, with the first $m$ natural numbers) and $\const$ is a set of constraints $\{\const_1, \ldots,\const_q\}$, such that each constraint $\const_i\in \const$ is a pair $\langle s_i,R_i\rangle$, where $s_i\subseteq V$ is a set of variables of size $n_i$ (called the \emph{constraint scope}) and $R_i \in \Gamma$ is an $n_i$-ary relation over $\domain$, called the \emph{constraint relation}; w.l.o.g., we assume that $s_i$ is increasingly ordered, and by $s_i[j]$ we denote the $j^{th}$ element of $s_i$.  An instance ${\mathcal P}=\langle V, \domain, \const\rangle$ of $\CSP(\Gamma)$  is called \emph{normalised} if there do not exist two different constraints in $\const$ that have the same scope. A solution for an instance $\mathcal P$ of $\CSP(\Gamma)$ is a function $\phi: V\rightarrow \domain$ such that for each $\langle s_i,R_i\rangle \in \const$ the tuple $\langle \phi (s_i[1]), \ldots, \phi(s_i[n_i])\rangle$ is in $R_i$. 
Accordingly, for some instance ${\mathcal P}$ of $\CSP(\Gamma)$, we are interested in whether it admits a solution. 

Note that, in this definition of CSPs, the problems are parameterised by the constraint language $\Gamma$, which restricts the type of allowed constraints. A constraint language $\Gamma$ is said to be \emph{tractable} if $\CSP(\Gamma')$ can be solved in polynomial time, for each finite subset $\Gamma'\subseteq \Gamma$. For the sake of completeness, let us mention also that a constraint language $\Gamma$ is said to be \emph{NP-complete} if $\CSP(\Gamma')$ is NP-complete, for each finite subset $\Gamma'\subseteq \Gamma$. 

In \cite{JeavonsCG97,Jeavons98} it was shown that the tractability of $\CSP(\Gamma)$ depends on the algebraic properties of the set $\Pol(\Gamma)$ of operations which preserve all relations of $\Gamma$. Consequently, in \cite{CohenJ06} and the references therein, a series of results are overviewed, where one is interested in the complexity of solving $\CSP(\Gamma)$ when the constraint language $\Gamma$ fulfils a series of structural and algebraic properties. Finally, in a breakthrough result, a complete characterisation of the type of constraints (i.e., the constraint languages) for which the corresponding CSPs are solvable in polynomial time was given in \cite{Zhuk20} (see also \cite{Zhuk17,Bulatov17}).
\begin{theorem}[\cite{Zhuk20}]\label{thm:dichotomyCSP}
For any constraint language $\Gamma$ over a finite domain $\domain$, if $\Pol(\Gamma)$ contains a WNU, then $\Gamma$ is tractable. Otherwise, $\CSP(\Gamma)$ is NP-complete.
\end{theorem}

In particular, the next proposition follows now immediately.
\begin{proposition}[Proposition 6.37 from \cite{CohenJ06}, originally from \cite{JeavonsCG97}]\label{prop:semilatticeCSP}
For any constraint language $\Gamma$ over a finite domain $\domain$, if $\Pol(\Gamma)$ contains a semilattice operation, then $\Gamma$ is tractable.
\end{proposition}

Indeed, semilattice operations are particular WNUs. Interestingly (see \cite{Bessiere06,CohenJ06}), the tractability of $\Gamma$, under the conditions of Proposition~\ref{prop:semilatticeCSP} and moreover assuming that there is a constant upper bound on the arity of the constraint relations, follows from the fact that all instances in $\CSP(\Gamma)$ can be solved in polynomial time by a class of concrete algorithms which implement an algorithmic technique called {\em enforcing generalised arc consistency} (these algorithms are significantly less complex compared to the general algorithm from \cite{Zhuk20}, but the degree of the polynomial describing their runtime depends on the aforementioned arity-bound). In particular, we recall the algorithms AC4 (defined in \cite{YuanlinY01}, as a more efficient implementation of AC3 \cite{Mackworth77}), AC2001/3.1 or GAC2001/3.1 \cite{BessiereRYZ05}, which solve a normalised instance ${\mathcal P}=\langle V,\const,\domain\rangle$, which only contains unary or binary constraints, in $O(|\const| |\domain|^2)$ time (which is also optimal~\cite{Bessiere06}).

Further, there are quite a few interesting examples of constraint languages $\Gamma$ such that $\Pol(\Gamma)$ contains a semilattice operation; see \cite{CohenJ06}. 
Relevant for our paper is that for any constraint language $\Gamma$ which is a collection of \emph{min-closed} (or \emph{max-closed}) relations, $\CSP(\Gamma)$ can be solved in polynomial time. A relation is min-closed (respectively, max-closed) if $\min\{\cdot,\cdot\}$, the binary operation computing the minimum of its two arguments, is a polymorphism of $R$ (respectively, $\max\{\cdot,\cdot\}$, the binary operation computing the maximum of its two arguments, is a polymorphism of $R$). In more detail, this means that a $k$-ary relation $R$ is min-closed, if for every pair of $k$-tuples $t,t'\in R$ we have that the $k$-tuple $\langle \min\{t[1],t'[1]\}, \ldots, \min\{t[k],t'[k]\}\rangle \in R$ (max-closed relations can also be defined analogously). Clearly, if $\Gamma$ is a collection of \emph{min-closed} (or \emph{max-closed}) relations, then $\min$ (respectively, $\max$) can be used as the semilattice operation from the statement of Proposition~\ref{prop:semilatticeCSP}, so we know from that proposition that $\Gamma$ is tractable.
As shown in \cite{JeavonsC95}, there are many relevant classes of min-closed and max-closed constraint languages. For instance, all unary constraints are min-closed, and all basic arithmetic constraints over the natural numbers in the constraint programming language CHIP \cite{HentenryckDT92} are min- and max-closed; such relations include, e.g., those relations that are defined by linear integer (in)equations on the set of variables. \looseness=-1

After this long introduction, let us come back to the subsequence matching with gap-constraints problem, and see how this can be formalised as a CSP-instance. Recall that we are given a \emph{query string} $u \in \Sigma^*$, $|u|=m$, along with a set $\mathcal{C}$ of regular gap-constraints for $u$, and a \emph{document} $\doc \in \Sigma^*$, $|\doc|=n$, and we want to see whether there is a $\mathcal{C}$-embedding of $u$ in $\doc$. We define $V=\{1,\ldots, m\}$ and $\domain=\{1,\ldots,n\}$. The constraints $\const'$ are defined as follows:
\begin{itemize}
\item For every $i\in [m]$, we define a unary constraint $\langle (i) , R_i\rangle$, where $R_i=\{ a\in [n]\mid  \doc[a]=u[i]\}$. 
\item For every $i\in [m-1]$, we define a binary constraint $\langle (i,i+1) , R_{(i,i+1)}\rangle$, where $R_{(i,i+1)} =\{(a,b)\in [n]\times [n] \mid a<b\}$. 
\item For every gap-constraint $(i,j,L)\in {\mathcal C}$, we define a binary constraint $\langle (i,j) , R_{(i,j,L)}\rangle$, where $ R_{(i,j,L)} =\{(a,b) \in [n]\times[n]\mid  \doc[a+1:b-1]\in L\}$. 
\end{itemize}
Let $\CSP'_{u,\doc,{\mathcal C}}=\langle V,\domain,\const'\rangle$ be the CSP-instance defined for the instance $u,\doc, {\mathcal C}$ of the subsequence matching with gap-constraints problem. It is immediate that this CSP-instance has a solution if and only if there exists a ${\mathcal C}$-embedding of $u$ in $\doc$. Note that, while having only unary and binary constraints, this CSP-instance is not necessarily normalised (as there can be multiple constraints for the same pair of variables).
Fortunately, it can be shown that in $O(\lVert\mathcal C \rVert |\doc|^2)$ time we can construct $\CSP'_{u,\doc,{\mathcal C}}=\langle V,\domain,\const'\rangle$ and then turn it into an equivalent normalised CSP-instance 
by efficiently intersecting the constraints having the same scope (the same proof obviously works for Lemma~\ref{lem:RegularMembershipSubstr}):

\begin{lemma}[Lemma~\ref{lem:RegularMembershipSubstr}, restated in the more general setting]
Given $u, \doc,$ and ${\mathcal C}$, we can construct in $O(\lVert{\mathcal C} \rVert |\doc|^2)$ time the CSP-instance $\CSP'_{u,\doc,{\mathcal C}}=\langle V,\domain,\const'\rangle$. In the same time complexity, we can construct $\CSP_{u,\doc,{\mathcal C}}=\langle V,\domain,\const\rangle$, a normalised CSP-instance with $O(|u|+|{\mathcal C}|)$ constraints, which has a solution if and only if there exists a ${\mathcal C}$-embedding of $u$ in $\doc$.\looseness=-1
\end{lemma}

\begin{proof}
Let $n=|\doc|$ and $m=|u|$. To show the first part of our claim, let us recall our assumption that each regular language $L$ appearing in ${\mathcal C}$ is given as an $\eword$NFA. Thus, for each constraint $(i,j,L)\in {\mathcal C}$ and each position $a$ of $\doc$, we run (in the standard way, sometimes called state-set simulation) the $\eword$NFA accepting $L$ with the suffix $\doc[a:n]$ of $\doc$ as input string: we start in the initial state of the respective $\eword$NFA, and maintain the set of reached states after the substring $\doc[a:b]$ was processed. Whenever we reach a position $b$ such that the maintained state-set contains a final state, we add $(a-1,b+1)$ to the set $ R_{(i,j,L)} $ (defined above when $\CSP'_{u,\doc,{\mathcal C}}$ was introduced); for simplicity, the sets $ R_{(i,j,L)} $ are implemented as lists. This whole process (executed for all constraints and all suffixes of $\doc$) correctly constructs all binary constraints  $\langle (i,j) , R_{(i,j,L)}\rangle$ from the set $\const'$ of constraints. For a fixed position of $\doc$ and a fixed constraint $(i,j,L)\in {\mathcal C}$ this process described above takes $O(n \cdot s_L)$ time, where $s_L$ is the total size of the $\eword$NFA accepting $L$. Thus, the processing done for a fixed constraint and all positions of $\doc$ takes $O(n^2 s_L)$ time in total. 
Consequently, the time needed to process all constraints is $O(n^2 \lVert {\mathcal C}\rVert)=O(\lVert {\mathcal C} \rVert |\doc|^2)$. 

The other constraints can be trivially constructed in $O(nm)$ time. As $n\geq m$,  the total time needed to construct $\CSP'_{u,\doc,{\mathcal C}}=\langle V,\domain,\const'\rangle$ is $O(\lVert{\mathcal C} \rVert |\doc|^2)$. The total number of constraints in this CSP-instance is $O( |u| + |{\mathcal C}| )$.

In the next step, we normalise this CSP-instance. For each pair $(i,j)\in [m]\times [m]$ such that $\CSP'_{u,\doc,{\mathcal C}}=\langle V,\domain,\const'\rangle$ has at least one constraint $\langle (i,j), R_{(i,j,L)} \rangle$, we proceed as follows. 

Let $\langle (i,j), X_1 \rangle,\ldots, \langle (i,j), X_p \rangle$ be all constraints involving that pair $(i,j)$ (note that an $X_{\ell}$ can have the form $R_{(i, j, L)}$ or the form $R_{(i, j)}$ in the case where $j = i+1$). We define an $n\times n$ matrix $M$, whose elements are initialised with $0$; then, we traverse the lists $X_c$, with $c\in [p]$, and each time we meet a pair $(a,b)$ we increase $M[a,b]$ by $1$. In the end, we produce the list $\widehat R_{(i,j)}$ of all pairs $(a,b)$ such that $M[a,b]=p$, and define the new constraint $\langle (i,j), \widehat R_{(i,j)} \rangle$ as the single constraint involving the pair $(i,j)$. This process is equivalent to intersecting the lists $X_1,\ldots,X_p$, and takes $O(p\cdot n^2)$ time. 
So, running it on all constraints from $\CSP'_{u,\doc,{\mathcal C}}=\langle V,\domain,\const'\rangle$ produces a normalised CSP-instance, equivalent to $\CSP'_{u,\doc,{\mathcal C}}$, in time $O(|{\mathcal C}| n^2)$. The total number of constraints in this instance is still $O( |u| + |{\mathcal C}| )$.

The conclusion now follows immediately.
\end{proof}

Now, if $\CSP_{u,\doc,{\mathcal C}}$ is an instance of $\CSP(\Gamma)$ for some constraint language $\Gamma$ over $\domain$ such that $\Pol(\Gamma)$ contains a WNU, then $\CSP_{u,\doc,{\mathcal C}}$ can be solved in polynomial time. If $\Pol(\Gamma)$ contains a semilattice operation, then $\CSP_{u,\doc,{\mathcal C}}$ can be solved, by algorithms such as AC4, which enforce arc consistency, in $O((|u|+|\mathcal C|)|\doc|^2)$ time. 
In particular, this holds for the case when the constraints of $\CSP_{u,\doc,{\mathcal C}}$ are min-closed. We directly get the following result. 

\CSPsolutionRest*

This CSP-centred approach is interesting as it seems to provide the right level of abstraction for the fundamental understanding of the subsequence matching with gap-constraints problem. Due to the very particular type of the domain and of the constraints involved in the definition of $\CSP_{u,\doc,{\mathcal C}}$, the general results regarding CSP do not immediately provide a dichotomy regarding the tractability of this problem; however, they suggest that there is some hope in reaching such a result based on insights inherited from the general theory. 

Looking a bit closer at Theorem~\ref{thm:CSPsolution}, we also notice a downside. Our theorem assumes that the constraints of $\CSP_{u,\doc,{\mathcal C}}$ are min-closed. But, in fact, we can also check this property, in a straightforward manner, in polynomial time (once $\domain$ is determined by $\doc$). However, this does not fit our approach to (and motivation of) the subsequence matching problem: it would mean that if the constraints of $\CSP_{u,\doc,{\mathcal C}}$ are not min-closed, then this algorithm would not work, and this situation is only detected after the CSP is constructed. So, in the following we focus on identifying classes of regular languages which, when used as gap-constraints, guarantee that $\const$ is always min-closed (respectively, max-closed), irrespective of the document $\doc$. 

To this end, recall the definition of left-convex languages from Section~\ref{sec:intro} (Definition~\ref{def:lcon}): a language $L$ is left-convex if $uvw \in L$ and $v\in L$ implies $uv \in L$. Let $\LCON$ be the class of left-convex languages. We can immediately see 
that if all $(i, j, L) \in \mathcal{C}$ satisfy $L \in \LCON$, then all constraints of $\CSP_{u,\doc,{\mathcal C}}$ are min-closed.

\LCONmin*
 
\begin{proof}
The statement holds for the unary constraints by definition.
For a constraint $\langle (i,j),R\rangle$ of $\CSP_{u,\doc,{\mathcal C}}$, we have, on the one hand, that $R$ contains only pairs $(a,b)$ with $a<b$ and, on the other hand, that there exist left-convex languages $H_1,\ldots, H_p$ such that $(i,j,H_c)\in {\mathcal C}$ for all $c\in [p]$, and for each pair $(a,b)\in R$ we have that $\doc[a+1:b-1]\in L=\cap_{c\in [p]}H_c$.
As the class of left-convex languages is closed under intersection, $L$ is left-convex. 
(In particular, in the case where $j = i+1$ it may be the case that the constraint $\langle (i,j),R\rangle$ only arises from a constraint of the form $R_{(i,i+1)}$ and there might be no gap-constraint $(i,j,H_c)\in\mathcal C$. In this case, we have $p=0$ and $L=\Sigma^\ast$, which is left-convex.)

Now, if $(a, b), (c, d) \in R$, the situation is as follows. If the intervals $(a, b)$ and $(c, d)$ are completely disjoint or they properly overlap (i.e., $a \leq c \leq b \leq d$ or $c \leq a \leq d \leq b$), then $(\min\{a, c\}, \min\{b, d\}) \in \{(a, b), (c, d)\}$. If, on the other hand, $a \leq c < d \leq b$ (the case $c \leq a < b \leq d$ is analogous), then, setting $u = \doc[a + 1 : c]$, $v = \doc[c+1 : d - 1]$ and $w = \doc[d : b-1]$, we have that $uvw \in L$ and $v \in L$. Thus, due to the left-convexity of $L$, $uv \in L$ and $uv = \doc[a+1 : d-1] = \doc[\min\{a, c\} + 1 : \min\{b, d\} - 1]$. 
\end{proof}

Similarly, if $L$ is right-convex, then all constraints of $\CSP_{u,\doc,{\mathcal C}}$, the CSP-instance produced in Lemma~\ref{lem:RegularMembershipSubstr}, are max-closed. So, the result of Theorem~\ref{thm:CSPsolution} holds for the case when all gap-constraints in ${\mathcal C}$ are left-convex (respectively, right-convex). 

As mentioned in the introduction, the matching problem is NP-complete in general, and we have just shown that requiring left-convexity (or right-convexity) for the regular gap-constraint languages is sufficient for tractability. Interestingly, Appendix~\ref{sec:leftConvexProp} exhibits a plethora of left-convex languages which cover important practically relevant cases for the subsequence matching with gap-constraints problem, especially length-constrained languages. Also, Section~\ref{sec:hardness} shows that even slight deviations from the left-convex setting lead to intractability, thus suggesting that left-convexity (and right-convexity) might be a good property to consider in trying to show a dichotomy regarding the tractability of $\CSP_{u,\doc,{\mathcal C}}$.

\section{Full Details for Section~\ref{sec:mainAlgo}}\label{mainAlgoAppendix}

\subsection{Full Proof of Lemma~\ref{mainAlgoCorrectnessLemma}}

\mainAlgoCorrectnessRest*

\begin{proof}
For every $e \colon [|u|] \to [|\doc|]$, we define the following invariant.

\smallskip

\noindent \emph{Invariant $(\dagger)_{e}$}: If there exists a $\mathcal C$-embedding $e^*$ of $u$ in $\doc$ with $e_0 \leq e^*$, then $e \leq e^*$ (where $e$ is the current mapping of the algorithm). 

\smallskip

Let us first observe that before the first iteration, we have that $e = e_0$; thus, if there exists a $\mathcal C$-embedding $e^*$ of $u$ in $\doc$ with $e_0 \leq e^*$, then $e \leq e^*$ is trivially true. This shows that $(\dagger)_e$ is satisfied before the first iteration. Let us now assume that $(\dagger)_e$ is satisfied before an iteration of the while-loop. We will consider the if-statements of Lines~\ref{symbolCond},~\ref{orderCond},~and~\ref{condLineOne} separately. 

\paragraph{Line~\ref{symbolCond}:}
If the condition is satisfied, then we have that $u[s] \neq \doc[e(s)]$ and we will change $e$ into $e'$ with $e'(i) = e(i)$ for every $i \in [|u|] \setminus \{s\}$ and $e'(s) = e(s) + 1$. Since $(\dagger)_e$ is satisfied, we know that any $\mathcal C$-embedding $e^*$ of $u$ in $\doc$ with $e_0 \leq e^*$ satisfies $e \leq e^*$. But since $u[s] \neq \doc[e(s)]$, we also know that $e^*(s) \neq e(s)$, which means that $e'(s) = e(s) + 1 \leq e^*(s)$. Consequently, any $\mathcal C$-embedding $e^*$ of $u$ in $\doc$ with $e_0 \leq e^*$ satisfies $e' \leq e^*$, which means that $(\dagger)_{e'}$ is satisfied. 

\paragraph{Line~\ref{orderCond}}
If the condition is satisfied, then we have that $s < |u|$ and $e(s) \geq e(s + 1)$. Moreover, we will change $e$ into $e'$ with $e'(i) = e(i)$ for every $i \in [|u|] \setminus \{s + 1\}$ and $e'(s + 1) = e(s) + 1$. Since $(\dagger)_e$ is satisfied, we know that any $\mathcal C$-embedding $e^*$ of $u$ in $\doc$ with $e_0 \leq e^*$ satisfies $e \leq e^*$. But since $e^*$ is an embedding of $u$ in $\doc$, we know that $e^*(s) + 1 \leq e^*(s + 1)$, which means that $e'(s + 1) = e(s) + 1 \leq e^*(s) + 1 \leq e^*(s + 1)$. Consequently, any $\mathcal C$-embedding $e^*$ of $u$ in $\doc$ with $e_0 \leq e^*$ satisfies $e' \leq e^*$, which means that $(\dagger)_{e'}$ is satisfied. 

\paragraph{Line~\ref{condLineOne}}
If the condition is satisfied for some $(i, j, L) \in\mathcal C$ with $s \in \{i, j\}$, then we have that $\doc[e(i) + 1 : e(j)-1] \notin L$. There are two cases:
(1.) either $e(i) < e(j)$,
in which case $\doc[e(i) + 1 : e(j)-1]$ is defined (and possibly $\eword$, if $e(j) = e(i) + 1$);
(2.) or $e(i) \geq e(j)$, in which case
$\doc[e(i) + 1 : e(j)-1] = \bot \notin L$ (recall that we define $\doc[p : q] = \bot$ in the case that $p > q+1$).
We first discuss case~(1.).
There are two subcases, depending on whether $\longright(e(i), L) > e(j)$
or not.
In the first subcase, where $\longright(e(i), L) > e(j)$, we change $e$ into $e'$ with $e'(k) \coloneq e(k)$ for every $k \in [|u|] \setminus \{j\}$ and $e'(j) \coloneq e(j) + 1$. In particular, we also know that $\doc[e(i) + 1 : \longright(e(i), L) - 1] \in L$. If $(\dagger)_{e'}$ is not satisfied, then there is some $\mathcal C$-embedding $e^*$ of $u$ in $\doc$ with $e_0 \leq e^*$, but $e' \not\leq e^*$. Since $(\dagger)_e$ is satisfied, we know that $e \leq e^*$, which means that $e(j) = e^*(j)$. Since $\doc[e(i) + 1 : e(j)-1] \notin L$, we also know that $e'(i) = e(i) < e^*(i)$. Let us now consider the string $uvw$ with $u = \doc[e(i) + 1 : e^*(i)]$, $v = \doc[e^*(i) + 1 : e^*(j)-1]$ and $w = \doc[e^*(j) : \longright(e(i), L) - 1]$. As observed above, $uvw \in L$, and since $e^*$ is a $\mathcal C$-embedding of $u$ in $\doc$, we also know that $v \in L$, which, by the left-convexity of $L$, implies that $uv = \doc[e(i) + 1 : e(j) - 1] \in L$; a contradiction. 

In the second subcase, where $\longright(e(i), L) \leq e(j)$, we change $e$ into $e'$ with $e'(k) \coloneq e(k)$ for every $k \in [|u|] \setminus \{i\}$ and $e'(i) \coloneq e(i) + 1$. If $(\dagger)_{e'}$ is not satisfied, then there is some $\mathcal C$-embedding $e^*$ of $u$ in $\doc$ with $e_0 \leq e^*$, but $e' \not\leq e^*$. Since $(\dagger)_e$ is satisfied, we know that $e \leq e^*$, which means that $e(i) = e^*(i)$. Since $\doc[e(i) + 1 : e(j)-1] \notin L$, we also know that $e'(j) = e(j) < e^*(j)$, which means that $\longright(e(i), L) > e(j)$; a contradiction.

Let us now discuss case~(2.), where $e(i) \geq e(j)$.
Again there are two subcases, depending on whether $\longright(e(i), L) > e(j)$ or not.
In the first subcase, where
$\longright(e(i), L) > e(j)$,
the algorithm changes $e$ into $e'$ in the same way as for case~(1.)
(i.e., by shifting $e(j)$ by one position) and, just like discussed above, we can conclude that if $(\dagger)_{e'}$ is not satisfied, then there is some $\mathcal C$-embedding $e^*$ of $u$ in $\doc$ with $e_0 \leq e^*$ and $e' \not\leq e^*$, which means that $e(j) = e^*(j)$ and $e'(i) = e(i) < e^*(i)$. Hence, $e^*(i) > e^*(j)$ and therefore $\doc[e^*(i) + 1 : e^*(j)-1] = \bot \notin L$, which is a contradiction. In the second subcase, where $\longright(e(i), L) \leq e(j)$, the algorithm changes $e$ into $e'$ by shifting $e(i)$ by one position. As in case~(1.), if $(\dagger)_{e'}$ is not satisfied, then there is some $\mathcal C$-embedding $e^*$ of $u$ in $\doc$ with $e(i) = e^*(i)$ and $e'(j) = e(j) < e^*(j)$. If $e^*(i) < e^*(j)$, then $\longright(e^*(i), L) \geq e^*(j)$, which, since $e(i) = e^*(i)$ and $e(j) < e^*(j)$, means that $\longright(e(i), L) > e(j)$ --- a contradiction. If $e^*(i) \geq e^*(j)$, then $\doc[e^*(i) + 1 : e^*(j)-1] = \bot$, which contradicts $\doc[e^*(i) + 1 : e^*(j)-1] \in L$.

\medskip

This shows that $(\dagger)_{e}$ holds before every iteration of the while-loop. 

For the next part, we need the following terminology. We say that $s \in [|u|]$ \emph{satisfies the symbol condition} with respect to $e$ if $u[s] = \doc[e(s)]$, $s \in [|u|-1]$ \emph{satisfies the order condition} with respect to $e$ if $e(s) < e(s + 1)$, and $s, s' \in [|u|]$ \emph{satisfy the gap condition} with respect to $e$ if $(i, j, L) \in\mathcal C$ with $\{i, j\} = \{s, s'\}$ implies $\doc[e(i) + 1 : e(j) - 1] \in L$. Obviously, if all $s \in [|u|]$ satisfy the symbol condition and the order condition with respect to $e$, and all $s, s' \in [|u|]$ satisfy the gap condition with respect to $e$, then $e$ is a $\mathcal C$-embedding of $u$ in $\doc$. 

Next, for every $e \colon [|u|] \to [|\doc|]$ and $S \subseteq [|u|]$, we define another invariant.

\smallskip

\noindent \emph{Invariant $(\ddagger)_{e, S}$}: Each $s \in [|u|] \setminus S$ satisfies the symbol condition and the order condition with respect to $e$, and all $s, s' \in [|u|] \setminus S$ satisfy the gap condition with respect to $e$.

\smallskip

Obviously, $(\ddagger)_{e, S}$ is satisfied at the beginning of the algorithm, since $S = [|u|]$.

\begin{claim}\label{SInvariantClaim}
Let $e \colon [|u|] \to [|\doc|]$ and let $S \subseteq [|u|]$ such that $(\ddagger)_{e, S}$ is satisfied. For some $s \in [|u|]$, let $e' \colon [|u|] \to [|\doc|]$ be such that $e'(i) = e(i)$ for every $i \in [|u|] \setminus \{s\}$ and $e'(s) > e(s)$, and let $S' = S \cup \{s\}$. Then $(\ddagger)_{e', S'}$ is satisfied. 
\end{claim}

\begin{proof}
We first show that every $s' \in [|u|] \setminus S'$ satisfies the symbol condition and order condition with respect to $e'$. Note that $s' \notin S'$ implies that $s' \notin S$, which, due to $(\ddagger)_{e, S}$, means that the symbol and order conditions for $s'$ are satisfied with respect to $e$. Since $s' \neq s$, we know that $e'(s') = e(s')$; thus, the symbol condition for $s'$ is satisfied with respect to $e'$. We further note that $e'$ is obtained from $e$ by only moving symbol $s \neq s'$ to the right, which means that $e'(s') = e(s')$ is also sufficient to conclude that the order condition for $s'$ is satisfied with respect to $e'$. Let us now consider some $s_1, s_2 \in [|u|] \setminus S'$, and some arbitrary gap-constraint $(i, j, L) \in\mathcal C$ with $\{i, j\} = \{s_1, s_2\}$ (if no such gap-constraint exists, then, by definition, $s_1, s_2$ satisfy the gap condition with respect to $e'$). By definition, $i, j \notin S'$ implies that $i, j \notin S$, which, due to $(\ddagger)_{e, S}$, means that $i, j$ satisfy the gap condition with respect to $e$. Since $s \in S'$, we also know that $s \notin \{i, j\}$, which means that $e'(i) = e(i)$ and $e'(j) = e(j)$; thus, $i, j$ satisfy the gap condition with respect to $e'$. Consequently, $(\ddagger)_{e', S'}$ is satisfied.
\end{proof}

Note that Claim~\ref{SInvariantClaim} also covers the case where $s \in S$. We can now use Claim~\ref{SInvariantClaim} to show that the invariant $(\ddagger)_{e, S}$ is maintained by any iteration of the while-loop. Let us assume that $(\ddagger)_{e, S}$ holds before an iteration of the while-loop.
Since the iteration may repeatedly change $S$, let us consistently use $\widehat{S}$ to denote the current version of that set, i.e., the original set $S$ with elements removed and added as done by the iteration up to the point that we currently consider in our argument. Analogously, we denote by $\widehat{e}$ the current mapping. In order to prove that $(\ddagger)_{\widehat{e}, \widehat{S}}$ holds after the iteration, we proceed in two steps: We will first show that $(\ddagger)_{\widehat{e}, \widehat{S} \cup \{s\}}$ holds after the iteration (for which we use Claim~\ref{SInvariantClaim}), and then we show that $(\ddagger)_{\widehat{e}, \widehat{S}}$ holds. 

Let us first observe that $(\ddagger)_{\widehat{e}, \widehat{S} \cup \{s\}}$ trivially holds before the iteration since $\widehat{S} = \widehat{S} \cup \{s\}$, and let us next argue that $(\ddagger)_{\widehat{e}, \widehat{S} \cup \{s\}}$ is maintained by the iteration. Indeed, if we perform any of the code inside one of the if-statements of Lines~\ref{symbolCond},~\ref{orderCond},~and~\ref{condLineOne}, then we will change $\widehat{e}$ by only moving a single position $s'$ to the right and adding $s'$ to the current set $\widehat{S}$ and therefore to $\widehat{S} \cup \{s\}$ (note that $s'$ might already be in $\widehat{S} \cup \{s\}$). With Claim~\ref{SInvariantClaim}, this directly implies that $(\ddagger)_{\widehat{e}, \widehat{S} \cup \{s\}}$ is maintained by all these if-statements and therefore $(\ddagger)_{\widehat{e}, \widehat{S} \cup \{s\}}$ is satisfied after the iteration. 

Next, let us see that in fact $(\ddagger)_{\widehat{e}, \widehat{S}}$ must be satisfied. If we have $s \in \widehat{S}$ after the iteration, then $(\ddagger)_{\widehat{e}, \widehat{S} \cup \{s\}} = (\ddagger)_{\widehat{e}, \widehat{S}}$ and we are done. Hence, let us consider the remaining case where $s \notin \widehat{S}$, which particularly means that $e(s)$ is not moved by any modification of $\widehat{e}$ in the iteration. In this case, we can conclude that $s$ satisfies the symbol condition with respect to $\widehat{e}$ due to the fact that we did not carry out the if-statement of Line~\ref{symbolCond}, which would have added $s$ to $\widehat{S}$. Similarly, $s$ must satisfy the order condition with respect to $\widehat{e}$ due to the if-statement of Line~\ref{orderCond} (if this if-statement has not been carried out, then the order condition was already satisfied, and if it has been carried out, then moving $e(s+1)$ to $e(s) + 1$ makes the order condition satisfied). Finally, let us consider some arbitrary $i, j \in [|u|] \setminus \widehat{S}$. If $s \notin \{i, j\}$, then $i, j$ satisfy the gap condition with respect to $\widehat{e}$ due to $(\ddagger)_{\widehat{e}, \widehat{S} \cup \{s\}}$. If, on the other hand, $s \in \{i, j\}$ and there is some $(i, j, L) \in\mathcal C$, then $\doc[e(i) + 1 : e(j) - 1] \in L$ due to the if-statements of Line~\ref{condLineOne} in the foreach-loop. Indeed, since neither $i$ nor $j$ are in $\widehat{S}$ after the iteration, neither $i$ nor $j$ could have been added to $\widehat{S}$ when this instance of Line~\ref{condLineOne} was executed, which means that $\doc[e(i) + 1 : e(j) - 1] \in L$ must have been satisfied. Moreover, neither $e(i)$ nor $e(j)$ are changed in the remainder of the iteration, since neither $i$ nor $j$ are in $\widehat{S}$ after the iteration. Hence, $i, j$ satisfy the gap condition with respect to $\widehat{e}$. Therefore, we can conclude that $(\ddagger)_{\widehat{e}, \widehat{S}}$ holds.

It can be easily seen that our algorithm always terminates: In every iteration of the while-loop, we remove a position from $S$ without adding a new one to it, or we move at least one position at least one step to the right. This means that at some point either $S$ becomes empty or it contains some $s$ with $e(s) > |\doc|$. In both cases our algorithm terminates.

As proven above, before the last iteration of the while-loop, both $(\dagger)_e$ and $(\ddagger)_{e, S}$ are satisfied. In the last iteration, we either turn $S = \{s\}$ into $S = \emptyset$ and return $e$ in Line~\ref{returneLine}, or we return $\bot$ due to Line~\ref{returnBotLine}. In the first case, the conditions of Lines~\ref{symbolCond}~and~\ref{orderCond} are not satisfied and we find no gap-constraint such that the condition in Line~\ref{condLineOne} is satisfied. As a result, $e$ is not changed, and $s$ satisfies the symbol condition and the order condition with respect to $e$, and, due to $(\ddagger)_{e, S}$, all other positions $s' \in [|u|] \setminus \{s\}$ also satisfy the symbol condition and the order condition with respect to $e$. Now for every $i, j \in [|u|]$ with a gap-constraint $(i, j, L) \in\mathcal C$, if $s \notin \{i, j\}$, then $\doc[e(i) + 1 : e(j) - 1] \in L$ is implied by  $(\ddagger)_{e, S}$, and if $s \in \{i, j\}$, then $\doc[e(i) + 1 : e(j) - 1] \in L$ is implied by the fact that in this last iteration, we did not find a gap-constraint such that the condition in Line~\ref{condLineOne} is satisfied. Consequently, all $i \in [|u|]$ satisfy the symbol condition and the order condition with respect to $e$, and all $i, j \in [|u|]$ satisfy the gap condition with respect to $e$, which means that $e$ is a $\mathcal C$-embedding of $u$ in $\doc$. Moreover, $(\dagger)_e$ implies that if there exists a $\mathcal C$-embedding $e^*$ of $u$ in $\doc$ with $e_0 \leq e^*$, then $e \leq e^*$, which means that $e$ is $e_0$-minimal.

On the other hand, if we return $\bot$ in the last iteration, then this means that $e(s) > |\doc|$. Due to $(\dagger)_{e}$, we know that if there exists a $\mathcal C$-embedding $e^*$ of $u$ in $\doc$ with $e_0 \leq e^*$, then $e \leq e^*$, which would mean that $e^*(s) > |\doc|$; a contradiction. Hence, if we return $\bot$, then there is no $\mathcal C$-embedding $e^*$ of $u$ in $\doc$.
\end{proof}

\subsection{Full Proof of Lemma~\ref{mainAlgoRunningTimeLemma}}

\mainAlgoRunningTimeRest*

\begin{proof}
We say that an iteration of the while-loop is an \emph{$s$-iteration} if $s$ is
the position removed from $S$ in Line~\ref{startLine} at the beginning of the
loop. Each single $s \in [|u|]$ can cause only $O(|\doc|)$ $s$-iterations, since
when it is added to $S$ it is also moved by at least one position, and if it is
removed from $S$, it has to be moved first in order to be added again to $S$.
Moreover,
everything except the foreach-loop of an $s$-iteration is executed in constant
time. Consequently, over the whole course of the algorithm, executing
Lines~\ref{startLine}~to~\ref{endFirstPartLine} of all iterations of the
while-loop requires time $O(|\doc| \cdot |u|)$. Let us next estimate the
total time needed for executing the foreach-loop of all iterations of the while-loop.

For every $s \in [|u|]$, let $k_s$ be the number of gap-constraints $(i, j, L)$
with $s \in \{i, j\}$. Obviously, $\sum^{|u|}_{i = 1} k_i = 2|\mathcal C|$. In
every $s$-iteration there are at most $k_s$ iterations of the foreach-loop,
assuming that we have precomputed for each $s \in [|u|]$ a list of the $k_s$
constraints in which $s$ occurs: this preprocessing takes $O(|u| +
|\mathcal C|)$ time and therefore has no effect on the overall complexity.
Thus, we need time $O(|\doc| k_s)$ for all the foreach-loops inside of
$s$-iterations, using
our assumption that
the conditions of Lines~\ref{condLineOne}~and~\ref{condLineTwo} can be
checked in constant time. This means that we need time $O(\sum^{|u|}_{i = 1} |\doc| \cdot k_i) = O(|\doc| \sum^{|u|}_{i = 1} k_i) = O(|\doc| \cdot |\mathcal C|)$ for all the foreach-loops in total. 

Consequently, the total time is $O(|\doc| \cdot |u| + |\doc| \cdot |\mathcal C|) = O(|\doc|(|u| + |\mathcal C|))$.
\end{proof}

\subsection{Full Proof of Lemma~\ref{computeLRSLArrays}}\label{computeLRSLArraysAppendix}

\computeLRSLArraysRest*

\begin{proof}
Let us first recall the definition of $\longright(k, L) = \max\{t\in[k+1,|\doc|]\mid \doc[k+1:t-1]\in L\}$ and $\shortleft(k, L) = \max\{t\in[k-1]\mid \doc[t+1:k-1]\in L\}$,
with the value being 0 if such a factor does not exist.
Intuitively speaking, if $\longright(k, L)=t$, then $\doc[k+1:t-1]$ is the \emph{longest} substring of $\doc$ starting at position $k+1$ that is contained in $L$, and if $\shortleft(k, L)=t$, then $\doc[t+1:k-1]$ is the \emph{shortest} substring of $\doc$ ending at position $k-1$ that is contained in $L$. The arrays $\LRArray$ and $\SLArray$ are the data structures that contain the values given by $\longright$ and $\shortleft$, i.e., $\LRArray[k][L] = \longright(k, L)$ and $\SLArray[k][L] = \shortleft(k, L)$.
    
        In the following, given an $\eword$NFA $M_L=(Q,\Sigma,q_0,F,\delta)$ accepting the gap-constraint language $L$, we will build a product graph $G_L=(V,E)$ with nodes $(k,q)$ for $k\in[|\doc|],q\in Q$ such that there is a path from $(k,q)$ to $(k',q')$ in $G_L$ if and only if there is a path from $q$ to $q'$ in $M_L$ that is labelled with $\doc[k+1:k']$ (when ignoring $\eword$-transitions).
        Clearly, $\LRArray[k][L]=\max\{t\in[|\doc|]\mid \text{there is a path from } (k,q_0) \text{ to } (t-1,q_f) \text{ in } G_L \text{ for some } q_f\in F\}$. Thus, we can compute $\LRArray[k][L]$ by finding the longest (ignoring $\eword$-transitions) path from $(k,q_0)$, corresponding to the initial state $q_0$ at position $k$, to $(t-1,q_f)$, corresponding to some final state $q_f\in F$ at position $t-1$: This path from $(k,q_0)$ to $(t-1,q_f)$ in $G_L$ corresponds to an accepting run in $M_L$ and is labelled with $\doc[k+1:t-1]$. Similarly, $\SLArray[k][L]=\max\{t\in[|\doc|]\mid \text{there is a path from } (t,q_0) \text{ to } (k-1,q_f) \text{ in } G_L \text{ for some } q_f\in F\}$ and we can compute $\SLArray[k][L]$ by finding the shortest (when ignoring $\eword$-transitions) path from $(t,q_0)$ to $(k-1,q_f)$, corresponding to some accepting run in $M_L$ that is labelled with $\doc[t+1:k-1]$. 
        
        Before constructing $G_L$, we will preprocess the $\eword$NFA $M_L$ to remove cycles consisting only of $\eword$-transitions. 
        To do this, let $M_\eword=(Q,\Sigma,q_0,F,\{(p,\eword,q)\in\delta\})$ be the $\eword$NFA that we obtain from $M_L$ when only considering $\eword$-transitions. Using a depth-first search, we construct the set of strongly connected components of $M_\eword$ in time $O(|M_\eword|)=O(|M_L|)$. For ease of notation, let $SCC[q]$ be the strongly connected component that contains state $q\in Q$ and let $SCC[A]=\{SCC[q]\mid q\in A\}$ for subset $A\subseteq Q$. Then, instead of $M_L$ we will from now on work with the $\eword$NFA $M'_L=(SCC[Q],\Sigma, SCC[q_0], SCC[F], \delta')$, where $\delta'=\{(SCC[p],a,SCC[q])\mid (p,a,q)\in\delta\}$. Since we collapse states $p$ and $q$ into a single strongly connected component (i.e., $SCC[p]=SCC[q]$) if and only if there is an $\eword$-labelled path from $p$ to $q$ and an $\eword$-labelled path from $q$ to $p$, it is easy to see that $\LL(M_L)=\LL(M'_L)$ and $|M'_L|\leq |M_L|$. For ease of notation, from now on we write $M_L=(Q,\Sigma,q_0,F,\delta)$ while meaning $M'_L$. 

        Now, let us construct $G_L=(V,E)$. As explained above, every node in $V=\{(k,q)\mid k\in[|\doc|],q\in Q\}$ consists of two components: $k\in[|\doc|]$ symbolises that we already consumed prefix $\doc[1:k]$ (or some suffix thereof) when reaching node $(k,q)$, while $q$ stores the current state of the automaton after reading the already consumed symbols. For ease of notation, we say that node $(k,q)$ has \emph{index $k$}. We add an edge from $(k,q)$ to $(k,q')$ if there is an $\eword$-labelled edge from $q$ to $q'$ in $M_L$, and an edge from $(k,q)$ to $(k+1,q')$ if there is an edge from $q$ to $q'$ in $M_L$ that is labelled with letter $\doc[k+1]$. Clearly, there is a path from $(k,q)$ to $(k',q')$ for $k'>k$ if and only if there is a $\doc[k+1:k']$-labelled path (when ignoring $\eword$-transitions) from state $q$ to $q'$ in $M_L$, and there is a path from $(k,q)$ to $(k,q')$ if and only if there is an $\eword$-labelled path from $q$ to $q'$ in $M_L$. Clearly, we can build $G_L$ in $O(|\doc||M_L|)$ time. 

        We can now compute $\LRArray$ and $\SLArray$ using dynamic programming. Let us first consider the procedure for $\LRArray$. Since we removed all $\eword$-cycles, we can topologically sort all nodes with index $k\in[|\doc|]$ in time $O(|M_L|)$ (and for the whole graph thus in overall time $O(|\doc||M_L|)$), to maintain the following invariant: When considering node $(k,q)$, we already computed the values of all nodes reachable from $(k,q)$. Indeed, we consider all nodes $(k,q)$ in decreasing order, such that when computing the value of node $(k,q)$, all nodes with index $k+1$ have already been considered, and such that, if $q'$ is larger than $q$ in the topological sorting for index $k$, then the value of $(k,q')$ has been computed before the value of $(k,q)$. Now, for every $(k,q)$ we compute the value $v_j(k,q)$, which is the largest $k'$ such that there is a path from $(k,q)$ to some node $(k',q_f)$ with $q_f\in F$,
        or $-1$ if there is no such $k'$. It is easy to see that 
        \begin{equation*}
            v_j(k,q) = \begin{cases}
                \max(\{v_j(k',q')\mid ((k,q),(k',q'))\in E\}) & \text{if } q\notin F,\\
                \max(\{v_j(k',q')\mid ((k,q),(k',q'))\in E\}\cup\{k\}) & \text{if } q\in F,
            \end{cases}
            \end{equation*}
        where $\max(\emptyset)$ is denoted 
        by~$-1$.
        Finally, we set $\LRArray[k][L]=v_j(k,q_0)+1$. This is correct, since for $t = v_j(k,q_0)+1$,
        if $t>0$ we have that $\doc[k+1:t-1] \in L$ and $t$ is the largest such value with $t \in [k+1,|\doc|]$, and if $t=0$ we had $v_j(k,q_0)=-1$ which witnesses that we do not have any suitable factor.

        A similar approach can be used for $\SLArray$: we consider the nodes $(k,q)$ in increasing order according to their index $k$, again in decreasing order of the topological sorting of index $k$. For every node $(k,q)$ we compute the value $v_i(k,q)$, which is the largest $k'$ such that there is a path from $(k',q_0)$ to node $(k,q)$, or $-1$ if there is no such $k'$. Now, $\SLArray[k][L]=\max\{v_i(k-1,q_f)\mid q_f\in F\}$ with similar reasoning. 

        Since there are $O(|M_L|)$ edges to be considered for every index $k$, this dynamic programming approach clearly takes time $O(|\doc||M_L|)$.
    \end{proof}

\section{Full Details for the Approach of Section~\ref{gaoetalapproach}}\label{sec:GaoAppendix}

As noted in Section~\ref{gaoetalapproach}, when in an iteration of the while-loop the condition of Line~\ref{condLineOne} is satisfied, i.e., $\doc[e(i) + 1 : e(j)-1] \notin L$, then Algorithm~\ref{mainAlgo} moves either position $i$ or $j$ by only one position, i.e., either $e(i) \coloneq e(i) + 1$ or $e(j) \coloneq e(j) + 1$. However, we could as well set $e(i) \coloneq p$ and $e(j) \coloneq q$, where $p,q\in[|\doc|]$ such that $e(i) \leq p, e(j) \leq q$, $\doc[p+1:q-1]\in L$ and $p,q$ are minimal with this property (note that if there are any $p'\geq e(i),q'\geq e(j)$ such that $\doc[p'+1:q'-1]\in L$, then the uniqueness of the (pointwise) minimal pair $(p,q)$ is a direct consequence of the left-convexity of $L$; see Lemma~\ref{uniquenessLemma}). Similarly,  if $\doc[e(i)]\neq u[i]$, then $e$ cannot be a $\mathcal C$-embedding of $u$ in $\doc$, so when evaluating the symbol condition, it suffices to only consider positions $x>e(i)$ with $\doc[x]=u[i]$. 
The correctness of this follows again from the left-convexity property, as shown below. The resulting algorithm is not asymptotically faster in the worst case. However, we expect that the algorithm could be more efficient in practice, especially for instances where $\doc$ is sparse with respect to factors that match the constraint languages; indeed, shifting positions over larger chunks of $\doc$ as explained above could help finding a valid embedding (or verifying that none exists) more efficiently.
We can state the following result:

\begin{theorem}\label{GaoAlgoTheorem}
Algorithm~\ref{GaoAlgo} solves the matching problem in $O(|\doc| (|u|\log\log|\doc| + |\mathcal C|\sqrt{\log |\doc|} + \lVert \mathcal C\rVert))$ time.
\end{theorem}

\begin{algorithm}
    \caption{$\embedLCONSubseq(u, \mathcal{C}, \doc, e_0)$}
        \label{GaoAlgo}
   \KwIn{$u, \mathcal{C}$ s.t. $\mathcal{C}$ only contains left-convex gap-constraints, $\doc \in\Sigma^\ast$, $e_0 \colon [|u|] \to [|\doc|]$.}
    \KwOut{$e_0$-minimal $\mathcal C$-embedding of $u$ in $\doc$ if it exists, and $\bot$ otherwise.}
        $e \coloneq e_0$; $S \coloneq \{1, 2, \ldots, |u|\}$\;
        \While{$S \neq \emptyset$}
        {
        Let $s \in S$ be arbitrarily chosen and $S \gets S \setminus \{s\}$\;
        \lIf{$e(s) > |\doc|$}{\textbf{return} $\bot$}
        \If{$u[s] \neq \doc[e(s)]$\label{GaosymbolCond}}
        {
        	$e(s) \gets \min(\{x\in[e(s)+1,|\doc|]\mid \doc[x]=u[s]\}\cup\{|\doc|+1\})$; $S \gets S \cup \{s\}$\label{GaoLetterQuery}\;
        }
        \If{$s < |u|$ and $e(s) \geq e(s + 1)$\label{GaoorderCond}}
        {
        	$e(s + 1) \gets e(s) + 1$; $S \gets S \cup \{s + 1\}$\;
        }
        \ForEach{$(i, j, L) \in\mathcal C$ with $s \in \{i, j\}$}
        {
        	\If{$\doc[e(i) + 1 : e(j)-1] \notin L$\label{GaocondLineOne}}
		{
            $p \gets \min(\{x\in[e(i),|\doc|]\mid \longright(x, L) \geq e(j)\}\cup\{|\doc|+1\})$\label{GaoLRQuery}\;
            $q \gets \min(\{x\in[e(j),|\doc|] \mid \shortleft(x, L)\geq e(i)\}\cup\{|\doc|+1\})$\label{GaoSLQuery}\;
            \lIf{$p>e(i)$}{$e(i)\gets p$; $S\gets S\cup \{i\}$}
            \lIf{$q>e(j)$}{$e(j)\gets q$; $S\gets S\cup \{j\}$}
        }
        }
        }
        \textbf{return} $e$\;
\end{algorithm}

\begin{proof}
The result of Theorem~\ref{GaoAlgoTheorem} relies on Algorithm~\ref{GaoAlgo}. To show the correctness and complexity of this algorithm, we show the following several intermediate results.
\begin{lemma}\label{GaoAlgoCorrectnessLemma}
On input $u \in \Sigma^*$, $\mathcal{C}$ (where $\mathcal{C}$ only contains left-convex gap-constraints), $\doc \in\Sigma^\ast$ and $e_0 \colon [|u|] \to [|\doc|]$, Algorithm~\ref{GaoAlgo} returns an $e_0$-minimal $\mathcal C$-embedding of $u$ in $\doc$ if it exists, and $\bot$ otherwise.
\end{lemma}
\begin{proof}
The general correctness of Algorithm~\ref{GaoAlgo} follows the correctness proof given of Lemma~\ref{mainAlgoCorrectnessLemma}. We will only highlight the differences.
Let us assume the invariant $(\dagger)_e$ is satisfied for the current mapping of the algorithm. We will first show that this invariant is maintained by the modifications caused by the if-statements of Line~\ref{GaosymbolCond}~and~\ref{GaocondLineOne}.

\paragraph{Line~\ref{GaosymbolCond}}
If the condition is satisfied, then we have that $u[s] \neq \doc[e(s)]$ and we will change $e$ into $e'$ with $e'(i) = e(i)$ for every $i \in [|u|] \setminus \{s\}$ and $e'(s) = \min(\{x\in[e(s)+1:|\doc|]\mid \doc[x]=u[s]\}\cup\{|\doc|+1\})$, i.e., the smallest position where letter $u[s]$ appears in $\doc[e(s)+1:|\doc|]$ (or $|\doc|+1$, if $u[s]$ does not occur in $\doc$ after position $e(s)$). Since $(\dagger)_e$ is satisfied, we know that any $\mathcal C$-embedding $e^*$ of $u$ in $\doc$ with $e_0 \leq e^*$ satisfies $e \leq e^*$. But since $u[s] \neq \doc[e(s)]$, we also know that $e^*(s) > e(s)$. Further, if $u[s]$ appears in $\doc[e(s):|\doc|]$, then its first occurrence must be at position $e'(s)$, by definition. Since $e^\ast$ has to satisfy $\doc[e^\ast(s)]=u[s]$, we can conclude that $e^\ast(s)\geq e'(s)$ and thus the invariant $(\dagger)_{e'}$ holds.

\paragraph{Line~\ref{GaocondLineOne}}
If the condition is satisfied for some $(i, j, L) \in\mathcal C$ with $s \in \{i, j\}$, then we have that $\doc[e(i) + 1 : e(j)-1] \notin L$. Let us consider the following claim, which is reminiscent of Lemma~\ref{longestRightShortestLeftLemma}:

\begin{claim}\label{claim:gaominimal}
    Let $e\colon [|u|]\to [|\doc|]$ and $(i,j,L)\in \mathcal C$ with $e(i)=a, e(j)=b$ such that the sets $\{x\in[a,|\doc|]\mid \longright(x, L) \geq b\}$ and $\{x\in[b,|\doc|] \mid \shortleft(x, L)\geq a\}$ are not empty. Let $p=\min\{x\in[a,|\doc|]\mid \longright(x, L) \geq b\}$ and $q=\min\{x\in[b,|\doc|] \mid \shortleft(x, L)\geq a\}$. Then $\doc[p+1:q-1]\in L$ and $(p,q)$ is pointwise minimal with this property, i.e., there is no pair $(p',q')\in[a,|\doc|]\times [b,|\doc|]$ with $\doc[p'+1:q'-1]\in L$ and $p'<p$ or $q'<q$. 
\end{claim}
\begin{proof}
    Firstly, let us assume that $e$ satisfies gap-constraint $(i,j,L)$. By Lemma~\ref{longestRightShortestLeftLemma}, we have $\longright(e(i),L)\geq e(j)=b$ and $\shortleft(e(j),L)\geq e(i)=a$, thus $e(i)=p$ and $e(j)=q$ and the pair $(p,q)$ is clearly minimal with this property. 

    Let us now assume that $\doc[a+1:b-1]\notin L$, i.e., $e$ does not satisfy constraint $(i,j,L)$. Let $t_p = \longright(p, L)$ and $t_q = \shortleft(q, L)$. These positions are well-defined, since $p=\min\{x\in[a,|\doc|]\mid \longright(x, L) \geq b\}$ and $q=\min\{x\in[b,|\doc|] \mid \shortleft(x, L)\geq a\}$ and both sets are non-empty. Since $t_q$ is a valid choice for $x$ in the condition $\longright(x, L) \geq b$ and $t_p$ is a valid choice for $x$ in the condition $\shortleft(x, L)\geq a$, by the definition of $p$ and $q$, we can conclude that $p \leq t_q$ and $q \leq t_p$. Hence, $p \leq t_q < q \leq t_p$ with $\doc[p+ 1 : t_p - 1] \in L$ and $\doc[t_q + 1 : q - 1] \in L$. Since $L\in \LCON$, this means that $\doc[p + 1 : q - 1] \in L$. Further, if there was a factor $\doc[p'+1:q'-1]\in L$ for some $p'\geq a=e(i), q'\geq b=e(j)$ with $p'<p$ (the case $q'<q$ is analogous), then $\longright(p',L)\geq q'\geq b$ must hold, contradicting the minimality of $p$.
    
    See also Figure~\ref{fig:LRSL}. 
\end{proof}

\begin{figure}
 \centering
\includegraphics[scale=0.7]{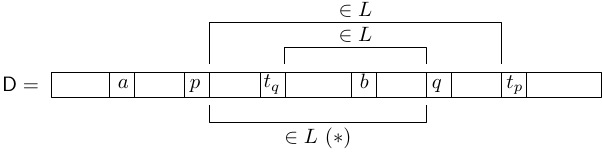}
\caption{Proof of Claim~\ref{claim:gaominimal}, where $(\ast)$ follows by left-convexity.}
\label{fig:LRSL}
\end{figure}

It directly follows from Claim~\ref{claim:gaominimal} that if there is a $\mathcal C$-embedding $e^\ast$ of $u$ in $\doc$ with $e\leq e^\ast$, then $e^\ast(i)\geq p$ and $e^\ast(j)\geq q$ must hold. Thus, if we change $e$ into $e'$ with $e'(k) = e(k)$ for every $k \in [|u|] \setminus \{i,j\}$, $e'(i)=p$, and $e'(j)=q$, then $(\dagger)_{e'}$ is satisfied as well. 

\medskip

As a direct consequence of Claim~\ref{SInvariantClaim}, the second invariant $(\ddagger)_{e,S}$ is also clearly satisfied after every iteration of the algorithm, concluding our proof.
\end{proof}

To achieve the stated complexity, we need the following result (based on the availability of data structures allowing us to quickly answer the queries in Lines~\ref{GaoLetterQuery} and~\ref{GaoLRQuery}-\ref{GaoSLQuery}, respectively). 

\begin{lemma}
    Under the assumption that we can compute the query in Line~\ref{GaoLetterQuery} in constant time and the queries in Lines~\ref{GaoLRQuery} and~\ref{GaoSLQuery} in time $O(\log\log|\doc|)$, we have that Algorithm~\ref{GaoAlgo} terminates after time $O(|\doc|(|u|\log\log|\doc|+|\mathcal C|))$. 
\end{lemma}
\begin{proof}
    Apart from the time needed to execute the Lines~\ref{GaoLetterQuery},~\ref{GaoLRQuery}~and~\ref{GaoSLQuery}, the complexity directly carries over from the proof of Lemma~\ref{mainAlgoRunningTimeLemma}. Namely, the time needed to check the conditions of Lines~\ref{GaosymbolCond},~\ref{GaoorderCond},~and~\ref{GaocondLineOne}, as well as the time needed for updating $e$ and $S$, is $O(|\doc|(|u|+|\mathcal C|))$. Further, over the whole computation, the time spent computing Line~\ref{GaoLetterQuery} is clearly $O(|\doc|\cdot|u|)$, since we evaluate the symbol condition in Line~\ref{GaosymbolCond} only $O(|\doc|\cdot|u|)$ many times, thus upper bounding the number of times Line~\ref{GaoLetterQuery} is executed.

    Now, let us consider the time needed to execute Lines~\ref{GaoLRQuery}~and~\ref{GaoSLQuery}. We can trivially bound the number of iterations of those lines by $O(|\doc|\cdot|\mathcal C|)$, thus leading to an overall factor of $O(|\doc|\cdot|\mathcal C|\log\log|\doc|)$. However, this is not optimal. As we only execute the Lines~\ref{GaoLRQuery}~and~\ref{GaoSLQuery} for a given $(i,j,L)$-iteration if and only if $\doc[e(i)+1:e(j)-1]\notin L$, we know that at least one position must be increased by this call. Thus, there can be at most $O(|\doc|\cdot|u|)$ many $(i,j,L)$-iterations that satisfy the condition of Line~\ref{GaocondLineOne}, so the time needed to execute Lines~\ref{GaoLRQuery}~and~\ref{GaoSLQuery} is $O(|\doc|\cdot|u|\log\log|\doc|)$, concluding the proof.
\end{proof}

As stated before the previous lemma, in order to conclude the proof of Theorem~\ref{GaoAlgoTheorem}, we have to show that the queries in Line~\ref{GaoLetterQuery}, as well as Lines~\ref{GaoLRQuery} and~\ref{GaoSLQuery}, can be computed in constant, respectively $O(\log\log|\doc|)$, time. 
Firstly, it is not hard to see that we can precompute a matrix $\text{next}[i][a]=\min\{x\in[i,|\doc|]\mid \doc[x]=a\}$ for $i\in[|\doc|]$ and $a\in\Sigma'$, where $\Sigma'=\{u[i]\mid i\in[|u|]\}$ is the alphabet of $u$, by a simple dynamic programming approach in $O(|\doc|\cdot|u|)$ time.

To efficiently compute the values $p=\min\{x\in[e(i),|\doc|]\mid \longright(x, L) \geq e(j)\}$ and $q=\min\{x\in[e(j),|\doc|] \mid \shortleft(x, L)\geq e(i)\}$ for given $e\colon [|u|]\to[|\doc|]$ and gap-constraint $(i,j,L)\in\mathcal C$, we employ a data structure result by Gao et al.~\cite{gao_et_al:LIPIcs.ESA.2020.54}. 
They show that, given a set $P$ of $N$ points over an $N \times N$ grid, we can construct in $O(N \sqrt{\log N})$ time a data structure that allows us to answer \emph{orthogonal range successor} queries: given a rectangle $R$, the query returns the point in $P \cap R$ with the smallest $x$-coordinate in $O(\log\log N)$ time. With this data structure we can compute $p$ and $q$ in Lines~\ref{GaoLRQuery}~and~\ref{GaoSLQuery} as follows. 

For every gap-constraint language $L$ in $\mathcal C$, we construct two data structures for orthogonal range successor queries on the $|\doc| \times |\doc|$ grid. The first one is for the set of points $P_L^{\LRArray}=\{(k,\LRArray[k][L]) \mid k\in[|\doc|]\}$, while the second is for the set of points $P^{\SLArray}_{L}=\{(k,\SLArray[k][L])\mid k\in[|\doc|]\}$; given the arrays $\LRArray[\cdot][L]$ and $\SLArray[\cdot][L]$ computed earlier, this takes $O(|\doc|\sqrt{\log|\doc|})$ time. In order to compute $p$ and $q$, we use these data structures as follows. Since $p=\min\{x\in[e(i),|\doc|]\mid \longright(x, L) \geq e(j)\}$, we obtain $p$ as the $x$-coordinate of the point of minimal $x$-coordinate in $R \cap P_L^{\LRArray}$, where $R\coloneq [e(i),|\doc|]\times[e(j),|\doc|]$; see also Figure~\ref{fig:gaostructure} for an illustration of the procedure for $P_L^{\LRArray}$. Similarly, $q=\min\{x\in[e(j),|\doc|]\mid \shortleft(x, L) \geq e(i)\}$ is obtained as the $x$-coordinate of the point of minimal $x$-coordinate in $P_L^{\SLArray}\cap R'$, where $R'\coloneq [e(j),|\doc|]\times[e(i),|\doc|]$. If any of the two orthogonal range successor queries does not return any point, then there is no pair $(p',q')$ with $p'\geq e(i), q'\geq e(j)$ and $\doc[p'+1:q'-1]\in L$. Thus, after a preprocessing time of $O(|\mathcal{C}| \cdot |\doc| \sqrt{\log |\doc|})$, for each gap-constraint language $L$ we can execute the queries in Lines~\ref{GaoLRQuery}~and~\ref{GaoSLQuery} in $O(\log\log|\doc|)$ time.

\begin{figure}[t]
 \centering
\includegraphics[scale=0.7]{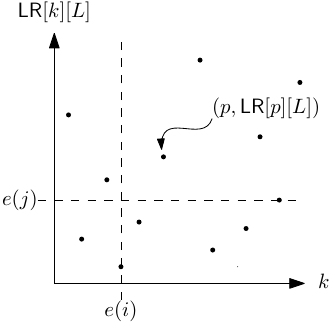}
\caption{Illustration of the computation of $p=\min\{x\in[e(i),|\doc|]\mid \longright(x,L)\geq e(j)\}$ using orthogonal range successor queries.}
\label{fig:gaostructure}
\end{figure}

This shows that the statement of Theorem~\ref{GaoAlgoTheorem} holds, and
concludes the proof of Theorem~\ref{GaoAlgoTheorem}.\end{proof} 

We conclude this section with a series of observations about two particularly relevant classes of gap-constraints, for which we can implement the variant of the algorithm from Theorem~\ref{mainAlgoTheorem}, which jumps over larger factors, quicker than the standard variant of this algorithm. 

\begin{observation}\label{obs:LengthSingleton}
First, let us consider the case when all gap-constraints are length constraints, i.e., constraints of the form $(i,j,L)$, with $L = \{w \in \Sigma^* \mid a \leq |w| \leq b\}$. In this case, it is natural to assume that these constraints are simply given as triples $(i,j,[a,b])$, with both $a$ and $b$ given in their binary representation; moreover, we can assume $a\geq j-i-1$ (otherwise, we can replace the constraint $[a,b]$ by $[j-i-1,b]$). It is not hard to see that, in this case, one can check, for a constraint $(i,j,[a,b])$, in constant time (without any preprocessing) whether some factor $\doc[p+1:q-1]$ satisfies the respective constraint; one simply needs to check whether $a\leq q-p-1\leq b$. So, we run the variant of our algorithm which jumps over longer factors, without any of the usual preprocessing. According to our approach, if in an iteration of the while-loop the condition of Line~\ref{condLineOne} (of the general algorithm) is satisfied, i.e., $\doc[e(i) + 1 : e(j)-1] \notin L$, we would like to set $e(i) \coloneq p$ and $e(j) \coloneq q$, where $p,q\in[|\doc|]$ such that $e(i) \leq p, e(j) \leq q$, $\doc[p+1:q-1]\in L$ and $p,q$ are minimal with this property. Let us note that if  $\doc[e(i) + 1 : e(j)-1] \notin L$, then either the factor $\doc[e(i)+1:e(j)-1]$ is too long or too short w.r.t. the constraint $[a,b]$. If $|\doc[e(i)+1:e(j)-1]|<a$, then we can set $e(i),e(j)$ to be the pair $(e(i),e(i)+a+1)$, and if $|\doc[e(i)+1:e(j)-1]|>b$, then we return the pair $(e(j)-b-1,e(j))$. Based on these insights, in this case, our algorithm runs in $O(|\doc| (|u|+|\mathcal{C}|))$ time.

Secondly, we consider the case when all gap-constraints are singletons, i.e., constraints of the form $(i,j,L)$, with $L = \{w\}$, for some string $w$. Assume that $L_1=\{w_1\},\ldots, L_{|\mathcal{C}|}=\{w_{|\mathcal{C}|}\}$ are the singleton languages which appear in the gap-constraints. In this case, the preprocessing is done differently. We construct the word $w=\doc w_1\cdots w_{|\mathcal{C}|}$, and \emph{longest common extension} data structures for it \cite{DBLP:journals/jda/IlieNT10,DBLP:conf/cpm/FischerH06} (also sometimes called \emph{longest common prefix} data structures); this takes $O(\lVert \mathcal{C} \rVert + |\doc|)$ time and allows us to test in $O(1)$ time, for some $i\in [|\mathcal C|]$ and $j\in [|\doc|]$, whether $w_i=\doc[j:j+|w_i|-1]$. Using these data structures, we can trivially compute in $O(|\doc|\cdot|\mathcal{C}|)$ time the arrays $\LRArray[\cdot][L_i]$ and $\SLArray[\cdot][L_i]$. Indeed, we have $\LRArray[k][L_i]=k+|w_i|+1$ if and only if $w_i=\doc[k+1:k+|w_i|]$; for convenience, we set $\LRArray[k][L_i]=\infty$ otherwise. Similarly, $\SLArray[k][L_i]=k-|w_i|-1$ if and only if $w_i=\doc[k-|w_i|:k-1]$, and $\SLArray[k][L_i]=\infty$ otherwise. Moreover, we construct data structures in total $O(|\doc|\cdot|\mathcal{C}|)$ time (for all $i$) to answer \emph{range minimum} queries~\cite{DBLP:conf/cpm/FischerH06} for the arrays $\SLArray[\cdot][L_i]$ allowing us to retrieve in $O(1)$ time the minimum value from a range $\SLArray[a:b][L_i]$; these (more efficient) structures replace, in this particular case, the structures we had for orthogonal range successor in the alternative implementation of the matching algorithm. 
As above, if in an iteration of the while-loop the condition of Line~\ref{condLineOne} (of the general algorithm) is satisfied, i.e., $\doc[e(i) + 1 : e(j)-1] \notin L$, with $L=\{w\}$, we would like to set $e(i) \coloneq p$ and $e(j) \coloneq q$, where $p,q\in[|\doc|]$ such that $e(i) \leq p, e(j) \leq q$, $\doc[p+1:q-1]=w$ and $p,q$ are minimal with this property. In this case, for $b'=\max\{e(i)+|w|+1,e(j)\}$, we now compute $p'$ as the answer to the range minimum query for $\SLArray[b':|\doc|][L]$ (i.e., $p'+1$ is the leftmost position, greater or equal to $\max\{e(i)+1,e(j)-|w|\}$, where $w$ occurs) and set $(e(i),e(j))$ to be the pair  $(p',p'+|w|+1)$. The matching algorithm runs, in this case, in $O(\lVert \mathcal{C} \rVert + |\doc| (|u|+|\mathcal{C}|))$ time. 

In the case of singleton gap-constraints, this approach of jumping over larger factors may even work better in practical settings, as the expected number of occurrences of a constraint string $w$ (corresponding to the constraint $(i,j,w)$) in a random input document $\doc$ (over an alphabet of size $\sigma$, and independent from the constraints) is $O(|\doc|/\sigma^{|w|})$ -- so the number of positions we need to consider for $e(i)$ and $e(j)$ can be significantly smaller than $|\doc|$. However, a full average-case complexity analysis of our algorithms is out of the scope of the present paper.
\hfill $\triangleleft $
\end{observation}

\section{Full Details for Section~\ref{sec:computeAll}}\label{sec:enumDetails}

\begin{algorithm}
    \caption{$\embedAll(u,\mathcal C, \doc, e, i)$}\label{algorithm:findall}
    \KwIn{$u,\mathcal C$ s.t. $\mathcal C$ only contains $\LCON$-constraints, $\doc\in\Sigma^\ast$, $e\colon[|u|]\rightarrow[|\doc|]$, positions $1,\ldots,i$ locked for $i\in\{0,1,\ldots,|u|-1\}$.}%
    \KwOut{output all embeddings from $S(e,i)$.
    }
    $e^\ast \gets \embedLCONSubseq(u,\mathcal C,\doc,e)$\;
    \label{lin:firstreccall}
    \lIf{$e^\ast=\bot$}{\Return}
    output $e^\ast$\; \label{lin:findalloutput}
    $j\gets\nextMove(e^\ast,i)$\; \label{lin:firstnextmove}
    \While{$j\neq\bot$}{ \label{lin:whilecondition}
        $(e^\ast)^{(j)}\gets\incr(e^\ast,j)$\; 
        $\embedAll(u,\mathcal C,\doc, (e^\ast)^{(j)},j-1)$\;
        $j\gets \nextMove(e^\ast,j)$\; \label{lin:secondnextmove}
    }
\end{algorithm}

\begin{theorem}\label{theorem:compvariantAppendix}
Let $u\in\Sigma^\ast$ be a query string with a set $\mathcal C$ of gap-constraints for $u$ such that $\mathcal C$ only contains left-convex gap-constraints, and let $\doc\in\Sigma^\ast$ with $|u|\leq|\doc|$. 
Then we can compute the set $\query(\doc)$ of all $\mathcal C$-embeddings of $u$ in $\doc$ in time $O(|\query(\doc)|\cdot|u|\cdot|\doc| (|u| + \lVert \mathcal C\rVert))$.
\end{theorem}

\begin{proof}
For simplicity, let us denote by $T(|u|, |\doc|, \lVert \mathcal{C} \rVert)$ the running time of Algorithm~\ref{mainAlgo}, i.e., $O(|\doc| (|u| + \lVert \mathcal C\rVert))$.

    Let us first give some definitions. We say that position $i$ is \emph{locked} (w.r.t. mapping $e_0$) if, when finding the $e_0$-minimal $\mathcal C$-embedding $e$ of $u$ in $\doc$, we only consider the set $\{e\mid e_0 \leq e, e(i)=e_0(i)\}$.
     In other words, by locking position $i$, we only look at embeddings $e$ that have the same value for position $i$ as mapping $e_0$ (i.e., $e(i)=e_0(i)$), and return $\bot$ (signaling that the solution set is empty) if there is a constraint that can only be satisfied by shifting $e(i)$ to some value greater than $e_0(i)$. Further, given some mapping $e$ and position $i\in\{0,1,\ldots,|u|\}$, we define the solution set $S(e,i)=\{e'\mid e\leq e', e'(1)=e(1),\ldots,e'(i)=e(i), u\preceq_{e',\mathcal C}\doc\}$ as the set of $\mathcal C$-embeddings $e'$ of $u$ in $\doc$ that can be obtained from $e$ by increasing some (or none) of its positions $i+1,\ldots,|u|$, and locking positions $1,\ldots,i$. Note that $S(e,0)$ contains all $\mathcal C$-embeddings $e'$ of $u$ in $\doc$ with $e\leq e'$ and is empty if and only if there is no $e$-minimal $\mathcal C$-embedding of $u$ in $\doc$; clearly, if $u\preceq_{e,\mathcal C}\doc$ then $e$ is contained in $S(e,i)$ for any $i\in\{0,1,\ldots,|u|\}$. 

    Next, in order to obtain subsets of $S(e,i)$, let us define the minimal mappings $e^{(j)}$ that are strictly larger than $e$ and can be obtained from $e$ by shifting $e(j)$ to $e(j)+1$ for some $j\in[|u|]$. 
    In particular, for every mapping $e\colon[|u|]\to[|\doc|]$ and $j\in[|u|]$, the mapping $e^{(j)}\colon[|u|]\to[|\doc|]$ moves $e(j)$ to position $e(j)+1$, and coincides with $e$ on all other positions. Formally, $e^{(j)}$ is defined by $e^{(j)}(p)\coloneq e(p)$ for all $p\in[|u|]\setminus\{j\}$ and $e^{(j)}(j)\coloneq e(j)+1$. Let $\incr(e,j)$ be the function that computes $e^{(j)}$ for a given mapping $e$ and $j\in[|u|]$; obviously, $e^{(j)}=\incr(e,j)$ can be computed in $O(|u|)$ time. Note that, whenever $e^{(j)}(j)>|\doc|$, then Algorithm~\ref{mainAlgo} returns $\bot$ after evaluating $j\in S$ for the first time and $S(e^{(j)},j-1)$ must necessarily be empty.   

    It is not hard to see that, given some $i\in\{0,1,\ldots,|u|-1\}$ and $j>i$,
    the set $S(e^{(j)},i)$ is a (not necessarily strict) subset of $S(e,i)$. We can even use these mappings to partition $S(e,i)$:  
    
    \begin{claim}\label{claim:disjointsets}
        Let $e\colon[|u|]\rightarrow[|\doc|]$ be a mapping and $i\in\{0, 1, \ldots, |u|\}$.
        \begin{itemize}[nosep]
            \item $S(e,i)=\emptyset$ if and only if $S(e,0) = \emptyset$ or $e_{\min}(k)>e(k)$ for some $k\in[i]$, where $e_{\min}$ is the $e$-minimal $\mathcal C$-embedding of $u$ in $\doc$.
            \item If $S(e,i) \neq \emptyset$, then $S(e,i)$ is the disjoint union of the sets 
                \begin{equation*}
                \{e_{\min}\}, S(e_{\min}^{(i+1)},i), S(e_{\min}^{(i+2)},i+1), \ldots, S(e_{\min}^{(|u|)},|u|-1)\,,
                \end{equation*} 
                where $e_{\min}$ is the $e$-minimal $\mathcal C$-embedding of $u$ in $\doc$.
        \end{itemize} 
    \end{claim}

    \begin{proof}
We start with the \emph{if-direction} of the first bullet point. If $S(e,0)=\emptyset$, then $S(e,i)$ is empty by definition. Therefore, let us assume in the following that the $e$-minimal $\mathcal C$-embedding $e_{\min}$ of $u$ in $\doc$ exists (and thus is an element of $S(e,0)$). If there is some $k\in[i]$ such that $e_{\min}(k)>e(k)$, then, by minimality of $e_{\min}$, there cannot be a $\mathcal C$-embedding of $u$ in $\doc$ that agrees with $e$ on positions $1,\ldots,i$. For the \emph{only-if-direction}, let us assume that $S(e,i)=\emptyset$, which also means that $e_{\min}\notin S(e,i)$. Hence, there must be a $k\in[i]$ with $e_{\min}(k)>e(k)$.

        Now, let us consider the second bullet point. By assumption $S(e,i)\neq\emptyset$; thus, using the first bullet point, the first $i$ positions of $e_{\min}$ must coincide with $e$. This means that $e_{\min} \in S(e,i)$. Further, by the minimality of $e_{\min}$, $e_{\min}\leq e'$ must hold for all $e'\in S(e,i)$ and thus $S(e,i)=S(e_{\min},i)$. By definition of $e^{(j)}_{\min}$, any set $S(e^{(j)}_{\min}, j-1)$ with $j\in[i+1,|u|]$ must be contained in $S(e,i)$. This proves the inclusion from right to left.

        For the inclusion from left to right, we first observe that $e_{\min} \in S(e,i)$ is included in the right side by definition. Let us consider some embedding $e'\in S(e,i)$ with $e'\neq e_{\min}$, then $e_{\min}<e'$ and there must be a minimal $j\in[i+1,|u|]$ such that $e'(j)>e_{\min}(j)$. Thus, $e_{\min}^{(j)}\leq e'$ and $e'(k)=e_{\min}(k)$ for all $k\in[j-1]$, so $e'$ must be contained in $S(e^{(j)}_{\min},j-1)$. This proves the inclusion from left to right. 

        For disjointness, let $e' \in S(e_{\min}^{(j)},j-1) \cap S(e_{\min}^{(j')},j'-1)$ with $j \neq j'\in[i+1,|u|]$. If $j < j'$, then $e' \in S(e_{\min}^{(j)},j-1)$ implies $e'(j) > e_{\min}(j)$ and $e' \in S(e_{\min}^{(j')},j'-1)$ implies $e'(j) = e_{\min}(j)$, which is a contradiction.
    \end{proof}
    
\paragraph{The \textsf{nextMove} Procedure}
    To efficiently test which subsets $S(e_{\min}^{(j)},j-1)$, for $j\in[i+1,|u|]$, defined in Claim~\ref{claim:disjointsets}, are non-empty, we use the following procedure, which will be crucial for the running time of Algorithm~\ref{algorithm:findall}. For every $i\in\{0,1,\ldots,|u|-1\}$ and every $\mathcal C$-embedding $e\colon[|u|]\to[|\doc|]$ of $u$ in $\doc$, the procedure $\nextMove(e,i)$ returns the minimal $j\in[i+1,|u|]$ such that the subset $S(e^{(j)},j-1)$ of $S(e,i)$ is not empty, or returns $\bot$ if there is no such $j$. Na\"ively, as implemented in Algorithm~\ref{algo:nextMoveLCON}, this can be done in time $O(|u|\cdot T(|u|, |\doc|, \lVert \mathcal{C} \rVert))$ by testing for every $j=i+1,\ldots,|u|$ whether there is an $e^{(j)}$-minimal $\mathcal C$-embedding of $u$ in $\doc$ that agrees with $e$ on the first $j-1$ positions and returning the first $j$ for which this condition is true, or $\bot$ if there is no such $j$.
    It is easy to see that this implementation runs in time $O((j-i)\cdot T(|u|, |\doc|, \lVert \mathcal{C} \rVert))$ if $\nextMove(e,i)=j$, and in time $O((|u|-i)\cdot T(|u|, |\doc|, \lVert \mathcal{C} \rVert))$ if $\nextMove(e,i)=\bot$.

    \begin{algorithm}
    \caption{$\nextMove(e,i)$}\label{algo:nextMoveLCON}
    \KwIn{mapping $e\colon[|u|]\rightarrow[|\doc|]$, positions $1,\ldots,i$ locked for $i\in\{0,1,\ldots,|u|-1\}$.}
    \KwOut{minimal $j\in[i+1,|u|]$ s.t. $S(e^{(j)},j-1)\neq\emptyset$, or $\bot$ if there is none.}
    \For{$j=i+1$ \KwTo $|u|$}{
        $e^{(j)}\gets\incr(e,j)$\;
          $e_j^\ast\gets \embedLCONSubseq(u,\mathcal C,\doc,e^{(j)})$\;
          \If{$(e_j^\ast\neq \bot)\wedge (\forall\: k\in[j-1]\colon e_j^\ast(k)=e^{(j)}(k))$}{
              \KwRet{$j$}\;}
    }
    \KwRet{$\bot$}\;
\end{algorithm}

    Our algorithm (see Algorithm~\ref{algorithm:findall}) now works by recursively computing, given the mapping $e$ and locked positions $1,\ldots,i$ for some $i\in\{0,1,\ldots,|u|-1\}$, the $e$-minimal $\mathcal C$-embedding $e_{\min}$ of $u$ in $\doc$, and then partitioning the remaining space of solutions (i.e., $\mathcal C$-embeddings) into disjoint subsets. In particular, we compute $S(e,i)$ as follows: First, compute $e_{\min}$ as the $e$-minimal $\mathcal C$-embedding of $u$ in $\doc$. If, as in the first case of Claim~\ref{claim:disjointsets}, $e_{\min}=\bot$ or $e_{\min}(k)>e(k)$ for some position $k\in[i]$, then $S(e,i)$ is empty.\footnote{Note that this can only occur in the initial call to compute $S(e_0,0)$, where $e_0(k)=k$ for all $k\in[|u|]$. In all other cases we only recursively call the algorithm to compute $S(e,i)$ if we already determined that this set is non-empty, and thus $e_{\min}\in S(e,i)$.} Otherwise we recursively compute the non-empty sets $S(e_{\min}^{(j)}, j-1)$ for $j\in[i+1,|u|]$. Finding all values $j\in[i+1,|u|]$ with corresponding non-empty set $S(e_{\min}^{(j)}, j-1)$ can be done using successive calls of the $\nextMove$ procedure: first, we determine $j'=\nextMove(e_{\min},i)$ corresponding to the minimal (w.r.t. $j'$) non-empty subset, and then iteratively find the next such subset by updating $j'=\nextMove(e_{\min},j')$ (thus finding the smallest non-empty subset of $S(e_{\min},j')\subseteq S(e_{\min},i)$), until the procedure returns $\bot$, signaling that we found all subsets. 
     
     Initially, we call the procedure for $S(e_0,0)$, where $e_0\colon[|u|]\rightarrow[|\doc|]$ is the trivial mapping with $e_0(k)=k$ for all $k\in[|u|]$, and thus compute all $\mathcal C$-embeddings of $u$ in $\doc$ (since there can be no $\mathcal C$-embedding $e'$ of $u$ in $\doc$ with $e'<e_0$). 
     If the first call to $\embedLCONSubseq$ indicates that there is no embedding, then we stop the computation immediately.

Otherwise, given a mapping $e$ with locked positions $1,\ldots,i$ for some $i\in\{0,\ldots,|u|-1\}$, we can compute the $e$-minimal $\mathcal C$-embedding $e_{\min}$ of $u$ in $\doc$ by calling Algorithm~\ref{mainAlgo} in time $O(T(|u|, |\doc|, \lVert \mathcal{C} \rVert))$.

Then, we have to determine all non-empty subsets of $S(e_{\min},i)=S(e,i)$ using calls to $\nextMove$. If $\{j_1,\ldots,j_r\}\subseteq [i+1,|u|]$ are the positions corresponding to the non-empty subsets (i.e., $S(e_{\min}^{(j_\ell)},j_\ell-1)\neq\emptyset$ for all $\ell\in[r]$), then we have exactly $r+1$ calls of $\nextMove(e_{\min},\cdot)$, namely with the parameters $i,j_1,\ldots,j_r$, where $\nextMove(e_{\min},j_\ell)=j_{\ell+1}$ for all $\ell\in[r-1]$, $\nextMove(e_{\min},i)=j_1$, and $\nextMove(e_{\min},j_r)=\bot$. The time needed to execute these calls adds up to a telescopic sum, resulting in $O(|u|\cdot T(|u|, |\doc|, \lVert \mathcal{C} \rVert))$ time for each recursive call. We can then recursively generate all elements of the sets $S(e_{\min}^{(j_\ell)},j_{\ell}-1)$ for $\ell\in[r]$. Overall, since we have exactly one recursive call for every embedding $e\in\query(\doc)$ (in which this embedding is also output), we produce all elements of $\query(\doc)$ in time $O(|\query(\doc)|\cdot|u|\cdot T(|u|, |\doc|, \lVert \mathcal{C} \rVert))$.
This concludes the proof of Theorem~\ref{theorem:compvariantAppendix}.
\end{proof}

We can argue that we can use the algorithm computing the set $\query(\doc)$ to directly obtain an enumeration algorithm for $\query(\doc)$.

\begin{theorem}
Let $u\in\Sigma^\ast$ be a query string with a set $\mathcal C$ of gap-constraints for $u$ such that $\mathcal C$ only contains left-convex gap-constraints, and let $\doc\in\Sigma^\ast$ with $|u|\leq|\doc|$. 
Then we can enumerate the set $\query(\doc)$ of all $\mathcal C$-embeddings of $u$ in $\doc$ with $O(|\doc| (|u| + \lVert \mathcal C\rVert))$ preprocessing time and $O(|u|\cdot |\doc| (|u| + \lVert \mathcal C\rVert))$ delay.
\end{theorem}

\begin{proof}
  We modify Algorithm~\ref{algorithm:findall} to make sure that the last
  recursive call performed (if any) is tail-recursive, as described in Algorithm \ref{algorithm:tailrec}. 
  \begin{algorithm}
    \caption{$\embedAll(u,\mathcal C, \doc, e, i)$}\label{algorithm:tailrec}
    \KwIn{$u,\mathcal C$ s.t. $\mathcal C$ only contains $\LCON$-constraints, $\doc\in\Sigma^\ast$, $e\colon[|u|]\rightarrow[|\doc|]$, positions $1,\ldots,i$ locked for $i\in\{0,1,\ldots,|u|-1\}$.}%
    \KwOut{output all embeddings from $S(e,i)$.
    }
    $e^\ast \gets \embedLCONSubseq(u,\mathcal C,\doc,e)$\;
    \label{tailrec:firstreccall}
    \lIf{$e^\ast=\bot$}{\Return}
    output $e^\ast$\; \label{tailrec:findalloutput}
    $j\gets\nextMove(e^\ast,i)$\; \label{tailrec:firstnextmove}
    \lIf{$j=\bot$}{\textbf{return}}
    \While{$\nextMove(e^\ast,j)\neq\bot$\label{tailrec:loopcond}}{
        $(e^\ast)^{(j)}\gets\incr(e^\ast,j)$\; 
        $\embedAll(u,\mathcal C,\doc, (e^\ast)^{(j)},j-1)$\;
        $j\gets \nextMove(e^\ast,j)$\; \label{tailrec:secondnextmove}
    }
    $(e^\ast)^{(j)}\gets\incr(e^\ast,j)$\;\label{tailrec:lastmoveA} 
    $\embedAll(u,\mathcal C,\doc, (e^\ast)^{(j)},j-1)$\;\label{tailrec:lastmoveB}
\end{algorithm}

  To ensure this, we basically unravel the last iteration of the while-loop of Line \ref{tailrec:loopcond}. This is achieved by simply changing the condition from Line \ref{lin:whilecondition} of Algorithm~\ref{algorithm:findall} (Line \ref{tailrec:loopcond} in Algorithm \ref{algorithm:tailrec}) from $j\neq \bot$ to $\nextMove(e^\ast,j)\neq\bot$, and treating the case when $\nextMove(e^\ast,j)=\bot$ separately, as the case that corresponds to the tail call (which will not be added to the stack of function calls); as long as $\nextMove(e^\ast,j)\neq\bot$, there will be at least one subsequent call and we deal with that in a standard recursive call.

  The correctness of the algorithm follows from the fact that it correctly
  computes the solution set, with no duplicates, as we had already argued: it restates, in a straightforward manner, Algorithm~\ref{algorithm:findall} with the last iteration of the while-loop executed separately (since, in that iteration, Line \ref{tailrec:secondnextmove} does not need to be executed). The
  preprocessing time bound is immediate. Note that the preprocessing covers
  in particular the root-level call to $\embedAll$, so we can in
  particular use it to detect whether the solution space is empty and abort
  otherwise. Let us now argue that the delay bound is respected. 

  For this, we observe that our modification to the algorithm ensures
  ($\dagger$): every recursive call to $\embedAll$ will
  produce one solution, after a time of $O(T(|u|, |\doc|, \lVert \mathcal{C}
  \rVert))$, namely, the complexity of the call at
  Line~\ref{tailrec:firstreccall}. (The fact that this is true already at the
  root-level call is ensured because we use the preprocessing to handle the
  case where there are no solutions at all.)

  From then on, when considering one given call of $\embedAll$, there are two possibilities. The
  first one is that we do no recursive call.
  The other possibility is that we do
  recursive calls,
  with the last call being tail-recursive. Let us
  argue that the delay bound is respected by considering the current
  recursive call and bounding the number of computation steps after the
  solution that it outputs at Line~\ref{tailrec:findalloutput}, until either the
  next output or the end of the enumeration.

  In the first case, we conclude after time $O(|u|\cdot T(|u|, |\doc|, \lVert
  \mathcal{C} \rVert))$ (namely, the complexity of
  Algorithm~\ref{algo:nextMoveLCON}) that no recursive calls are necessary, and
  we exit the current call.
  Thanks to the use of tail-recursion, either this terminates the overall
  algorithm, or it returns to a parent call in which some subsequent recursive call
  still needed to be made (i.e., the 
  parent call
  was not
  tail-recursive). In the latter case, after time $O(|u|\cdot T(|u|, |\doc|, \lVert
  \mathcal{C} \rVert))$, we know that a subsequent recursive call will be made,
  and by invariant ($\dagger$) we know that this call will produce a result after time $O(T(|u|, |\doc|, \lVert
  \mathcal{C} \rVert))$. Thus the overall delay satisfies our bound of $O(|u|\cdot T(|u|, |\doc|, \lVert
  \mathcal{C} \rVert))$.

  In the second case, we conclude after time $O(|u|\cdot T(|u|, |\doc|, \lVert
  \mathcal{C} \rVert))$ that some recursive call needs to be made. Note that
  this includes two invocations of Algorithm~\ref{algo:nextMoveLCON}: the
  one at Line~\ref{tailrec:firstnextmove}, and the one done in anticipation to check
  whether the first recursive call is tail-recursive (which is at
  Line~\ref{tailrec:loopcond}): however
  this just amounts to a factor of two, so the time since the output is still 
  $O(|u|\cdot T(|u|, |\doc|, \lVert \mathcal{C} \rVert))$. We then
  do the first
  recursive call to $\embedAll$ of the current call, which by invariant ($\dagger$) will produce a result after time $O(T(|u|, |\doc|, \lVert
  \mathcal{C} \rVert))$. 
  Hence, the delay bound is respected
  again. The time spent between two recursive calls originating in two consecutive iterations of the while-loop (and, thus, between two consecutive outputs caused by these calls) is, by similar arguments, $O(|u|\cdot T(|u|, |\doc|, \lVert
  \mathcal{C} \rVert))$. This concludes the proof.
\end{proof}

As we have seen in the proof of Theorem~\ref{theorem:compvariantAppendix}, the running time of Algorithm~\ref{algorithm:findall} depends significantly on the implementation of the $\nextMove$ procedure. We can now show that if $\mathcal C$ consists of only constraints that are not just left-convex but also right-convex, then we can implement $\nextMove$ more efficiently, leading to the following, improved result. Recall that $\RCON$ denotes the class of right-convex languages (see Appendix~\ref{sec:leftConvexProp}).

\lrconEnumRest*

\begin{proof}
For simplicity, let us again denote by $T(|u|, |\doc|, \lVert \mathcal{C} \rVert)$ the running time of Algorithm~\ref{mainAlgo}, i.e., $O(|\doc| (|u| + \lVert \mathcal C\rVert))$.

    The general approach follows Algorithm~\ref{algorithm:findall}
    (Theorem~\ref{theorem:compvariantAppendix}), with the same modifications
    to ensure that the last recursive call is tail-recursive. However, we can now exploit that every language $L$ is not just left- but also right-convex, and use this to find the next non-empty subset $S(e^{(j)},j-1)$ in time $O(T(|u|, |\doc|, \lVert \mathcal{C} \rVert))$, instead of time $O(|u|\cdot T(|u|, |\doc|, \lVert \mathcal{C} \rVert))$. Recall that a language $L$ is right-convex if and only if its mirror image $L^R=\{w[|w|]w[|w|-1]\cdots w[1] \mid w\in L\}$ is left-convex.
    
    Analogously to minimal embeddings, given mapping $e_0\colon[|u|]\to[|\doc|]$, we say an embedding $e$ is an \emph{$e_0$-maximal} $\mathcal C$-embedding of $u$ in $\doc$ if and only if $e$ is a $\mathcal C$-embedding of $u$ in $\doc$ with $e\leq e_0$ and $e$ is maximal within the set of all $\mathcal C$-embeddings $e'$ of $u$ in $\doc$ with $e'\leq e_0$. Analogously to Lemma~\ref{uniquenessLemma} for left-convex languages, if $\mathcal C$ consists of only right-convex gap-constraints, then the $e_0$-maximal $\mathcal C$-embedding of $u$ in $\doc$ is unique (if it exists). 
    By symmetry, it is immediate that, given strings $u,\doc$, a set $\mathcal C$ of right-convex gap-constraints for $u$, and a mapping $e_0$, we can compute the $e_0$-maximal $\mathcal C$-embedding of $u$ in $\doc$ in time $O(T(|u|, |\doc|, \lVert \mathcal{C} \rVert))$. Let us call this procedure $\mathsf{EmbedRCONMax}$ (see an implementation in Algorithm~\ref{algo:rconmax}).

\begin{algorithm}
    \caption{$\mathsf{EmbedRCONMax}(u,\mathcal C,\doc,e_0)$}\label{algo:rconmax}
    \KwIn{$u,\mathcal C$ s.t. $\mathcal C$ contains only $\RCON$-constraints, $\doc\in\Sigma^\ast$.}
    \KwOut{$\mathcal C$-embedding $e$ of $u$ in $\doc$, s.t. $e\leq e_0$ and $e$ is maximal with this property, if it exists, and $\bot$ otherwise.}
    set $\mathcal C^R$ as the reverse of $\mathcal C$, i.e., $\mathcal C^R \gets \{(|u|-j+1,|u|-i+1,L^R)\mid (i,j,L)\in\mathcal C\}$\;
    set $e_0^R$ as the reverse of $e_0$, i.e., $e_0^R \colon [|u^R|]\rightarrow [|\doc^R|]$ s.t. $e_0^R(k)=|\doc|-e_0(|u|-k+1)+1$ for all $k\in[|u|]$\;
    $e^R \gets \embedLCONSubseq(u^R,\mathcal C^R, \doc^R, e_0^R)$\;
    \lIf{$e^R=\bot$}{\KwRet{$\bot$}}
    set $e$ as the reverse of $e^R$\;
    \KwRet{$e$}\;
\end{algorithm}

    We recall the $\nextMove$ procedure. Given $\mathcal C$-embedding $e$ and $i\in\{0,1,\ldots,|u|-1\}$, $\nextMove(e,i)$ computes the minimal $j\in[i+1,|u|]$ such that $S(e^{(j)},j-1)\neq\emptyset$, and returns $\bot$ if there is no such $j$. As we have seen in the proof of Theorem~\ref{theorem:compvariantAppendix}, if $\mathcal C$ consists of left-convex languages, this can be implemented in $O((j-i)\cdot T(|u|, |\doc|, \lVert \mathcal{C} \rVert))$ time. 
    
    However, if $\mathcal C$ consists of only constraints in $\LCON\cap\RCON$, then we can compute $\nextMove(e,i)$ faster by analysing the $e'$-maximal $\mathcal C$-embedding $e_{\max}$ of $u$ in $\doc$ (computed using $\mathsf{EmbedRCONMax}$), where $e'(k)=e(k)$ for all $k\in[i]$ and $e'(k)=|\doc|-(|u|-k)$ for all $k\in[i+1,|u|]$. Intuitively, $e'$ coincides with $e$ for all locked positions and all other positions are maximal (i.e., $e(|u|)=|\doc|, e(|u|-1)=|\doc|-1,\ldots$), in order to not restrict their placement in $e_{\max}$. Thus, when computing the maximal embedding, we only consider embeddings $e''$ with $e''(k)\leq e(k)$ for all $k\in[i]$, and, because $e$ itself is a $\mathcal C$-embedding of $u$ in $\doc$, $e\leq e_{\max}$ and $e(k)=e_{\max}(k)$ for all $i\in[|u|]$ must hold. Now, $\nextMove(e,i)$ returns the smallest $j\in[i+1,|u|]$ with $e(j)<e_{\max}(j)$, or $\bot$ if $e=e_{\max}$. See Algorithm~\ref{algo:nextmoveLCONRCON} for an implementation.

\begin{algorithm}
    \caption{$\nextMoveLRCON(e,i)$}\label{algo:nextmoveLCONRCON}
    \KwIn{$\mathcal C$-embedding $e$ of $u$ in $\doc$, positions $1,\ldots,i$ locked for $i\in\{0,1,\ldots,|u|-1\}$.}
    \KwOut{minimal $j\in[i+1,|u|]$ s.t. $S(e^{(j)},j-1)\neq\emptyset$, or $\bot$ if there is none.}
    \For{$j\in[i]$}{set $e'(j)=e(j)$\;}
    \For{$j\in[i+1,|u|]$}{set $e'(j)=|\doc|-(|u|-j)$\;}
    $e_{\max}\gets \mathsf{EmbedRCONMax}(u,\mathcal C,\doc,e')$\;
    \lIf{$\not\exists\ k\in[i+1,|u|]$ s.t. $e_{\max}(k)>e(k)$}{\KwRet{$\bot$}}
    \KwRet{$\min\{k\in[i+1,|u|]\mid e_{\max}(k)>e(k)\}$}\;
\end{algorithm}

    Let us write $\nextMove_L(e,i)$ for the $O((j-i)\cdot T(|u|, |\doc|, \lVert \mathcal{C} \rVert))$ time na\"ive implementation for left-convex languages, and $\nextMove_{LR}(e,i)$ for the variant that requires all languages to be both left- and right-convex. We can now argue that $\nextMove_L(e,i)=\nextMove_{LR}(e,i)$ for all $\mathcal C$-embeddings $e$ of $u$ in $\doc$ and $i\in\{0,1,\ldots,|u|-1\}$. 

    \begin{claim}\label{claim:nextMoveRCON}
        Let $e\colon[|u|]\to[|\doc|]$ be a $\mathcal C$-embedding of $u$ in $\doc$ and $i\in\{0,1,\ldots,|u|-1\}$.
        \begin{itemize}[nosep]
            \item $S(e,i)=\{e\}$ if and only if $\nextMove_{LR}(e,i)=\bot$. 
            \item If $\nextMove_{LR}(e,i)=j\in[i+1,|u|]$ then $S(e^{(j)},j-1)$ is the smallest (w.r.t. $j$) non-empty subset of $S(e,i)\setminus\{e\}$.
        \end{itemize}
    \end{claim}

    \begin{proof}
        Let $e'\colon[|u|]\to[|\doc|]$ be a mapping with $e'(k)=e(k)$ for all $k\in[i]$ and $e'(k)=|\doc|-(|u|-k)$ for all $k\in[i+1,|u|]$, like in the computation of $\nextMove_{LR}(e,i)$. We will first show that the $e'$-maximal $\mathcal C$-embedding $e_{\max}$ of $u$ in $\doc$ is the maximal element of $S(e,i)$. By definition, $e_{\max}\leq e'$ and thus $e_{\max}(k)\leq e'(k)=e(k)$ for all $k\in[i]$. Since all other positions of $e'$ are shifted as far as possible to the right (there cannot be an embedding $e''$ with $e''(k)>|\doc|-(|u|-k)$ for any $k\in[|u|]$, as this would contradict the order condition 
        $1\leq e''(1)<e''(2)<\ldots<e''(|u|)\leq|\doc|$), and because $e$ is a possible candidate for $e_{\max}$, we have that $e_{\max}(k)\geq e(k)$ for all $k\in[i+1,|u|]$. Thus, $e_{\max}\in S(e,i)$. Assume now that there is an $e''\in S(e,i)$ with $e_{\max}\leq e''$. By definition of $S(e,i)$, $e''$ must satisfy $e''(k)=e(k)=e_{\max}(k)$ for all $k\in[i]$ and $e''(k)>e_{\max}(k)$ for some $k\in[i+1,|u|]$. However, this contradicts that $e_{\max}$ is $e'$-maximal. Therefore, $e_{\max}$ must be the maximal element of $S(e,i)$. 

        Now, if $e_{\max}=e$, then the minimal and maximal element of $S(e,i)$ coincide and there cannot be another embedding $e''\in S(e,i)$ with $e''\neq e$. Since $\nextMove_{LR}(e,i)$ returns $\bot$ if and only if $e_{\max}=e$, the procedure correctly determines if there is no non-empty subset of $S(e,i)\setminus\{e\}=\emptyset$. 

        On the other hand, assume that $S(e,i)\neq\{e\}$. Then $e_{\max}\neq e$ and $e_{\max}\in S(e,i)$, i.e., $S(e,i)\setminus\{e\}$ must contain at least one element. Therefore, there is at least one non-empty subset $S(e^{(j)},j-1)$ for $j\in[i+1,|u|]$, namely the one containing $e_{\max}$. If $j\in[i+1,|u|]$ is the minimal position such that $e_{\max}(j)>e(j)$, then, by definition, $e_{\max}\in S(e^{(j)},j-1)$. Because $e_{\max}$ is $e'$-maximal, there can be no non-empty subset $S(e^{(k)},k-1)$ for $k\in[i+1,j-1]$: If there was some embedding $e''\in S(e^{(k)},k-1)$, then $e''(k)>e_{\max}(k)$ would contradict that $e_{\max}$ is the unique $e'$-maximal $\mathcal C$-embedding of $u$ in $\doc$. 
        Therefore, $\nextMove(e,i)$ correctly returns value $j$ if $S(e^{(j)},j-1)$ is the smallest (w.r.t. $j$) non-empty subset of $S(e,i)\setminus\{e\}$.
    \end{proof}

    Given some $\mathcal C$-embedding $e$ of $u$ in $\doc$ and $i\in\{0,1,\ldots,|u|\}$, the time needed to compute $\nextMove_{LR}(e,i)$ is $O(T(|u|, |\doc|, \lVert \mathcal{C} \rVert))$: We first compute the $e'$-maximal $\mathcal C$-embedding $e_{\max}$ of $u$ in $\doc$ in time $O(T(|u|, |\doc|, \lVert \mathcal{C} \rVert))$ and then determine the smallest $j\in[i+1,|u|]$ with $e_{\max}(j)>e(j)$, which takes $O(|u|)$ time. Since $S(e,j')\subseteq S(e,j)$ holds for all mappings $e$ and $j<j'$, we compute all $j\in[i+1,|u|]$ such that $S(e^{(j)},j-1)\neq\emptyset$ as follows: If $\nextMove_{LR}(e,i)=\bot$, there is no such $j$. Otherwise, let $j_1=\nextMove_{LR}(e,i)$, then we compute $j_{\ell}=\nextMove_{LR}(e,j_{\ell-1})$ for increasing $\ell$ until $j_{\ell}=\bot$. Now, $S(e^{(k)},k-1)\neq\emptyset$ if and only if $k\in\{j_1,\ldots,j_{\ell-1}\}$ by Claim~\ref{claim:nextMoveRCON}.
    
    The rest of the procedure follows Algorithm~\ref{algorithm:findall} 
    but replaces 
    the $\nextMove_L$ function 
    (Algorithm~\ref{algo:nextMoveLCON}),
    having a running time of $O(|u|\cdot T(|u|, |\doc|,
    \lVert \mathcal C \rVert))$, by 
    the $\nextMove_{LR}$ function 
    (Algorithm~\ref{algo:nextmoveLCONRCON}),
    having a running time of $O(T(|u|, |\doc|,
    \lVert \mathcal C \rVert))$.
    Thus, the correctness and complexity are immediate by a similar analysis
    to that of the proof of Algorithm~\ref{algorithm:findall}, concluding the
    proof of Theorem~\ref{thm:lrconEnum}.
\end{proof}

\section{Full Details for Section~\ref{sec:hardness}}\label{sec:hardnessDetails}

In this section, we give full proofs for the hardness results of Section~\ref{sec:hardness}.

\aaLangThm*

\begin{proof}
  Let $F = \{c_1, c_2, \ldots, c_m\}$ be a $3$-CNF formula, where every conjunction $c_j = \{l_{j, 1}, l_{j, 2}, l_{j, 3}\} \subseteq \{x_1, \neg x_1, x_2, \neg x_2, \ldots, x_n, \neg x_n\}$ is a clause with three literals. We will construct an instance of the matching problem with length constraints and $L$-constraints as strings $u, \doc \in \Sigma^*$ and a set of gap-constraints. We will define $u \coloneq \widehat{u} \: \widetilde{u}$ and $\doc = \widehat{\doc} \: \widetilde{\doc}$, with $\widehat{u}$ and $\widehat{\doc}$ intuitively describing a choice of Boolean values for each variable of $\{x_1, \ldots, x_n\}$ (each of them will formally be the concatenation of $x_i$-assignment gadgets for each variable $x_i$); and with $\widetilde{u}$ and $\widetilde{\doc}$ intuitively verifying the satisfaction of each clause (each of them will formally be the concatenation of $c_j$-clause gadgets for each clause $c_j$). We will explain how length constraints can be used to ensure that each assignment gadget in $\widehat{u}$ is mapped in the corresponding assignment gadget in $\widehat{\doc}$, and likewise for clause gadgets.

The rest of the proof is structured as follows. First, we define assignment gadgets. Then, we define clause gadgets. Afterwards, we explain how length constraints allow us to synchronise the gadgets. Then, we explain how $L$-constraints are used to enforce the semantics of the clause gadgets; this finishes the definition of the instance of the matching problem that we are reducing to, i.e., concludes the presentation of the reduction. Finally, we show that the reduction is correct. We now go through these successive steps.

\paragraph{Assignment gadgets}
For every $i \in \{1, 2, \ldots, n\}$, the \emph{$x_i$-assignment gadget} is the following (where the $0$'s and $1$'s below correspond to an intuitive explanation given just below):

\begin{tabular}{ccccc}
$\widehat{u}_i =$ & $\tb$ & $\ta$ & $\ta$ & $\tb$ \\
$\widehat{\doc}_i =$ & $\tb$ & $\ta \ta$ & $\ta \ta$ & $\tb$ \\
$ $ & & $0 1$ & $01$ & 
\end{tabular}

  We interpret the first $\ta$ of $\widehat{u}_i$ to correspond to $x_i$ and the second $\ta$ of $\widehat{u}_i$ to correspond to $\neg x_i$. Furthermore, the occurrences of $\ta$ of $\widehat{\doc}_i$ are interpreted as either $0$ or $1$ as illustrated above, which indicate whether the literal mapped to this position is true or false. Consequently, embedding $\widehat{u}_i$ into $\widehat{\doc}_i$ can be interpreted as mapping $x_i$ to either $0$ or $1$ and mapping $\neg x_i$ to either $0$ or $1$. By using gap-constraints, we want to enforce that whenever $\widehat{u}_i$ is embedded into $\widehat{\doc}_i$, then this can only happen in one of the following two ways (where the $0$'s and $1$'s drawn below illustrate the intuitive intended semantics of the gadgets):

\begin{tabular}{cccccc}
$\tb$ & & $\ta$ & $\ta$ & & $\tb$ \\
$\tb$ & $\ta$ & $\ta$ & $\ta$ & $\ta$ & $\tb$ \\
& $0$ & $1$ & $0$ & $1$ &
\end{tabular}
\hspace{1cm}
or
\hspace{1cm}
\begin{tabular}{cccccc}
$\tb$ & $\ta$ &  & & $\ta$ & $\tb$ \\
$\tb$ & $\ta$ & $\ta$ & $\ta$ & $\ta$ & $\tb$ \\
& $0$ & $1$ & $0$ & $1$ &
\end{tabular}

The first embedding is interpreted as assigning $x_i$ to $1$ (and accordingly assigning $\neg x_i$ to $0$), while the second one is interpreted as assigning $x_i$ to $0$ (and accordingly assigning $\neg x_i$ to $1$). 

We enforce that the two embeddings depicted above are the only possible ones as follows. We use an $L$-constraint between the two $\ta$-occurrences in $\widehat{u}_i$, which excludes embeddings that would map both $x_i$ and $\neg x_i$ to $0$ or both $x_i$ and $\neg x_i$ to $1$. Then, in order to exclude embeddings that map both $\ta$-occurrences of $\widehat{u}_i$ to the two leftmost $\ta$-occurrences of $\widehat{\doc}_i$ or both $\ta$-occurrences of $\widehat{u}_i$ to the two rightmost $\ta$-occurrences of $\widehat{\doc}_i$, we use a $[0, 1]$-length constraint between the first $\tb$-occurrence and the first $\ta$-occurrence in $\widehat{u}_i$ and a $[0, 1]$-length constraint between the second $\ta$-occurrence and the second $\tb$-occurrence in $\widehat{u}_i$. It can be easily seen that these constraints are only satisfied by the two embeddings illustrated above.

We define $\widehat{u} = \widehat{u}_1 \widehat{u}_2 \cdots \widehat{u}_n$ and $\widehat{\doc} = \widehat{\doc}_1 \widehat{\doc}_2 \cdots \widehat{\doc}_n$, and we note that if $\widehat{u}$ can be embedded into $\widehat{\doc}$, then this means that every $\widehat{u}_i$ is embedded into $\widehat{\doc}_i$, which, as observed above, induces an assignment $\pi \colon \{x_1, x_2, \ldots, x_n\} \to \{0, 1\}$.

\paragraph{Clause gadgets}
For every $j \in \{1, 2, \ldots, m\}$, the \emph{$c_j$-clause gadget} is defined as follows:

\begin{tabular}{ccccccccccc}
 &  & $L_1$ & $S_1$  & $L_2$ & & $S_2$ & $L_{\vee}$  & $L_3$ & \\
$\widetilde{u}_j$ = & $\tb$ & $\ta$ & $\ta$ & $\ta$ & $\tb$ & $\ta$ & $\ta$ & $\ta$ & $\tb$ \\
$\widetilde{\doc}_j$ = & $\tb$ & $\ta \ta$ & $\ta \ta \ta$ & $\ta \ta$ & $\tb$ & $\ta \ta \ta$ & $\ta \ta$  & $\ta \ta$ & $\tb$ \\
&  & $01$ &  & $10$ &  &  & $10$ & $01$ & 
\end{tabular}

 As shown in the alignment above, we mark the $6$ occurrences of $\ta$ in $\widetilde{u}_j$ with $L_1, L_2, L_3$ (representing the three literals $l_{j, 1}, l_{j, 2}$ and $l_{j, 3}$ of clause $c_j$), $L_{\vee}$ (representing the disjunction $l_{j, 1} \vee l_{j, 2}$), and $S_1$ and $S_2$ (used for auxiliary purposes). In particular, we can then talk about the $L_1$-$\ta$, $S_1$-$\ta$, $L_2$-$\ta$, $S_2$-$\ta$, $L_{\vee}$-$\ta$ and $L_3$-$\ta$ of $\widetilde{u}_j$. Moreover, for $\ell \in \{1, 2, 3, \vee\}$, the $L_{\ell}$-$\ta$-block of $\widetilde{\doc}_j$ is the $\ta \ta$-factor aligned with the $L_{\ell}$-$\ta$ (according to the alignment given above), and the $S_{1}$-$\ta$-block and $S_{2}$-$\ta$-block of $\widetilde{\doc}_j$ is the $\ta \ta \ta$-factor aligned with the $S_1$-$\ta$ and $S_2$-$\ta$, respectively.

We define $\widetilde{u} = \widetilde{u}_1 \widetilde{u}_2 \cdots \widetilde{u}_m$ and $\widetilde{\doc} = \widetilde{\doc}_1 \widetilde{\doc}_2 \cdots \widetilde{\doc}_m$, and finally $u = \widehat{u} \: \widetilde{u}$ and $\doc = \widehat{\doc} \: \widetilde{\doc}$. Due to the occurrences of $\tb$, if $u$ can be embedded into $\doc$, then every $\widehat{u}_i$ is embedded into $\widehat{\doc}_i$ and every $\widetilde{u}_j$ is embedded into $\widetilde{\doc}_j$.

\paragraph{Synchronising the gadgets}
In the following, we will add further gap-constraints that enforce that if $u$ can be embedded into $\doc$ (such that all gap-constraints are satisfied), then the assignment $\pi$ induced by the embedding is satisfying, and if there is some satisfying assignment $\pi$, then $u$ can be embedded into $\doc$ (such that all gap-constraints are satisfied). 

The first step is to achieve a synchronisation between the assignment gadgets and the clause gadgets: For every $\ell \in \{1, 2, 3\}$, if $l_{j, \ell} \in \{x_i, \neg x_i\}$, then the $L_\ell$-$\ta$ of $\widetilde{u}_j$ is synchronised with the $x_i$-assignment gadget, i.e., if $l_{j, \ell} = x_i$, then the $L_{\ell}$-$\ta$ of $\widetilde{u}_j$ is mapped to $1$ if and only if $x_i$ is assigned to $1$ (in the sense defined above, i.e., $\widehat{u}_i$ is mapped to $\widehat{\doc}_i$ according to the first of the two possible embeddings), and if $l_{j, \ell} = \neg x_i$, then $L_{\ell}$-$\ta$ is mapped to $1$ if and only if $\neg x_i$ is assigned to $1$.

We first define, for every $i \in \{1, 2, \ldots, n\}$, $j \in \{1, 2, \ldots, m\}$ and $\ell \in \{1, 2, 3\}$, the number $\alpha(i, j, \ell)$ of symbols of $\doc$ that lie strictly between the first $\ta$-occurrence of $\widehat{\doc}_i$ and the $L_{\ell}$-$\ta$-block of $\widetilde{\doc}_j$, and let $\beta(i, j, \ell)$ denote the number of symbols of $\doc$ that lie strictly between the third $\ta$-occurrence of $\widehat{\doc}_i$ and the $L_{\ell}$-$\ta$-block of $\widetilde{\doc}_j$. Obviously, these numbers $\alpha(i, j, \ell)$ and $\beta(i, j, \ell)$ only depend on $i, j$ and $\ell$. 

For every $j \in \{1, 2, \ldots, m\}$, we add the following length constraints. If $l_{j, 1} = x_i$, then we add the length constraint $(\alpha(i, j, 1), \alpha(i, j, 1))$ between the first $\ta$-occurrence of $\widehat{u}_i$ and the $L_1$-$\ta$ of $\widetilde{u}_j$, and if $l_{j, 1} = \neg x_i$, then we add the length constraint $(\beta(i, j, 1), \beta(i, j, 1))$ between the second $\ta$-occurrence of $\widehat{u}_i$ and the $L_1$-$\ta$ of $\widetilde{u}_j$. We proceed analogously with respect to the $L_3$-$\ta$: if $l_{j, 3} = x_i$, then we add the length constraint $(\alpha(i, j, 3), \alpha(i, j, 3))$ between the first $\ta$-occurrence of $\widehat{u}_i$ and the $L_3$-$\ta$ of $\widetilde{u}_j$, and if $l_{j, 3} = \neg x_i$, then we add the length constraint $(\beta(i, j, 3), \beta(i, j, 3))$ between the second $\ta$-occurrence of $\widehat{u}_i$ and the $L_3$-$\ta$ of $\widetilde{u}_j$. With respect to the $L_2$-$\ta$, the situation is slightly different, due to the swapped order of $0$ and $1$ in the $L_2$-$\ta$-block of $\widetilde{\doc}_j$: if $l_{j, 2} = x_i$, then we add the length constraint $(\beta(i, j, 2), \beta(i, j, 2))$ between the second $\ta$-occurrence of $\widehat{u}_i$ and the $L_2$-$\ta$ of $\widetilde{u}_j$, and if $l_{j, 2} = \neg x_i$, then we add the length constraint $(\alpha(i, j, 2), \alpha(i, j, 2))$ between the first $\ta$-occurrence of $\widehat{u}_i$ and the $L_2$-$\ta$ of $\widetilde{u}_j$.

It can be easily seen that these length constraints enforce the desired synchronisation property, e.g., if $l_{j, 1} = x_i$ and the first $\ta$-occurrence of $\widehat{u}_i$ is mapped to $0$, then due to the definition of $\alpha(i, j, 1)$, the $L_1$-$\ta$ of $\widetilde{u}_j$ is also mapped to $0$, and if instead the first $\ta$-occurrence of $\widehat{u}_i$ is mapped to $1$, then the $L_1$-$\ta$ of $\widetilde{u}_j$ is also mapped to $1$. This means that embedding $u$ into $\doc$ induces an assignment $\pi$ such that $\pi(l_{j, \ell}) = 1$ if and only if the $L_{\ell}$-$\ta$ of $\widetilde{u}_j$ is mapped to $1$. 

\paragraph{Enforcing clause satisfaction}
Next, we have to introduce gap-constraints that enforce that $\widetilde{u}_j$ can be embedded into $\widetilde{\doc}_j$ if and only if $\pi$ satisfies clause $c_j$. We start with gap-constraints that enforce that $\pi(l_{j, 1}) \vee \pi(l_{j, 2}) = 0$ implies that the $S_1$-$\ta$ is necessarily mapped to the middle of the $S_1$-$\ta$-block. This can be achieved by using a $[0, 2]$-length constraint between the $L_1$-$\ta$ and the $S_1$-$\ta$, and a $[0, 2]$-length constraint between the $S_1$-$\ta$ and the $L_2$-$\ta$. Now, if $\pi(l_{j, 1}) \vee \pi(l_{j, 2}) = 0$, then the $L_1$-$\ta$ is mapped to $0$ and the $L_2$-$\ta$ is mapped to $0$, which means that mapping $S_1$-$\ta$ to the left or to the right of the $S_1$-$\ta$-block would violate one of these $[0, 2]$-length constraints, while mapping the $S_1$-$\ta$ to the middle satisfies both these $[0, 2]$-length constraints. Moreover, we observe that we can still map $S_1$-$\ta$ in all other constellations, i.e., if $\pi(l_{j, 1}) = 1$ and $\pi(l_{j, 2}) = 0$, then the $S_1$-$\ta$ can be mapped to the right or to the middle of the $S_1$-$\ta$-block, if $\pi(l_{j, 1}) = 0$ and $\pi(l_{j, 2}) = 1$, then the $S_1$-$\ta$ can be mapped to the left or to the middle of the $S_1$-$\ta$-block, and if $\pi(l_{j, 1}) = \pi(l_{j, 2}) = 1$, then the $S_1$-$\ta$ can be mapped to the left or to the right or to the middle of the $S_1$-$\ta$-block. In summary: if $\pi(l_{j, 1}) \vee \pi(l_{j, 2}) = 0$, then $S_1$-$\ta$ must be mapped to the middle of the $S_1$-$\ta$-block, and if $\pi(l_{j, 1}) \vee \pi(l_{j, 2}) = 1$, then it is possible that $S_1$-$\ta$ is \emph{not} mapped to the middle of the $S_1$-$\ta$-block (whether it is mapped to the left or right depends on the actual values of $\pi(l_{j, 1})$ and $\pi(l_{j, 2})$).

We next synchronise the $S_1$-$\ta$ and the $S_2$-$\ta$ such that the $S_1$-$\ta$ is mapped to the left (middle, right) of the $S_1$-$\ta$-block if and only if the $S_2$-$\ta$ is mapped to the left (middle, right) of the $S_2$-$\ta$-block. This can be done by a $[5, 5]$-length constraint between the $S_1$-$\ta$ and the $S_2$-$\ta$. In particular, this means that $\pi(l_{j, 1}) \vee \pi(l_{j, 2}) = 0$ implies that the $S_2$-$\ta$ is necessarily mapped to the middle of the $S_2$-$\ta$-block, and $\pi(l_{j, 1}) \vee \pi(l_{j, 2}) = 1$ implies that it is possible to map $S_2$-$\ta$ 
\emph{not} to the middle of the $S_2$-$\ta$-block.

We also want to achieve that if the $S_2$-$\ta$ is mapped to the left or to the right of the $S_2$-$\ta$-block, then the $L_{\vee}$-$\ta$ is necessarily mapped to $1$, and if the $S_2$-$\ta$ is mapped to the middle of the $S_2$-$\ta$-block, then the $L_{\vee}$-$\ta$ is necessarily mapped to $0$. To this end, we add a $[3, 4]$-length constraint between the second $\tb$ in $\widetilde{u}_j$ and the $L_{\vee}$-$\ta$, which means that the $L_{\vee}$-$\ta$ can only be mapped to the $L_{\vee}$-$\ta$-block. Then, we add an $L$-constraint between the $S_2$-$\ta$ and $L_{\vee}$-$\ta$, which means that if the $S_2$-$\ta$ is mapped to the left or to the right of the $S_2$-$\ta$-block, then $L_{\vee}$-$\ta$ is necessarily mapped to $1$, and if the $S_2$-$\ta$ is mapped to the middle of the $S_2$-$\ta$-block, then $L_{\vee}$-$\ta$ is necessarily mapped to $0$. Hence, if $\pi(l_{j, 1}) \vee \pi(l_{j, 2}) = 0$, then $S_2$-$\ta$ is mapped to the middle of the $S_2$-$\ta$-block and therefore $L_{\vee}$-$\ta$ is mapped to $0$, and if $\pi(l_{j, 1}) \vee \pi(l_{j, 2}) = 1$, then it is possible to map $S_2$-$\ta$ \emph{not} to the middle of the $S_2$-$\ta$-block and therefore $L_{\vee}$-$\ta$ can be mapped to $1$. 

Finally we add a $[1,2]$-length constraint for $L_{\vee}$-$\ta$ and $L_3$-$\ta$. This means that if $L_{\vee}$-$\ta$ is mapped to $0$, then $L_3$-$\ta$ must be mapped to $1$, and if $L_{\vee}$-$\ta$ is mapped to $1$, then $L_3$-$\ta$ can be mapped to $0$ or $1$. 

This concludes the definition of the gap-constraints and therefore concludes the definition of the instance of the matching problem.

\paragraph{Correctness proof}
We now assume that there is a satisfying assignment $\pi$, i.e., $\pi(l_{j, 1}) \vee \pi(l_{j, 2}) \vee \pi(l_{j, 3}) = 1$ for every $j \in \{1, 2, \ldots, m\}$, and we will show that this implies that $u$ can be embedded into $\doc$ such that all gap-constraints are satisfied. First, we embed each $\widehat{u}_i$ into $\widehat{\doc}_i$ as determined by $\pi$, i.e., if $\pi(x_i) = 1$, we use the first embedding mentioned above, and if $\pi(x_i) = 0$, we use the second embedding mentioned above. This describes an embedding of $\widehat{u}$ in $\widehat{\doc}$.

To embed $\widetilde{u}$ into $\widetilde{\doc}$, we explain how we can embed $\widetilde{u}_j$ into $\widetilde{\doc}_j$. For every $\ell \in \{1, 2, 3\}$, we map $L_{\ell}$-$\ta$ to $1$ if $\pi(l_{j, \ell}) = 1$ and to $0$ if $\pi(l_{j, \ell}) = 0$ (this obviously satisfies all the gap-constraints between $\widetilde{u}_j$ and the corresponding assignment gadgets). If $\pi(l_{j, 1}) \vee \pi(l_{j, 2}) = 1$, then, as observed above, we can map the $S_1$-$\ta$ to the left or to the right, depending on the actual values of $\pi(l_{j, 1})$ and $\pi(l_{j, 2})$. This means that we can map the $S_2$-$\ta$ to the left or to the right, and in both cases mapping $L_{\vee}$-$\ta$ to $1$ satisfies the $L$-constraint between $S_2$-$\ta$ and $L_{\vee}$-$\ta$. Furthermore, regardless of where $L_3$-$\ta$ is mapped to, the $[1,2]$-length constraint between $L_{\vee}$-$\ta$ and $L_3$-$\ta$ is satisfied. Consequently, we have an embedding of $\widetilde{u}_j$ into $\widetilde{\doc}_j$ that satisfies all gap-constraints. We have to consider the case that $\pi(l_{j, 1}) \vee \pi(l_{j, 2}) = 0$, which means that $\pi(l_{j, 1}) = 0$ and $\pi(l_{j, 2}) = 0$. In this case, $S_1$-$\ta$ (and therefore $S_2$-$\ta$) can only be mapped to the middle, which means that $L_{\vee}$-$\ta$ must be mapped to $0$. However, since $\pi(l_{j, 1}) \vee \pi(l_{j, 2}) \vee \pi(l_{j, 3}) = 1$, we know that $\pi(l_{j, 3}) = 1$, which means that $L_3$-$\ta$ is mapped to $1$, which satisfies the $[1,2]$-length constraint for $L_{\vee}$-$\ta$ and $L_3$-$\ta$. We conclude that also in this case, we have an embedding of $\widetilde{u}_j$ into $\widetilde{\doc}_j$ that satisfies all gap-constraints.

Assume now that we can embed $u$ into $\doc$ such that all gap-constraints are satisfied, and let $\pi$ be the assignment induced by embedding $\widehat{u}$ into $\widehat{\doc}$. We know that, for every $j \in \{1, 2, \ldots, m\}$, $\widetilde{u}_j$ is embedded into $\widetilde{\doc}_j$ such that all gap-constraints are satisfied. For the sake of a contradiction, let us assume that $\pi(l_{j, 1}) \vee \pi(l_{j, 2}) \vee \pi(l_{j, 3}) = 0$, which means that $\pi(l_{j, 1}) = \pi(l_{j, 2}) = \pi(l_{j, 3}) = 0$. This implies that $L_1$-$\ta$ and $L_2$-$\ta$ are both mapped to $0$, which means that $S_1$-$\ta$ (and therefore $S_2$-$\ta$) is mapped to the middle, which means that $L_{\vee}$-$\ta$ is mapped to $0$. Since $L_3$-$\ta$ is also mapped to $0$, the gap induced by $L_{\vee}$-$\ta$ and $L_3$-$\ta$ has length $0$, which violates the $[1,2]$-length constraint for $L_{\vee}$-$\ta$ and $L_3$-$\ta$. This is a contradiction to our assumption that $\widetilde{u}_j$ is embedded into $\widetilde{\doc}_j$ such that all gap-constraints are satisfied; thus, $\pi(l_{j, 1}) \vee \pi(l_{j, 2}) \vee \pi(l_{j, 3}) = 0$ is not possible. Hence, $\pi$ is satisfying.
\end{proof}

\begin{restatable}{theorem}{abLangThm}
\label{abepsLanguageHardnessTheorem}
The matching problem with length constraints and $\{\ta\tb,\varepsilon\}$-constraints is NP-complete.
\end{restatable}

\begin{proof}
    We use a similar reduction as for Theorem~\ref{aaepsLanguageHardnessTheorem}. 
    Let $F = \{c_1,c_2,\ldots, c_m\}$ be a $3$-CNF formula, where every $c_j = \{l_{j,1}, l_{j,2}, l_{j,3}\}\subseteq \{x_1,\neg x_1, x_2, \neg x_2, \ldots, x_n, \neg x_n\}$ is a clause with three literals. 

We first construct an instance over the alphabet $\{\ta, \tb, \tc\}$ and then explain later how this can be adapted to the binary alphabet $\{\ta, \tb\}$. For every $i \in \{1, 2, \ldots, n\}$, the \emph{$x_i$-assignment gadget} is defined as follows:

    \begin{tabular}{ccccc}
        $\widehat{u}_i =$ & $\tc$ & $\ta$ & $\tb$ & $\tc$ \\
        $\widehat{\doc}_i =$ & $\tc$ & $\ta \ta$ & $\tb \tb$ & $\tc$\\
         & & $0 1$ & $01$ & 
    \end{tabular}

    We interpret the $\ta$ of $\widehat{u}_i$ to correspond to $x_i$ and the $\tb$ of $\widehat{u}_i$ to correspond to $\neg x_i$. Furthermore, the occurrences of $\ta$ and $\tb$ of $\widehat{\doc}_i$ are interpreted as either $0$ or $1$ as illustrated above. Consequently, embedding $\widehat{u}_i$ into $\widehat{\doc}_i$ can be interpreted as mapping $x_i$ to either $0$ or $1$ and mapping $\neg x_i$ to either $0$ or $1$. By using gap-constraints, we want to enforce that whenever $\widehat{u}_i$ is embedded into $\widehat{\doc}_i$, then this can only happen in one of the following two ways:
    
    \begin{tabular}{cccccc}
        $\tc$ & & $\ta$ & $\tb$ & & $\tc$ \\
        $\tc$ & $\ta$ & $\ta$ & $\tb$ & $\tb$ & $\tc$ \\
        & $0$ & $1$ & $0$ & $1$ &
    \end{tabular}
    \hspace{1cm}
    or
    \hspace{1cm}
    \begin{tabular}{cccccc}
        $\tc$ & $\ta$ &  & & $\tb$ & $\tc$ \\
        $\tc$ & $\ta$ & $\ta$ & $\tb$ & $\tb$ & $\tc$ \\
        & $0$ & $1$ & $0$ & $1$ &
    \end{tabular}
    
    The first embedding is interpreted as assigning $x_i$ to $1$ (and accordingly assigning $\neg x_i$ to $0$), while the second one is interpreted as assigning $x_i$ to $0$ (and accordingly assigning $\neg x_i$ to $1$). 

    We enforce that the two embeddings depicted above are the only possible ones as follows. We use an $L$-constraint between the $\ta$- and $\tb$-occurrence in $\widehat{u}_i$, which excludes embeddings that would map both $x_i$ and $\neg x_i$ to $0$ or both $x_i$ and $\neg x_i$ to $1$. It can be easily seen that this constraint is only satisfied by the two embeddings illustrated above.

    As for Theorem~\ref{aaepsLanguageHardnessTheorem}, we define $\widehat{u} = \widehat{u}_1 \widehat{u}_2 \cdots \widehat{u}_n$ and $\widehat{\doc} = \widehat{\doc}_1 \widehat{\doc}_2 \cdots \widehat{\doc}_n$, and we note that if $\widehat{u}$ can be embedded into $\widehat{\doc}$, then this means that every $\widehat{u}_i$ is embedded into $\widehat{\doc}_i$, which, as observed above, induces an assignment $\pi \colon \{x_1, x_2, \ldots, x_n\} \to \{0, 1\}$.

    For every $j \in \{1, 2, \ldots, m\}$, the \emph{$c_j$-clause gadget} is defined as follows:

    \begin{tabular}{ccccccccc}
        & & $L_1$ & $S_1$ & $L_2$ & $L_{\vee}$ & $L_3$ & $S_2$ & \\ 
        $\widetilde{u}_j$ & $\tc$ & $\ta$ & $\ta$ & $\ta$ & $\ta$ & $\ta$ & $\ta\ \tb$ & $\tc$ \\
        $\widetilde{\doc}_j$ & $\tc$ & $\ta\ta$ & $\ta\ta\ta$ & $\ta\ta$ & $\ta\ta$ & $\ta\ta$ & $\ta\tb\ta\ta\tb\tb$ & $\tc$ \\
        & & $01$ & & $10$ & $10$ & $01$ & & 
    \end{tabular}

    As shown in the alignment above, we mark $3$ occurrences of $\ta$ in $\widetilde{u}_j$ with $L_1, L_2, L_3$ (representing the three literals $l_{j,1}, l_{j,2}$ and $l_{j,3}$ of clause $c_j$), one occurrence of $\ta$ with $L _{\vee}$ (representing the disjunction $l_{j,1}\vee l_{j,2}$), one occurrence of $\ta$ with $S_1$ and a factor $\ta \tb$ with $S_2$ (used for auxiliary purposes). In particular, we can then talk about the $L_1$-$\ta$, $S_1$-$\ta$, $L_2$-$\ta$, $L_{\vee}$-$\ta$, $L_3$-$\ta$, $S_2$-$\ta$ and $S_2$-$\tb$ of $\widetilde{u}_j$; note that $S_2$ corresponds to two letters $\ta$ and $\tb$, where the $\ta$ has to occur before the $\tb$ in the same block. Moreover, for $\ell\in\{1,2,3,\vee\}$, the $L_{\ell}$-$\ta$-block of $\widetilde{\doc}_j$ is the $\ta\ta$-factor aligned with the $L_{\ell}$-$\ta$ (according to the alignment given above), and the $S_1$-$\ta$-block of $\widetilde{\doc}_j$ is the $\ta\ta\ta$-factor aligned with the $S_1$-$\ta$. 
    The $S_2$-$\ta$-block (and $S_2$-$\tb$-block, respectively) of $\widetilde{\doc}_j$ is the $\ta\ta\ta$-subsequence ($\tb\tb\tb$-subsequence) of the $\ta\tb\ta\ta\tb\tb$-factor aligned with the $S_2$-$\ta$ ($S_2$-$\tb$).

    We define $\widetilde{u} = \widetilde{u}_1 \widetilde{u}_2 \cdots \widetilde{u}_m$ and $\widetilde{\doc} = \widetilde{\doc}_1 \widetilde{\doc}_2 \cdots \widetilde{\doc}_m$, and finally $u = \widehat{u} \: \widetilde{u}$ and $\doc = \widehat{\doc} \: \widetilde{\doc}$. Due to the occurrences of $\tc$, if $u$ can be embedded into $\doc$, then every $\widehat{u}_i$ is embedded into $\widehat{\doc}_i$ and every $\widetilde{u}_j$ is embedded into $\widetilde{\doc}_j$.

    In the following, we will add further gap-constraints that enforce that if $u$ can be embedded into $\doc$ (such that all gap-constraints are satisfied), then the assignment $\pi$ induced by the embedding is satisfying, and if there is some satisfying assignment $\pi$, then $u$ can be embedded into $\doc$ (such that all gap-constraints are satisfied). 

    The first step is to achieve a synchronisation between the assignment gadgets and the clause gadgets: For every $\ell \in \{1, 2, 3\}$, if $l_{j, \ell} \in \{x_i, \neg x_i\}$, then the $L_\ell$-$\ta$ of $\widetilde{u}_j$ is synchronised with the $x_i$-assignment gadget, i.e., if $l_{j, \ell} = x_i$, then the $L_{\ell}$-$\ta$ of $\widetilde{u}_j$ is mapped to $1$ if and only if $x_i$ is assigned to $1$ (in the sense defined above, i.e., $\widehat{u}_i$ is mapped to $\widehat{\doc}_i$ according to the first of the two possible embeddings), and if $l_{j, \ell} = \neg x_i$, then $L_{\ell}$-$\ta$ is mapped to $1$ if and only if $\neg x_i$ is assigned to $1$.

    We first define, for every $i \in \{1, 2, \ldots, n\}$, $j \in \{1, 2, \ldots, m\}$ and $\ell \in \{1, 2, 3\}$, the number $\alpha(i, j, \ell)$ of symbols of $\doc$ that lie strictly between the first $\ta$-occurrence of $\widehat{\doc}_i$ and the $L_{\ell}$-$\ta$-block of $\widetilde{\doc}_j$, and let $\beta(i, j, \ell)$ denote the number of symbols of $\doc$ that lie strictly between the first $\tb$-occurrence of $\widehat{\doc}_i$ and the $L_{\ell}$-$\ta$-block of $\widetilde{\doc}_j$. Obviously, these numbers $\alpha(i, j, \ell)$ and $\beta(i, j, \ell)$ only depend on $i, j$ and $\ell$. 

    For every $j \in \{1, 2, \ldots, m\}$, we add the following length constraints. If $l_{j, 1} = x_i$, then we add an $[\alpha(i, j, 1), \alpha(i, j, 1)]$-length constraint  between the $\ta$-occurrence of $\widehat{u}_i$ and the $L_1$-$\ta$ of $\widetilde{u}_j$, and if $l_{j, 1} = \neg x_i$, then we add a $[\beta(i, j, 1), \beta(i, j, 1)]$-length constraint between the $\tb$-occurrence of $\widehat{u}_i$ and the $L_1$-$\ta$ of $\widetilde{u}_j$. We proceed analogously with respect to the $L_3$-$\ta$: if $l_{j, 3} = x_i$, then we add an $[\alpha(i, j, 3), \alpha(i, j, 3)]$-length constraint between the $\ta$-occurrence of $\widehat{u}_i$ and the $L_3$-$\ta$ of $\widetilde{u}_j$, and if $l_{j, 3} = \neg x_i$, then we add a $[\beta(i, j, 3), \beta(i, j, 3)]$-length constraint between the $\tb$-occurrence of $\widehat{u}_i$ and the $L_3$-$\ta$ of $\widetilde{u}_j$. With respect to the $L_2$-$\ta$, the situation is slightly different, due to the swapped order of $0$ and $1$ in the $L_2$-$\ta$-block of $\widetilde{\doc}_j$: if $l_{j, 2} = x_i$, then we add a $[\beta(i, j, 2), \beta(i, j, 2)]$-length constraint between the $\tb$-occurrence of $\widehat{u}_i$ and the $L_2$-$\ta$ of $\widetilde{u}_j$, and if $l_{j, 2} = \neg x_i$, then we add an $[\alpha(i, j, 2), \alpha(i, j, 2)]$-length constraint between the $\ta$-occurrence of $\widehat{u}_i$ and the $L_2$-$\ta$ of $\widetilde{u}_j$. 

    It can be easily seen that these length constraints enforce the desired synchronisation property, e.g., if $l_{j, 1} = x_i$ and the $\ta$-occurrence of $\widehat{u}_i$ is mapped to $0$, then due to the definition of $\alpha(i, j, 1)$, the $L_1$-$\ta$ of $\widetilde{u}_j$ is also mapped to $0$, and if instead the $\ta$-occurrence of $\widehat{u}_i$ is mapped to $1$, then the $L_1$-$\ta$ of $\widetilde{u}_j$ is also mapped to $1$. This means that embedding $u$ into $\doc$ induces an assignment $\pi$ such that $\pi(l_{j, \ell}) = 1$ if and only if the $L_{\ell}$-$\ta$ of $\widetilde{u}_j$ is mapped to $1$. 

    In order to enforce that each position of $\widetilde{u}_j$ can only be mapped to its corresponding block in $\widetilde{\doc}_j$, we also add gap-constraints between the first $\tc$ of $\widetilde{u}_j$ and all other positions of $\widetilde{u}_j$, e.g. we add a $[0,1]$-length constraint between the first $\tc$ and the $L_1$-$\ta$, and an $[11,16]$-length constraint from the first $\tc$ to both $S_2$-$\ta$ and $S_2$-$\tb$. 

    Next, we have to introduce gap-constraints that enforce that $\widetilde{u}_j$ can be embedded into $\widetilde{\doc}_j$ if and only if $\pi$ satisfies clause $c_j$. We start with gap-constraints that enforce that $\pi(l_{j, 1}) \vee \pi(l_{j, 2}) = 0$ implies that $S_1$-$\ta$ is necessarily mapped to the middle of the $S_1$-$\ta$-block. This can be achieved by using a $[0, 2]$-length constraint between the $L_1$-$\ta$ and the $S_1$-$\ta$, and a $[0, 2]$-length constraint between the $S_1$-$\ta$ and the $L_2$-$\ta$. Now, if $\pi(l_{j, 1}) \vee \pi(l_{j, 2}) = 0$, then $L_1$-$\ta$ is mapped to $0$ and $L_2$-$\ta$ is mapped to $0$, which means that mapping $S_1$-$\ta$ to the left or to the right of the $S_1$-$\ta$-block would violate one of these $[0, 2]$-length constraints, while mapping $S_1$-$\ta$ to the middle satisfies both these $[0, 2]$-length constraints. Moreover, we observe that we can still map $S_1$-$\ta$ in all other constellations, i.e., if $\pi(l_{j, 1}) = 1$ and $\pi(l_{j, 2}) = 0$, then the $S_1$-$\ta$ can be mapped to the right or to the middle of the $S_1$-$\ta$-block, if $\pi(l_{j, 1}) = 0$ and $\pi(l_{j, 2}) = 1$, then the $S_1$-$\ta$ can be mapped to the left or to the middle of the $S_1$-$\ta$-block, and if $\pi(l_{j, 1}) = \pi(l_{j, 2}) = 1$, then the $S_1$-$\ta$ can be mapped to the left or to the right or to the middle of the $S_1$-$\ta$-block. In summary: if $\pi(l_{j, 1}) \vee \pi(l_{j, 2}) = 0$, then $S_1$-$\ta$ must be mapped to the middle of the $S_1$-$\ta$-block, and if $\pi(l_{j, 1}) \vee \pi(l_{j, 2}) = 1$, then it is possible that $S_1$-$\ta$ is \emph{not} mapped to the middle of the $S_1$-$\ta$-block (whether it is mapped to the left or right depends on the actual values of $\pi(l_{j, 1})$ and $\pi(l_{j, 2})$).

    We synchronise $S_1$-$\ta$ and $S_2$-$\ta$ such that $S_2$-$\ta$ can be mapped to the left (right) of the $S_2$-$\ta$-block if and only if $S_1$-$\ta$ is mapped to the left (right) of the $S_1$-$\ta$-block and $S_2$-$\ta$ is mapped to the middle of the $S_2$-$\ta$-block if $S_1$-$\ta$ is mapped to the middle of the $S_1$-$\ta$-block. This can be done by an $[8,9]$-length constraint between $S_1$-$\ta$ and $S_2$-$\ta$, which enforces the following: \begin{itemize}[nosep]
        \item If $S_1$-$\ta$ is mapped to the left of the $S_1$-$\ta$-block, then $S_2$-$\ta$ can only be mapped to the left of the $S_2$-$\ta$-block,
        \item if $S_1$-$\ta$ is mapped to the middle of the $S_1$-$\ta$-block, then $S_2$-$\ta$ can only be mapped to the middle of the $S_2$-$\ta$-block, and 
        \item if $S_1$-$\ta$ is mapped to the right of the $S_1$-$\ta$-block, then $S_2$-$\ta$ can be mapped to the middle or the right of the $S_2$-$\ta$-block. 
    \end{itemize}  
        In particular, this means that $\pi(l_{j,1})\vee\pi(l_{j,2})=0$ implies that $S_2$-$\ta$ is necessarily mapped to the middle of the $S_2$-$\ta$-block. 

    In order to map the three values of the $S_2$-$\ta$-block to two positions, we need to permute its values. To this end, we add an $L$-constraint between $S_2$-$\ta$ and $S_2$-$\tb$: \begin{itemize}[nosep]
        \item If $S_2$-$\ta$ is mapped to the left of the $S_2$-$\ta$-block, $S_2$-$\tb$ can only be mapped to the left of the $S_2$-$\tb$-block, 
        \item if $S_2$-$\ta$ is mapped to the middle of the $S_2$-$\ta$-block, $S_2$-$\tb$ can only be mapped to the right of the $S_2$-$\tb$-block, and
        \item if $S_2$-$\ta$ is mapped to the right of the $S_2$-$\ta$-block, $S_2$-$\tb$ can only be mapped to the middle of the $S_2$-$\tb$-block.
    \end{itemize}
        Thus, $\pi(l_{j,1})\vee\pi(l_{j,2})=0$ implies that $S_2$-$\tb$ has to be mapped to the right of the $S_2$-$\tb$-block. 
        Using an additional $[3,7]$-length constraint between the $S_2$-$\tb$ and the $L_{\vee}$-$\ta$ (note that we utilise gap-constraints to the first $\tc$ to guarantee that the $L_{\vee}$-$\ta$ is mapped to the $L_{\vee}$-$\ta$-block), we can therefore enforce that if $S_2$-$\tb$ is mapped to the right of the $S_2$-$\tb$-block, then the $L_{\vee}$-$\ta$ can only be mapped to the right of the $L_{\vee}$-$\ta$-block. Otherwise, if $S_2$-$\tb$ is mapped to the left or the middle of the $S_2$-$\tb$-block, then the $L_{\vee}$-$\ta$ can be mapped to either position of the $L_{\vee}$-$\ta$-block. In particular, if $\pi(l_{j,1})\vee\pi(l_{j,2})=0$ then the $L_{\vee}$-$\ta$ can only be mapped to the right of the $L_{\vee}$-$\ta$-block and if $\pi(l_{j,1})\vee\pi(l_{j,2})=1$ then there is an embedding such that the $L_{\vee}$-$\ta$ is mapped to the left of the $L_{\vee}$-$\ta$-block. 

    Finally, we add a $[1,2]$-length constraint between the $L_{\vee}$-$\ta$ and the $L_3$-$\ta$. This means that if $L_{\vee}$-$\ta$ is mapped to the right of the $L_{\vee}$-$\ta$-block (corresponding to value $0$), then the $L_3$-$\ta$ must be mapped to $1$ (i.e., to the right of the $L_3$-$\ta$-block), and if the $L_{\vee}$-$\ta$ is mapped to the left of the $L_{\vee}$-$\ta$-block (corresponding to $1$), then the $L_3$-$\ta$ can be mapped to $0$ or $1$.  

    This concludes the definition of the gap-constraints and therefore concludes the definition of the instance of the matching problem.

    We now assume that there is a satisfying assignment $\pi$, i.e., $\pi(l_{j, 1}) \vee \pi(l_{j, 2}) \vee \pi(l_{j, 3}) = 1$ for every $j \in \{1, 2, \ldots, m\}$, and we will show that this implies that $u$ can be embedded into $\doc$ such that all gap-constraints are satisfied. First, we embed each $\widehat{u}_i$ into $\widehat{\doc}_i$ as determined by $\pi$, i.e., if $\pi(x_i) = 1$, we use the first embedding mentioned above, and if $\pi(x_i) = 0$, we use the second embedding mentioned above. This describes an embedding of $\widehat{u}$ in $\widehat{\doc}$.

    For embedding $\widetilde{u}$ into $\widetilde{\doc}$, we explain how we can embed $\widetilde{u}_j$ into $\widetilde{\doc}_j$. For every $\ell \in \{1, 2, 3\}$, we map the $L_{\ell}$-$\ta$ to $1$ if $\pi(l_{j, \ell}) = 1$ and to $0$ if $\pi(l_{j, \ell}) = 0$ (this obviously satisfies all the gap-constraints between $\widetilde{u}_j$ and the corresponding assignment gadgets). If $\pi(l_{j, 1}) \vee \pi(l_{j, 2}) = 1$, then, as observed above, we can map $S_1$-$\ta$ to the left or to the right, depending on the actual values of $\pi(l_{j, 1})$ and $\pi(l_{j, 2})$. 
        This means that we can map $S_2$-$\ta$ to the left or to the right, and therefore the $S_2$-$\tb$ can be mapped to the left or the middle. In both cases, mapping the $L_{\vee}$-$\ta$ to $1$ (i.e., the left) satisfies the length constraint between the $S_2$-$\tb$ and the $L_{\vee}$-$\ta$. Furthermore, regardless of where the $L_3$-$\ta$ is mapped to, the $[1,2]$-length constraint between the $L_{\vee}$-$\ta$ and the $L_3$-$\ta$ is satisfied. 
        Consequently, we have an embedding of $\widetilde{u}_j$ into $\widetilde{\doc}_j$ that satisfies all gap-constraints. We have to consider the case that $\pi(l_{j, 1}) \vee \pi(l_{j, 2}) = 0$, which means that $\pi(l_{j, 1}) = 0$ and $\pi(l_{j, 2}) = 0$. In this case, the $S_1$-$\ta$ (and therefore the $S_2$-$\ta$) can only be mapped to the middle, which means that the $S_2$-$\tb$ has to be mapped to the right and further, the $L_{\vee}$-$\ta$ must be mapped to $0$. However, since $\pi(l_{j, 1}) \vee \pi(l_{j, 2}) \vee \pi(l_{j, 3}) = 1$, we know that $\pi(l_{j, 3}) = 1$, which means that the $L_3$-$\ta$ is mapped to $1$, which satisfies the $[1,2]$-length constraint for $L_{\vee}$-$\ta$ and $L_3$-$\ta$. We conclude that also in this case, we have an embedding of $\widetilde{u}_j$ into $\widetilde{\doc}_j$ that satisfies all gap-constraints.

    Assume now that we can embed $u$ into $\doc$ such that all gap-constraints are satisfied, and let $\pi$ be the assignment induced by embedding $\widehat{u}$ into $\widehat{\doc}$. We know that, for every $j \in \{1, 2, \ldots, m\}$, $\widetilde{u}_j$ is embedded into $\widetilde{\doc}_j$ such that all gap-constraints are satisfied. For the sake of a contradiction, let us assume that $\pi(l_{j, 1}) \vee \pi(l_{j, 2}) \vee \pi(l_{j, 3}) = 0$, which means that $\pi(l_{j, 1}) = \pi(l_{j, 2}) = \pi(l_{j, 3}) = 0$. This implies that the $L_1$-$\ta$ and the $L_2$-$\ta$ are both mapped to $0$, which means that $S_1$-$\ta$ (and therefore $S_2$-$\ta$) is mapped to the middle and $S_2$-$\tb$ is mapped to the right, which means that the $L_{\vee}$-$\ta$ is mapped to $0$. Since the $L_3$-$\ta$ is also mapped to $0$, the gap induced by the $L_{\vee}$-$\ta$ and the $L_3$-$\ta$ has length $0$, which violates the $[1,2]$-length constraint for the $L_{\vee}$-$\ta$ and the $L_3$-$\ta$. This is a contradiction to our assumption that $\widetilde{u}_j$ is embedded into $\widetilde{\doc}_j$ such that all gap-constraints are satisfied; thus, $\pi(l_{j, 1}) \vee \pi(l_{j, 2}) \vee \pi(l_{j, 3}) = 0$ is not possible. Hence, $\pi$ is satisfying.

It now only remains to adjust the reduction above to only require a binary alphabet $\Sigma=\{\ta,\tb\}$. We will only discuss the necessary adjustments.

    For every $i\in[n]$, the \emph{$x_i$-assignment gadget} is defined as:

    \begin{tabular}{ccc}
        $\widehat{u}_i =$ & $\ta$ & $\tb$ \\
        $\widehat{\doc}_i =$ &  $\ta \ta$ & $\tb \tb$\\
         & $0 1$ & $01$
    \end{tabular}

    Like defined above, we define $\widehat{u} = \widehat{u}_1 \widehat{u}_2 \cdots \widehat{u}_n$ and $\widehat{\doc} = \widehat{\doc}_1 \widehat{\doc}_2 \cdots \widehat{\doc}_n$, and we note that if $\widehat{u}$ can be embedded into $\widehat{\doc}$, then this means that every $\widehat{u}_i$ is embedded into $\widehat{\doc}_i$, which, as observed above, induces an assignment $\pi \colon \{x_1, x_2, \ldots, x_n\} \to \{0, 1\}$.

    Further, for every $j\in[m]$, the \emph{$c_j$-clause gadget} is defined as:

    \begin{tabular}{ccccccc}
        & $L_1$ & $S_1$ & $L_2$ & $L_{\vee}$ & $L_3$ & $S_2$ \\ 
        $\widetilde{u}_j$ & $\ta$ & $\ta$ & $\ta$ & $\ta$ & $\ta$ & $\ta\ \tb$ \\
        $\widetilde{\doc}_j$ & $\ta\ta$ & $\ta\ta\ta$ & $\ta\ta$ & $\ta\ta$ & $\ta\ta$ & $\ta\tb\ta\ta\tb\tb$ \\
        & $01$ & & $10$ & $10$ & $01$ &
    \end{tabular}

    Note that the assignment (resp., clause) gadgets directly correspond to the assignment (resp., clause) gadgets as used above, but without the occurrences of letter $\tc$. Again, we define $\widetilde{u} = \widetilde{u}_1 \widetilde{u}_2 \cdots \widetilde{u}_m$ and $\widetilde{\doc} = \widetilde{\doc}_1 \widetilde{\doc}_2 \cdots \widetilde{\doc}_m$. Now, let $u = \tb \widehat{u} \: \widetilde{u}\tb$ and $\doc = \tb\widehat{\doc} \: \widetilde{\doc}\tb$. We will add a $[|\doc|-2, |\doc|-2]$-gap-constraint between the first and last letter (i.e., the $\tb$ before $\widehat{u}$ and the $\tb$ after $\widetilde{u}$) of $u$, where $|\doc|= 4n+17m+2$. Thus, we force every satisfying embedding to contain the first and last position of $\doc$, which fixes the position of the first and last letter of $u$. 

    Using these two fixed letters, we can now use length constraints to guarantee that every position of $u$ can only be mapped to its corresponding block in $\doc$. E.g. we add a $[4n-2,4n-1]$-constraint between the first $\tb$ of $u$ and the $\tb$-occurrence of the $x_n$-assignment gadget, which enforces that $\widehat{u}$ has to be embedded in $\widehat{\doc}$. These constraints also enforce that $\widetilde{u}_j$ is embedded in $\widetilde{\doc}_j$ for all $j\in[m]$ and that every position is embedded in its corresponding block. This process works analogously to the length constraints for the occurrences of $\tc$ used in the construction above. The rest of the proof is now identical to the case discussed above for alphabet $\{\ta, \tb, \tc\}$.
\end{proof}

\begin{theorem}\label{leftAndRightConvexHardnessObservation}
The matching problem where every gap-constraint is left-convex or right-convex is NP-complete.
\end{theorem}

\begin{proof}
The language $\{\ta \tb, \ta, \eword\}$ is left-convex (but not right-convex) and $\{\ta \tb, \tb, \eword\}$ is right-convex (but not left-convex), but their intersection $\{\ta \tb, \ta, \eword\} \cap \{\ta \tb, \tb, \eword\} = \{\ta \tb, \eword\}$ is 
the language that we used in Theorem~\ref{abepsLanguageHardnessTheorem} in addition to length constraints. (Note that this language is neither left-convex nor right-convex.) Consequently, we can prove the theorem by the same reduction that we used for the proof of Theorem~\ref{abepsLanguageHardnessTheorem}, but each constraint $(i, j, \{\ta \tb, \eword\})$ is replaced by the two left- respectively right-convex constraints $(i, j, \{\ta \tb, \ta, \eword\})$ and $(i, j, \{\ta \tb, \tb, \eword\})$.
\end{proof}

\begin{theorem}\label{fixedLangHardness}
There is a fixed regular language $L$ such that the matching problem with $L$-constraints is NP-complete.
\end{theorem}

\begin{proof}
We use a similar reduction as for Theorem~\ref{aaepsLanguageHardnessTheorem}. Let $F = \{c_1, c_2, \ldots, c_m\}$ be a $3$-CNF formula, where every $c_j = \{l_{j, 1}, l_{j, 2}, l_{j, 3}\} \subseteq \{x_1, \neg x_1, x_2, \neg x_2, \ldots, x_n, \neg x_n\}$ is a clause with three literals. For convenience, we first define a reduction to the matching problem with different constraint languages. Later on, we will explain why this is indeed a reduction to the matching problem with $L$-constraints for a fixed language $L$.

For every $i \in \{1, 2, \ldots, n\}$, the \emph{$x_i$-assignment gadget} is defined as:

\begin{tabular}{cccccc}
$\widehat{u}_i =$ & $\tb$ & $\ta$ & $\tb \tb \tb \tb$ & $\ta$ & $\tb$ \\
$\widehat{\doc}_i =$ & $\tb$ & $\ta \ta$ & $\tb \tb \tb \tb$ & $\ta \ta$ & $\tb$\\
$ $ & & $0 1$ & & $01$ & 
\end{tabular}

Let $L_0 = \{\tb \tb \tb \tb, \ta \tb \tb \tb \tb \ta\}$. We add an $L_0$-constraint between the two $\ta$-occurrences of $\widehat{u}_i$, which means that, just like in the proof of Theorem~\ref{aaepsLanguageHardnessTheorem}, embedding $\widehat{u}_i$ into $\widehat{\doc}_i$ corresponds to assigning $x_i$ to $1$ and $\neg x_i$ to $0$, or the other way around.

For every $j \in \{1, 2, \ldots, m\}$, the \emph{$c_j$-clause gadget} is defined as:

\begin{tabular}{cccccccccccccc}
 &  & $L_1$ & & $S_1$ & & $L_2$ & & $S_2$ & & $L_{\vee}$ &  & $L_3$ & \\
$\widetilde{u}_j$ = & $\tb$ & $\ta$ & $\tb$ & $\ta$ & $\tb\tb$ & $\ta$ & $\tb$ & $\ta$ & $\tb \tb \tb$ & $\ta$ & $\tb \tb \tb \tb \tb$ & $\ta$ & $\tb$ \\
$\widetilde{\doc}_j$ = & $\tb$ & $\ta \ta$ & $\tb$ & $\ta \ta \ta$ & $\tb\tb$ & $\ta \ta$ & $\tb$ & $\ta \ta \ta$ & $\tb \tb \tb$ & $\ta \ta$  & $\tb \tb \tb \tb \tb$ & $\ta \ta$ & $\tb$ \\
&  & $01$ &  & & & $10$ &  & &  & $10$ & & $01$ & 
\end{tabular}

We use the same terminology of the proof of Theorem~\ref{aaepsLanguageHardnessTheorem} to talk about the $L_1$-$\ta$, $L_2$-$\ta$ and so on, and about the corresponding $\ta$-blocks in $\widetilde{\doc}_j$. Again, we set $\widehat{u} = \widehat{u}_1 \widehat{u}_2 \cdots \widehat{u}_n$, $\widehat{\doc} = \widehat{\doc}_1 \widehat{\doc}_2 \cdots \widehat{\doc}_n$, $\widetilde{u} = \widetilde{u}_1 \widetilde{u}_2 \cdots \widetilde{u}_m$, $\widetilde{\doc} = \widetilde{\doc}_1 \widetilde{\doc}_2 \cdots \widetilde{\doc}_m$, but we define $u = \widehat{u} \: \tb^7 \: \widetilde{u}$ and $\doc = \widehat{\doc} \: \tb^7 \: \widetilde{\doc}$. 

We observe that $|\widehat{\doc}_i| = 10$ and $|\widetilde{\doc}_j| = 28$ for every $i \in \{1, 2, \ldots, n\}$ and $j \in \{1, 2, \ldots, m\}$. Note that for every $i \in \{1, 2, \ldots, n\}$, $j \in \{1, 2, \ldots, m\}$, the factor of $\doc$ strictly between $\widehat{\doc}_i$ and $\widetilde{\doc}_j$ is $\widehat{\doc}_{i + 1} \cdots \widehat{\doc}_n \tb^7 \widetilde{\doc}_1 \cdots \widetilde{\doc}_{j-1}$. This means that the factors of $\doc$ that strictly lie between the first $\ta$-occurrence of $\widehat{\doc}_i$ and the $L_{1}$-$\ta$-block, the $L_{2}$-$\ta$-block, and the $L_{3}$-$\ta$-block of the $\widetilde{\doc}_j$ are
\begin{itemize}[nosep]
\item $\ta \tb^4 \ta \ta \tb \: \widehat{\doc}_{i + 1} \cdots \widehat{\doc}_n \tb^7 \widetilde{\doc}_1 \cdots \widetilde{\doc}_{j-1} \: \tb$,
\item $\ta \tb^4 \ta \ta \tb \: \widehat{\doc}_{i + 1} \cdots \widehat{\doc}_n \tb^7 \widetilde{\doc}_1 \cdots \widetilde{\doc}_{j-1} \: \tb \ta \ta \tb \ta^3 \tb \tb$, and
\item $\ta \tb^4 \ta \ta \tb \: \widehat{\doc}_{i + 1} \cdots \widehat{\doc}_n \tb^7 \widetilde{\doc}_1 \cdots \widetilde{\doc}_{j-1} \: \tb \ta \ta \tb \ta^3 \tb \tb \ta \ta \tb \ta^3 \tb^3 \ta \ta \tb^5$, respectively.
\end{itemize}
For convenience, for every $p, q \in \mathbb{N}\cup\{0\}$, we define $\alpha(p, q, 1) \coloneq 16 + 10p + 28q$, $\alpha(p, q, 2) \coloneq 24 + 10p + 28q$ and $\alpha(p, q, 3) \coloneq 40 + 10p + 28q$. From the considerations from above it follows that, for every $i \in \{1, 2, \ldots, n\}$, $j \in \{1, 2, \ldots, m\}$ and $\ell \in \{1, 2, 3\}$, there are exactly $\alpha(n-i, j-1, \ell)$ symbols that lie strictly between the first $\ta$-occurrence of $\widehat{\doc}_i$ and the $L_{\ell}$-$\ta$-block of $\widetilde{\doc}_j$. Moreover, with $\beta(p, q, \ell) = \alpha(p, q, \ell) - 6$ for every $p, q \in \mathbb{N}\cup\{0\}$ and $\ell \in \{1, 2, 3\}$, we have that there are exactly $\beta(n-i, j-1, \ell)$ symbols that lie strictly between the third $\ta$-occurrence of $\widehat{\doc}_i$ and the $L_{\ell}$-$\ta$-block of $\widetilde{\doc}_j$. In order to verify this, just observe that $\alpha(n-i, j-1, 1)$, $\alpha(n-i, j-1, 2)$, and $\alpha(n-i, j-1, 3)$ are exactly the lengths of the factors displayed above (and an analogous explanation applies to $\beta(n-i, j-1, 1)$, $\beta(n-i, j-1, 2)$, and $\beta(n-i, j-1, 3)$ just with the prefix $\ta \tb^4 \ta \ta \tb$ of the above factors replaced by $\ta \tb$).

We also define $L_{\alpha(p, q, \ell)}$ as the language of strings over $\{\ta, \tb\}$ of size $\alpha(p, q, \ell)$ and $L_{\beta(p, q, \ell)}$ as the language of strings over $\{\ta, \tb\}$ of size $\beta(p, q, \ell)$.

We can now achieve the synchronisation between the assignment gadgets and the clause gadgets in the same way as in the proof of Theorem~\ref{aaepsLanguageHardnessTheorem}, but instead of a length constraint $[\alpha(n-i, j-1, \ell), \alpha(n-i, j-1, \ell)]$ or $[\beta(n-i, j-1, \ell), \beta(n-i, j-1, \ell)]$, we use the constraint language $L_{\alpha(n-i, j-1, \ell)}$ or $L_{\beta(n-i, j-1, \ell)}$, respectively. 

Finally, we use gap-constraints in order to make the clause gadgets work, i.e., we want to enforce the following properties:

\begin{itemize}[nosep]
\item Every $\ta$ of $\widetilde{u}_j$ is mapped to its corresponding $\ta$-block of $\widetilde{\doc}_j$.
\item If $L_1$-$\ta$ is mapped to the left and $L_2$-$\ta$ is mapped to the left, then $S_1$-$\ta$ can be mapped to the left or middle.
\item If $L_1$-$\ta$ is mapped to the right and $L_2$-$\ta$ is mapped to the left, then $S_1$-$\ta$ can be mapped to the left, middle or right.
\item If $L_1$-$\ta$ is mapped to the right and $L_2$-$\ta$ is mapped to the right, then $S_1$-$\ta$ can be mapped to the right or middle.
\item If $L_1$-$\ta$ is mapped to the left and $L_2$-$\ta$ is mapped to the right, then $S_1$-$\ta$ can only be mapped to the middle.
\item $S_1$-$\ta$ is mapped to the left (middle, right) if and only if $S_2$-$\ta$ is mapped to the left (middle, right).
\item If $S_2$-$\ta$ is mapped to the left or to the right, then $L_{\vee}$-$\ta$ is mapped to $1$, and if $S_2$-$\ta$ is mapped to the middle, then $L_{\vee}$-$\ta$ is mapped to $0$.
\item If $L_{\vee}$-$\ta$ is mapped to $1$, then $L_3$-$\ta$ can be mapped to $0$ or to $1$, and if $L_{\vee}$-$\ta$ is mapped to $0$, then $L_3$-$\ta$ must be mapped to $1$.
\end{itemize}

We can achieve this by defining the following gap-constraints:

\begin{itemize}[nosep]
\item We use $L_{C, 1} = \{\tb, \tb \ta, \tb \ta \ta, \ta \tb, %
\ta \tb \ta\}$ as constraint language between $L_1$-$\ta$ and $S_1$-$\ta$.
\item We use $L_{C, 2} = \{\tb\tb, \tb\tb \ta, %
\ta \tb\tb, \ta \ta \tb\tb, \ta \tb\tb \ta\}$ as constraint language between $L_2$-$\ta$ and $S_1$-$\ta$.
\item We use $L_{C, 3} = \{\ta \ta \tb \tb \ta \ta \tb, \ta \tb \tb \ta \ta \tb \ta, \tb \tb \ta \ta \tb \ta \ta\}$ as constraint language between $S_1$-$\ta$ and $S_2$-$\ta$.
\item We use $L_{C, 4} = \{\ta \ta \tb \tb \tb, \tb \tb \tb, \ta \tb \tb \tb \ta\}$ as constraint language between $S_2$-$\ta$ and $L_{\vee}$-$\ta$.
\item We use $L_{C, 5} = \{\tb \tb \tb \tb \tb \ta, \ta \tb \tb \tb \tb \tb, \ta \tb \tb \tb \tb \tb \ta\}$ as constraint language between $L_{\vee}$-$\ta$ and $L_{3}$-$\ta$.
\end{itemize}

In order to see this, it is sufficient to observe that the allowed gap-strings of the languages above are exactly chosen such that the desired properties are enforced. 

The correctness of the reduction follows in the same way as in the proof of Theorem~\ref{aaepsLanguageHardnessTheorem}. 

It remains to show how this reduction can be changed into a reduction that only uses $L$-constraints for a fixed language $L$ (i.e., a language independent of the actual instance of the Boolean satisfiability problem).

We first observe that the languages $L_0$, $L_{C, 1}$, $L_{C, 2}$, $L_{C, 3}$, $L_{C, 4}$ and $L_{C, 5}$ are pairwise disjoint. Moreover, if we use a constraint language $L \in \{L_0, L_{C, 1}, L_{C, 2}, L_{C, 3}, L_{C, 4}, L_{C, 5}\}$ for any two positions $i$ and $j$ of $u$, then $i$ and $j$ will be necessarily embedded such that in between their images $e(i)$ and $e(j)$ in $\doc$ no string from $(L_0 \cup L_{C, 1} \cup L_{C, 2} \cup L_{C, 3} \cup L_{C, 4} \cup L_{C, 5}) \setminus L$ can occur. This is directly implied by the fact that in between $e(i)$ and $e(j)$ there is either exactly one occurrence of $\tb$, or exactly one occurrence of factor $\tb \tb$, or exactly one occurrence of a factor of the form $\tb \tb \ta^+ \tb$, or exactly one occurrence of a factor $\tb \tb\tb$, or exactly one occurrence of a factor $\tb \tb \tb \tb \tb$. Consequently, we can define $K \coloneq L_0 \cup L_{C, 1} \cup L_{C, 2} \cup L_{C, 3} \cup L_{C, 4} \cup L_{C, 5}$ and simply use $K$ as constraint language whenever we have used one of $L_0$, $L_{C, 1}$, $L_{C, 2}$, $L_{C, 3}$, $L_{C, 4}$ or $L_{C, 5}$.

Let $u$ be embedded into $\doc$ by some embedding (not necessarily satisfying any gap-constraints), and let the first $\ta$ of some $\widehat{u}_i$ be mapped to position $p$ of $\doc$ and let $L_{\ell}$-$\ta$ of some $\widetilde{u}_j$ be mapped to position $q$ of $\doc$. Then $(q - p - 1) \in \{\alpha(n-i, j-1, \ell) - 1, \alpha(n-i, j-1, \ell), \alpha(n-i, j-1, \ell) + 1\}$ (note that $q - p - 1$ is the number of symbols strictly between positions $q$ and $p$ of $\doc$, i.e., the number of symbols strictly between the positions where we mapped the mentioned $\ta$-occurrences of $u$). Similarly, if the second $\ta$ of some $\widehat{u}_i$ is mapped to position $p$ of $\doc$ and $L_{\ell}$-$\ta$ of some $\widetilde{u}_j$ is mapped to position $q$ of $\doc$, then $(q - p - 1) \in \{\beta(n-i, j-1, \ell) - 1, \beta(n-i, j-1, \ell), \beta(n-i, j-1, \ell) + 1\}$. This is a consequence of the fact that any embedding that embeds $u$ in $\doc$ must map each $\tb$-occurrence of $u$ to the corresponding $\tb$-occurrence of $\doc$.

For every $n, m \in \mathbb{N}$, $i \in \{1, 2, \ldots, n\}$, $j \in \{1, 2, \ldots, m\}$ and $\ell \in \{1, 2, 3\}$, the numbers $\alpha(n-i, j-1, \ell)$ and $\beta(n-i, j-1, \ell)$ are even, which means that the numbers $\alpha(n-i, j-1, \ell) - 1$, $\alpha(n-i, j-1, \ell) + 1$, $\beta(n-i, j-1, \ell) - 1$, $\beta(n-i, j-1, \ell) + 1$ (which are odd numbers) do not appear in the set $\{\alpha(n-i, j-1, \ell), \beta(n-i, j-1, \ell) \mid n, m \in \mathbb{N}, i \in \{1, 2, \ldots, n\}, j \in \{1, 2, \ldots, m\}, \ell \in \{1, 2, 3\}\}$. Hence, we can define 
\begin{equation*}
L_s \coloneq \bigcup_{n, m \in \mathbb{N}, i \in \{1, \ldots, n\}, j \in \{1, \ldots, m\}, \ell \in \{1, 2, 3\}} (L_{\alpha(n-i, j-1, \ell)} \cup L_{\beta(n-i, j-1, \ell)})
\end{equation*}
and simply replace every constraint language $L_{\alpha(n-i, j-1, \ell)}$ or $L_{\beta(n-i, j-1, \ell)}$ by $L_s$. 
Note that $L_s$ is a regular language.
It is important that $L_s$ caters for every possible choice of $n$ and $m$, since we need a single fixed constraint language that works for any possible SAT-instance (i.e., we can have an arbitrary number $n$ of variables and an arbitrary number $m$ of clauses). This is correct, since in any embedding of $u$ in $\doc$ (where $u$ and $\doc$ result from some arbitrary SAT-instance) the distance between the first $\ta$-occurrence of some $\widehat{\doc}_i$ and the $L_{\ell}$-$\ta$-block of some $\widetilde{\doc}_j$ is $\alpha(n-i, j-1, \ell) - 1$, $\alpha(n-i, j-1, \ell)$ or $\alpha(n-i, j-1, \ell) + 1$, but only $L_{\alpha(n-i, j-1, \ell)}$ is contained in $L_s$, while both $L_{\alpha(n-i, j-1, \ell) - 1}$ and $L_{\alpha(n-i, j-1, \ell) + 1}$ are disjoint from $L_s$ (and analogously for the third $\ta$-occurrence of $\widehat{\doc}_i$ and the $\beta$-numbers). Consequently, for those pairs of positions of $u$ where we use the constraint language $L_{\alpha(n-i, j-1, \ell)}$ (or $L_{\beta(n-i, j-1, \ell)}$), using $L_s$ instead has the exact same effect.

Finally, we note that $L_s \cap K = \emptyset$, since the longest strings of $K$ have length $7$, while the smallest strings of $L_s$ have length $\beta(0, 0, 1) = 10$. Hence, we can replace every constraint language by the constraint language $L \coloneq L_s \cup K$. It is obvious that $L$ is a regular language.
\end{proof}

\end{document}